\documentclass{lmcs}
\pdfoutput=1

\usepackage{lastpage}
\lmcsdoi{18}{1}{25}
\lmcsheading{}{\pageref{LastPage}}{}{}%
{Mar.~04,~2021}{Feb.~01,~2022}{}

\keywords{Regular tree languages, infinite trees, Factorization theory, IO and OI}

\usepackage[utf8]{inputenc}
\usepackage{latexsym}
\usepackage{xspace}
\usepackage{amssymb}
\usepackage{amsmath}
\usepackage{stmaryrd}
\usepackage{microtype}
\usepackage{tikz}

\usetikzlibrary{decorations.pathmorphing,shapes,arrows,automata}

\newcommand{\set}[2]{\{\, {#1}\;\mid\; {#2} \}}

\newcommand\rmax{r_{\text{max}}}

\newcommand{\prebest}{\preceq_{\text{best}}}

\newcommand\lds{,\ldots,}

\newcommand{\sse}{\subseteq}
\newcommand{\ssneq}{\varsubsetneq}
\newcommand{\es}{\emptyset}
\newcommand{\sm}{\setminus}
\newcommand{\os}[1]{\{#1\}}
\newcommand{\wh}[1]{\widehat{#1}}
\newcommand{\wt}[1]{\widetilde{#1}}

\newcommand{\pNTA}{parity-NTA\xspace}

\newcommand{\pATA}{parity-ATA\xspace}

\newcommand{\ScX}{\Sig\cup\cX}

\newcommand{\N}{\mathbb{N}}

\newcommand{\B}{\mathbb{B}}

\newcommand{\PSPACE}{\text{\sc Pspace}}
\newcommand{\EXPSPACE}{\text{\sc Expspace}}
\newcommand{\DEXPTIME}{\text{\sc Dexptime}}

\newcommand{\cC}{\mathcal{C}}
\newcommand{\cD}{\mathcal{D}}

\newcommand{\cL}{\mathcal{L}}

\newcommand{\cP}{\mathcal{P}}

\newcommand{\cT}{\mathcal{T}}

\newcommand{\cX}{\mathcal{X}}

\newcommand{\Sig}{\Sigma}
\newcommand{\Del}{\Delta}
\newcommand{\Gam}{\Gamma}

\newcommand{\OO}{\Omega}

\newcommand{\sig}{\sigma}

\newcommand{\del}{\delta}
\newcommand{\eps}{\varepsilon}
\newcommand{\alp}{\alpha}
\newcommand{\bet}{\beta}
\newcommand{\gam}{\gamma}

\newcommand{\oo}{\omega}

\newcommand{\oi}[1]{{#1}^{-1}}

\newcommand{\tra}{transition\xspace}
\newcommand{\tras}{transitions\xspace}

\newcommand{\Pro}{\mathrm{\cP}}

\newcommand{\abs}[1]{\left|\mathinner{#1}\right|}
\newcommand{\Abs}[1]{\left\Vert\mathinner{#1}\right\Vert}

\newcommand\la{\leftarrow}%{\longleftarrow}
\newcommand\lsa{\leftarrow}

\def\math#1{\ifmmode #1\else \mbox{$#1$} \fi}

\newcommand{\IFF}{if and only if\xspace}

\newcommand{\sr}{\mathop{\boldsymbol{\iota}}} %source of an edge
\newcommand{\tr}{\mathop{\boldsymbol{o}}}  %target of an edge
\SetSymbolFont{stmry}{bold}{U}{stmry}{m}{n}
\renewcommand*{\set}[2]{\left\{#1\, \middle| \,#2 \right\}}

\newcommand{\Nat}{\mathbb{P}}% {\N_+}
\newcommand{\rk}{\operatorname{rk}}

\newcommand*{\Pos}[0]{\operatorname{Pos}}
\newcommand*{\pos}[0]{\operatorname{pos}}

\renewcommand*{\root}[0]{\operatorname{root}}

\newcommand*{\leaf}[0]{\operatorname{leaf}}

\newcommand*{\Run}[2]{\operatorname{Run}(#1,#2)}
\newcommand*{\RUN}{\operatorname{Run}}

\makeatletter
\newcommand*{\TDFTA}{\@ifnextchar{.}{$\downarrow$FTA}{$\downarrow$FTA\@\xspace}}
\newcommand*{\BUFTA}{\@ifnextchar{.}{$\uparrow$FTA}{$\uparrow$FTA\@\xspace}}
\newcommand*{\TDDFTA}{\@ifnextchar{.}{$\downarrow$DFTA}{$\downarrow$DFTA\@\xspace}}
\newcommand*{\BUDFTA}{\@ifnextchar{.}{$\uparrow$DFTA}{$\uparrow$DFTA\@\xspace}}
\newcommand*{\FTA}{\@ifnextchar{.}{FTA}{FTA\@\xspace}}
\newcommand*{\eg}{\@ifnextchar{.}{e.\,g}{e.\,g.\@\xspace}}
\newcommand*{\ie}{\@ifnextchar{.}{i.\,e}{i.\,e.\@\xspace}}
\newcommand*{\wrt}{\@ifnextchar{.}{w.r.t}{w.r.t.\@\xspace}}

\newcommand*{\OBdA}{\@ifnextchar{.}{O.\,B.\,d.\,A}{O.\,B.\,d.\,A.\@\xspace}}
\newcommand*{\oBdA}{\@ifnextchar{.}{o.\,B.\,d.\,A}{o.\,B.\,d.\,A.\@\xspace}}
\newcommand*{\usw}{\@ifnextchar{.}{usw}{usw.\@\xspace}}
\renewcommand*{\dh}{\@ifnextchar{.}{d.\,h}{d.\,h.\@\xspace}}
\newcommand*{\zB}{\@ifnextchar{.}{z.\,B}{z.\,B.\@\xspace}}
\newcommand*{\idR}{\@ifnextchar{.}{i.\,d.\,R}{i.\,d.\,R.\@\xspace}}
\newcommand*{\bzw}{\@ifnextchar{.}{bzw}{bzw.\@\xspace}}
\newcommand*{\s}{\@ifnextchar{.}{s}{s.\@\xspace}}
\newcommand*{\su}{\@ifnextchar{.}{s.\,u}{s.\,u.\@\xspace}}
\newcommand*{\iZ}{\@ifnextchar{.}{i.\,Z}{i.\,Z.\@\xspace}}
\newcommand*{\iW}{\@ifnextchar{.}{i.\,W}{i.\,W.\@\xspace}}
\newcommand*{\ua}{\@ifnextchar{.}{u.\,a}{u.\,a.\@\xspace}}
\newcommand*{\iA}{\@ifnextchar{.}{i.\,A}{i.\,A.\@\xspace}}
\makeatother

\newcommand{\Ip}{In particular,\xspace}

\newcommand{\solu}{solution\xspace}

\renewcommand{\hom}{homomorphism\xspace}
\newcommand\Homs{Homomorphisms\xspace}

\newcommand{\subst}{substitution\xspace}

\newcommand{\Tfin}{\mathop{T_{\text{fin}}}}
\newcommand{\TfinX}{\mathop{T_{\text{$\cX$-fin}}}}
\newcommand{\TfinGam}{\mathop{T_{\text{$\Gam$-fin}}}}
\newcommand{\TfinH}{\mathop{T_{\text{$H$-fin}}}} %gss terms with finitely many holes

\renewcommand{\phi}{\varphi}

\newcommand{\sio}{\sigma_{\mathrm{io}}}
\newcommand{\sext}{\sigma_{\mathrm{e}}}
\newcommand{\mio}{{\mathrm{io}}} %gss
\newcommand{\moi}{{\mathrm{oi}}}
\newcommand{\whsio}{\wh\sigma_{\mathrm{io}}}
\newcommand{\vsio}{\check{\sigma}_{\mathrm{io}}}
\newcommand{\vsig}{\check\sigma}

\newcommand{\ngam}[1]{\gam_{#1}}

\newcommand{\gaminf}{\gam_\infty}
\newcommand{\gaminfs}{\gam_\infty(s)}

\newcommand{\soi}{\sigma_{\mathrm{oi}}}

\newcommand{\pr}[1]{\text{pr}_{#1}}

\usepackage{prettyref}
\newcommand{\prref}[1]{\prettyref{#1}}

\newrefformat{thm}{Thm.~\ref{#1}}
\newrefformat{lem}{Lem.~\ref{#1}}
\newrefformat{def}{Def.~\ref{#1}}
\newrefformat{cor}{Cor.~\ref{#1}}
\newrefformat{prop}{Prop.~\ref{#1}}
\newrefformat{sec}{Sec.~\ref{#1}}
\newrefformat{kap}{Chapt.~\ref{#1}}
\newrefformat{ex}{Ex.~\ref{#1}}
\newrefformat{eq}{Eq.(\ref{#1})}
\newrefformat{rem}{Rem.~\ref{#1}}
\newrefformat{fig}{Fig.~\ref{#1}}
\newrefformat{par}{Par.~\ref{#1}}
\newrefformat{foot}{Footnote~\ref{#1}}

\newcommand{\lref}[1]{\prettyref{#1}}

\begin{document}
\title{Regular matching problems for infinite trees}

\author[C. Camino]{Carlos Camino\rsuper{a}}
\author[V. Diekert]{Volker Diekert\rsuper{a}}
\author[B. Dundua]{Besik Dundua\rsuper{b}}
\author{Mircea Marin\rsuper{c}}
\author[G. S{\'e}nizergues]{G{\'e}raud S{\'e}nizergues\rsuper{d}}

\address{FMI, Universit\"at Stuttgart, Germany}
\email{cfcamino@gmail.com, diekert@fmi.uni-stuttgart.de}
\address{Kutaisi International University and  VIAM, Tbilisi State University}
\email{bdundua@gmail.com}
\address{FMI, West University of Timi\c soara, Romania}
\email{mircea.marin@e-uvt.ro}
\address{LaBRI, Universit{\'e} de Bordeaux, France}
\email{geraud.senizergues@u-bordeaux.fr}

%% the abstract has to PRECEDE the command \maketitle:
%% be sure not to issue the \maketitle command twice!

\begin{abstract}
  \noindent We study the matching problem of regular tree languages, that is, ``$\exists \sigma:\sigma(L)\subseteq R$?'' where $L,R$ are  regular tree languages over the union of finite ranked alphabets $\Sigma$ and $\mathcal{X}$ where $\mathcal{X}$ is an alphabet of variables and $\sigma$ is a substitution such that $\sigma(x)$ is a set of trees in $T(\Sigma\cup H)\setminus H$ for all $x\in \mathcal{X}$. Here, $H$ denotes a set of ``holes'' which are used to define a ``sorted'' concatenation of trees. Conway studied this problem in the special case for languages of finite words in his classical textbook \emph{Regular algebra and finite machines} published in 1971. He showed that if $L$ and $R$ are regular, then the problem ``$\exists \sigma \forall x\in \mathcal{X}: \sigma(x)\neq \emptyset\wedge  \sigma(L)\subseteq R$?'' is decidable. Moreover, there are only finitely many maximal solutions, the maximal solutions are regular substitutions, and they are effectively computable.
We extend Conway's results when $L,R$ are  regular languages of finite and infinite trees, and language substitution is applied inside-out, in the sense of Engelfriet and Schmidt (1977/78). More precisely, we show that if $L\subseteq T(\Sigma\cup\mathcal{X})$ and $R\subseteq T(\Sigma)$ are regular tree languages over finite or infinite trees, then the problem ``$\exists \sigma \forall x\in \mathcal{X}: \sigma(x)\neq \emptyset\wedge  \sigma_{\mathrm{io}}(L)\subseteq R$?'' is decidable. Here, the subscript ``$\mathrm{io}$'' in $\sigma_{\mathrm{io}}(L)$ refers to ``inside-out''. Moreover, there are only finitely many maximal solutions $\sigma$, the maximal solutions are regular substitutions and effectively computable. The corresponding question for the outside-in extension $\sigma_{\mathrm{oi}}$ remains open, even in the restricted setting of finite trees.

In order to establish our results we use alternating tree automata with a parity condition and games.
\end{abstract}

\maketitle

%% start the paper here:
\section*{Preamble}
The additional material in the appendix (\prref{sec:append}) is not needed to understand the results in the main body of the paper.

\section{Introduction}\label{sec:intro}
\subsection{Historical background}\label{sec:hb}
Regular matching problems using generalized sequential machines were studied first by Ginsburg and Hibbard. Their publication~\cite{GinsburgHibbard64}, dating back to 1964, showed that it is decidable whether there is a generalized sequential machine which maps $L$ onto $R$ if $L$ and $R$ are regular languages of finite words. The paper also treats several variants of this problem. For example, the authors notice that the decidability cannot be lifted to context-free languages. Another paper in that area is by Prieur et al.~\cite{PrieurCL97}. It appeared in 1997 and studies the problem whether there exists a sequential bijection from a finitely generated free monoid to a given rational set $R$.
Earlier, in the 1960's Conway studied regular matching problems in the following variant of~\cite{GinsburgHibbard64}:
Let $\cX$, $\Sig$ be finite alphabets and $L\sse (\Sig\cup\cX)^*$, $R\sse \Sig^*$.
A \subst $\sig:\cX \to 2^{\Sig^*}$ is extended to $\sig:\Sig\cup\cX \to 2^{\Sig^*}$ by $\sig(a)=\os a$ for all $a\in \Sig\sm \cX$.
It is called a \emph{\solu} of  the problem ``$L\sse R$?'' if
$\sig(L)\subseteq R$.
In his textbook~\cite[Chapt.~6]{conway1971regular}, Conway developed
a \emph{factorization theory} of formal languages. Thereby he found a nugget in formal language theory: Given as input regular word languages
$L\sse (\Sig\cup\cX)^*$ and $R\sse \Sig^*$, it holds:
\begin{enumerate}
\item
It is decidable whether there is a \subst
	$\sig:\cX \to 2^{\Sig^*}$ such that $\sig(L)\sse R$ and  $\es\neq \sig(x)$ for all $x\in \cX$.
\item
	Define $\sig\leq \sig'$
	by $\sig(x)\sse \sig'(x)$ for all $x\in \cX$. Then every solution is bounded from above by  a maximal \solu; and the number of maximal solutions is finite.
\item
If $\sig$ is maximal, then $\sig(x)$ is regular for all $x\in \cX$; and all maximal \solu{s} are effectively computable.
\end{enumerate}
 The original proof is rather technical and not easy to digest. On the other hand, using the algebraic concept of \emph{recognizing morphisms}, elegant and simple proofs exist.
 Regarding the complexity, it turns out that the problem
``$\exists \sig :  \sig(L)\sse R$?'' is $\PSPACE$-complete if $L$ and $R$ are given by NFAs by~\cite[Lemma 3.2.3]{koz77}.
The apparently similar problem
``$\exists \sig :  \sig(L)=R$?'' is more difficult: Bala showed that it is $\EXPSPACE$-complete~\cite{bala2006complexity}.

Conway also asked whether
the unique maximal \solu of the language equation $L x=xL$ is given by a \subst such that $\sig(x)$ is regular\footnote{There is a unique maximal \solu since the union over all \solu{s} is a \solu.}. This question was answered by Kunc in a highly unexpected way: there is a finite set $L$ such that the unique maximal \solu $\sig(x)$ of $L x=xL$ is co-recursively-enumerable-complete~\cite{kunc2007power}. A recent survey on language equations is in~\cite{KuncOkhotin21}.

\subsection{Conway's result for trees}\label{sec:tr}
The present paper generalizes Conway's result to regular tree languages. We consider finite and infinite trees simultaneously
over a finite ranked alphabet $\Del$. We let
$T(\Del)$ be the set of all trees with labels in $\Del$, and by $\Tfin(\Del)$ we denote its subset of finite trees. More specifically, we consider
finite ranked alphabets $\cX$ of \emph{variables} and $\Sig$ of \emph{function symbols}. In order to define a notion of a concatenation we also need a set of \emph{holes} $H$.  These are symbols of rank zero. For simplicity, throughout the set of holes is chosen as
$H=\os{1\lds\abs H}$. We require $(\Sig\cup\cX)\cap H=\es$.  Trees in $T(\Del)$ are rooted and they can be written as terms $x(s_1\lds s_r)$ where $r = \rk(x)\geq 0$ and  $s_i$ are trees. In particular, all symbols of rank $0$ are trees. Words $a_1\cdots a_n\in \Sig^*$ (with $a_i\in \Sig$) are encoded as terms $a_1(\cdots (a_n(\$))\cdots )$ where the $a_i$ are function symbols of rank $1$ and  the only symbol of rank $0$ is~$\$$ which signifies ``\emph{end-of-string}''.  An infinite word $a_1a_2\ldots\in\Sig^\oo$ is encoded  as  $a_1(a_2(\ldots))$ and no hole appears.
In contrast to the case of classical term rewriting, \subst{s} are applied at inner positions, too. In the word case it is clear what to do. For example, let $w=xyx\in \cX^*$ with $\sig(x)=L_x$ and $\sig(y)=L_y$, then
we obtain $\sig(w) = L_xL_yL_x$. Translated to term notation, we obtain $w= x(y(x(\$)))$, $\sig(z)=\set{u(1)}{u\in L_z}$ for $z\in \os{x,y}$ with the result $\sig(w) = L_xL_yL_x(\$)$.
On the other hand, in the tree case variables of any rank may exist. As a result, variables may appear at inner nodes as well as at leaves.
 Throughout, if $t$ is any tree, then  $\leaf_i(t)$ denotes the set of leaves labeled by the hole $i\in H$, and by $i_j$ we denote elements of $\leaf_i(t)$.

In the following, a \emph{\subst} means a mapping $\sig:  \cX\to 2^{T(\Sig\cup H)\sm H}$ such that for all $x\in \cX$  we have
$\sig(x)\sse T(\Sig\cup \os{1\lds \rk(x)})$.
A \emph{homomorphism} (resp.~\emph{partial \hom}) is a \subst $\sig$ such that $|\sig(x)|=1$ (resp.~$|\sig(x)|\leq 1$)
for all $x\in \cX$. We write $\sig_1\leq \sig_2$ if $\sig_1(x)\sse \sig_2(x) $ for all $x\in \cX$; and we say that $\sig$ is \emph{regular}\footnote{There are several equivalent definitions for regular tree languages, e.g.~see~\cite{tata2007,rab69,MullerSchupp87tcs,MullerSchupp95tcs,tho90handbook}.}
 if
$\sig(x)$ is regular for all $x\in \cX$.

Trees are represented graphically, too.
For example, $g(1,f(g(a,1)))$ and $g(t,f(g(a,t)))$ are represented in \prref{fig:tree1}. We obtain $g(t,f(g(a,t)))$ by replacing the positions labeled by hole~$1$ in $g(1,f(g(a,1)))$ by any rooted tree $t$.
\begin{figure}[t]
\begin{center}
\begin{tikzpicture}[xscale=0.6,yscale=0.3]
		\node (s) at (-3,0) {$t'=$};
		\node (rew) at (-1,0) {$g$};
		\node (r1) at (-2,-1) {$1$};
		\node (r2) at (0,-1) {$f$};
		\node (r21) at (1,-2) {$g$};
		\node (r211) at (0,-3) {$a$};
		\node (r212) at (2,-3) {$1$};
		\draw (rew) to (r1);
		\draw (rew) to (r2);
		\draw (r2) to (r21);
		\draw (r21) to (r211);
		\draw (r21) to (r212);

		 \node (st) at (5,0) {$t'[1_j\lsa t]= $};

		\node (ew) at (8,0) {$g$};

		\draw[fill=gray!10] (7,-1) -- (6.2,-2.5) -- (7.8,-2.5) -- (7,-1);
		\node (1) at (6.8,-1.2) {};
\node (1t) at (7,-1.8) {$t$};
		\node (2) at (9,-1) {$f$};
		\node (21) at (10,-2) {$g$};
		\node (211) at (9,-3) {$a$};

		\draw[fill=gray!10] (11,-3) -- (10.3,-4.5) -- (11.7,-4.5) -- (11,-3);
		\node (212) at (11.2,-3.3) {};
         \node (212t) at (11,-3.8) {$t$};
		\draw (ew) to (1);
		\draw (ew) to (2);
		\draw (2) to (21);
		\draw (21) to (211);
		\draw (21) to (212);
\end{tikzpicture}
\end{center}
\caption{The left tree $t'$ has two positions labeled with hole $1\in H$. Holes define a composition of trees. The right tree is obtained by composing $t'$ with a tree $t$ over the hole $1$ which is denoted by  $t'[1_j\lsa t]$.}%
\label{fig:tree1}
\end{figure}
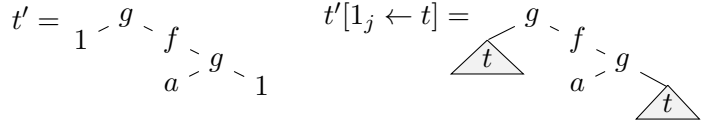
As soon as $\sig$ is not a partial \hom, one has to distinguish
between ``Inside-Out'' (IO for short)  and ``Outside-In''
(OI for short) as advocated and defined in~\cite{EngelfrietS77,EngelfrietS78}. Our positive decidability results concern IO, only. The IO-definition $\sio(s)$  for a given tree
$s\in T(\ScX)$ and a \subst $\sig$ has the following interpretation.
First, we extend $\sig$ to a mapping from
$\Sig \cup \cX \cup H$ to $2^{T(\Sig \cup  H)}$ by $\sig(f) =\{ f(1\lds \rk(f))\}$ for $f\in \Sig\sm \cX$. Second, we use a term
notation $s=x(s_1\lds s_r)$ (for finite and infinite trees $s$) and we let  $\sio(s) \sse T(\Sig)$ to be a certain
fixed point of the language equation
\begin{align}\label{eq:introsio}
\sio(s)=\bigcup \set{t[i_j\lsa t_i]}{t\in \sig(x) \wedge \forall\, 1\leq i \leq r: \;t_i\in \sio(s_i)}.
\end{align}
The notation $t[i_j\lsa t_i]$ in \prref{eq:introsio} means that
for all $i\in H$ and $i_j\in \leaf_i(t)$ all leaves $i_j$ are replaced by $t_{i}$ where $t_i$ depends on $i$, only.
It is a special case of the notation $t[i_j\lsa t_{i_j}]$ which says that
for all $i\in H$ and $i_j\in \leaf_i(t)$ some
tree $t_{i_j}$ is chosen and then each leaf $i_j$ is replaced
by $t_{i_j}$.
This more general notation appears below when defining ``Outside-In'' \subst{s}.
The computation of the elements in $\sio(s)$ follows a recursive procedure which at each call first selects the tree $t\in \sig(x)$ (if $x$ is the label of the root of $s$) and then it makes recursive  calls for each hole $i$ which appears as a label in $t$ to compute the elements $t_i$.
After that, all positions $i_j$ in $t$ which have the label $i$ are replaced
by the same tree $t_i$.
For finite trees the procedure terminates.
For example, $g(t,f(g(a,t)))= g(1,f(g(a,1)))[1_j\lsa t]$ in \prref{fig:tree1} represents an instance of an IO-\subst. For that, we let
$s=g(x,f(g(a,x)))$ where $x$ is a variable of rank zero and
$\sig(x)=\os{t_1,t_2}$. Then $\sio(s) = \os{g(t_1,f(g(a,t_1))),g(t_2,f(g(a,t_2)))}$.
For infinite trees, the recursion is not guaranteed to stop. Running it for $n$ recursive calls  defines a Cauchy sequence in an appropriate complete metric space $T_\bot(\Sig \cup \cX \cup H)$ where $\bot$ plays the role of undefined if the procedure cannot select a tree because the corresponding set $\sig(x)$ happens to be empty.

As explained above, the feature of IO is that every leaf~$i_j$ in $t\in \sig(x)$ labeled with a hole~$i$ is substituted with the same tree $t_i\in \sio(s_i)$. If we remove this restriction (that is: in \prref{eq:introsio} we replace $t[i_j\lsa t_i]$ by $t[i_j\lsa t_{i_j}]$), then we obtain the OI-\subst $\soi(s)$ which can be much larger than $\sio(s)$. For example,
 for $s=g(x,f(g(a,x)))$ and
$\sig(x)=\os{t_1,t_2}$ we obtain
\begin{align*}
\soi(s) =\set{g(t_i,f(g(a,t_j)))}{i,j \in \os{1,2}}.
\end{align*}
Clearly, $\sio(s)=\soi(s)$ if $\sig$ is a partial \hom.
Thus, for a partial \hom~$\phi$ we may write $\phi(s)$ without risking ambiguity.
Another situation for $\sio(s)=\soi(s)$ is when
duplications of holes do not appear. \emph{Duplication} means
that there are some $x$ and $t\in \sig(x)$ where a hole $i\in H$ appears at least twice in $t$. Duplications cannot appear in the traditional framework of words, but  ``duplication'' is a natural concept for trees.

Allowing duplications complicates the situation because it might happen that $\sio(L)$ is not regular, although $L$ and $\sig$ are regular. Actually, this may happen {even if} $\sig$ is defined by a homomorphism $\phi:  \cX\to \Tfin(\Sig\cup H)$. Recall that the notation $\Tfin(\Sig\cup H)$ refers to the subset of finite trees in $T(\Sig\cup H)$.
Examples of a \hom $\phi$ where $\phi(L)$ is not regular are easy to construct. The classical example is
$L= \set{x^n(\$)}{n\in \N}$ and $\phi(x)=f(1,1)$. Then
$\phi(L)$ is not regular.
The corresponding trees of height $3$ are depicted in \lref{fig:gunnar}.
\begin{figure}[h]
	\begin{center}
		\begin{tikzpicture}[scale=0.8]
		\node(s) at (-4,0) {$s_3=$};
				\node (ew) at (-3,0) {$x$};
				\node (1) at (-3,-1) {$x$};
				\node (2) at (-3,-2) {$x$};
				\node (3) at (-3,-3) {$\$$};

				\draw (ew) to (1);
				\draw (1) to (2);
				\draw (2) to (3);
				\node(sigxy) at (-1.2,0) {$\phi(x)= $};
				\node (rew) at (0,0) {$f$};
				\node (r1) at (-0.5,-1) {$1$};
				\node (r2) at (0.5,-1) {$1$};

				\draw (rew) to (r1);
				\draw (rew) to (r2);

			\node(sigs) at (2,0) {$\phi(s_3)= $};
				\node (rew) at (5,0) {$f$};
				\node (r1) at (3,-1) {$f$};
				\node (r2) at (7,-1) {$f$};
				\node (r21) at (6,-2) {$f$};
				\node (r22) at (8,-2) {$f$};
				\node (r11) at (2,-2) {$f$};
				\node (r12) at (4,-2) {$f$};
				\node (r111) at (1.5,-3) {$\$$};
				\node (r112) at (2.5,-3) {$\$$};
				\node (r121) at (3.5,-3) {$\$$};
				\node (r122) at (4.5,-3) {$\$$};
				\node (r211) at (5.5,-3) {$\$$};
				\node (r212) at (6.5,-3) {$\$$};
				\node (r221) at (7.5,-3) {$\$$};
				\node (r222) at (8.5,-3) {$\$$};
				\draw (rew) to (r1);
				\draw (rew) to (r2);
				\draw (r1) to (r11);
				\draw (r1) to (r12);
				\draw (r2) to (r21);
				\draw (r2) to (r22);

				\draw (r11) to (r111);
				\draw (r11) to (r112);
				\draw (r12) to (r121);
				\draw (r12) to (r122);
				\draw (r21) to (r211);
				\draw (r21) to (r212);
				\draw (r22) to (r221);
				\draw (r22) to (r222);
				\end{tikzpicture}
            \end{center}
	\vspace{0.5cm}
	\caption{$L=x^*(\$)$  is regular, but $\phi(L)$ is not.}\label{fig:gunnar}
	\end{figure}
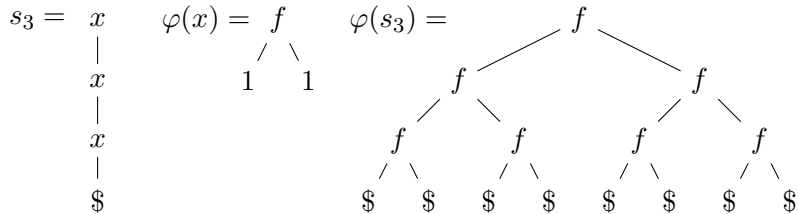
This led to the \emph{\text{HOM}-problem}. The inputs are a homomorphism~$\phi$ and a regular tree language $L$. The question is whether $\phi(L)$ is regular. The problem is decidable {in the setting of finite trees by~\cite{GodoyG13JACM}.} It is $\DEXPTIME$-complete by~\cite{CreusGGR16SIAMCOMP}.

Most of our work deals with regular tree languages. However,
once the results are established for regular sets, we can easily push the results further to some border of decidability.
We consider a class $\cC$ of tree languages such that on input
$L\in \cC$ and a regular tree language $K$, the emptiness problem
$L\cap K$ is decidable. For example, in the word case, the class $\cC$ can be defined by the class of context-free languages, and then Conway's result for finite words still holds if $L$ is context-free and $R$ is regular.

\subsection{Statement of the main results}%: \prref{thm:s1s2main}}%
\label{sec:ms}
Our main results can be found in \prref{sec:mainres}.
We are ready to formulate them now. \prref{thm:s1s2main} can be rephrased as follows.
First, the following decision problem is decidable
\begin{itemize}
\item Input: Regular \subst{s} $\sig_1,\sig_2$ and tree languages $L\sse T(\Sig\cup \cX)$,
$R\sse T(\Sig)$ such that $R$ is regular and $L\in \cC$. Here,
$\cC$ is as defined above such that emptiness of the intersection
with regular sets is decidable.
\item Question: Is there some \subst $\sig: \cX\to 2^{T(\Sig\cup H)\sm H}$ satisfying both, $\sio(L) \sse R$  and  $\sig_1\leq \sig\leq \sig_2$?
\end{itemize}
Second, we can effectively compute the set of maximal \subst{s} $\sig$ satisfying $\sio(L) \sse R$  and  $\sig_1\leq \sig \leq \sig_2$. It is a finite set of regular \subst{s}.
A special case is stated in \prref{cor:main}: we can compute the finite set of maximal \subst{s} $\sig$ satisfying $\sio(L) \sse R$  and  $\sig(x)\neq \es$. This statement reflects the original setting of Conway.

\subsection{Roadmap to prove Theorem~\ref{thm:s1s2main}}\label{sec:road}
Our proof is designed to be accessible for readers who are familiar with basic results about metric spaces and regular languages over trees.\footnote{The simpler case of finite and infinite words can be found in \prref{sec:oldconwords} as a kind of warm-up exercise.}

We don't rely on (or use) the theory of monads. This is a categorial concept, which leads to a general notion of a \emph{syntactic algebra}, see the arXiv-paper of Boja{\'n}czyk~\cite{Bojanczyk15arxiv} for finite trees.\footnote{Conway's result for finite trees follows from the theory of monads, too. Personal communication Miko{\l}aj Boja{\'n}czyk, 2019. For the notion of a monad see~\cite[Chapters III to V]{maclane71}.} For infinite trees, the notion of a \emph{syntactic algebra} was given by Blumensath~\cite{Blumensath20lmcs}.
In our paper we use nondeterministic finite (top-down) parity-tree automata to define an appropriate \emph{congruence} of finite index. This is conceptually simple but there is no free lunch: the index of the congruence defined by an automaton is not guaranteed to be the smallest one.\footnote{With respect to worst case complexity this is more an advantage than a problem. ``Generically'' it is debatable whether it makes sense to spend any efforts in computing syntactic congruences. It doesn't.}

Having a convenient notion of a congruence, the next step is to define $\sio(s)$ for finite and infinite trees such that  $\sio(s)$ is indeed the intended fixed point for
\prref{eq:introsio}. For finite trees the set  $\sio(s)$ can be defined by
induction on the size of $s$. Then $\sio(s)$ becomes the unique least fixed point of
\prref{eq:introsio} satisfying $\sio(x)=\sig(x)$ for all symbols of rank zero.
For infinite trees the definition is more subtle.  Given a tree $s$ and a \subst $\sig$ we introduce a notion of a \emph{choice function}. That is a function $\gamma: \Pos(s) \to T(\Sigma\cup H)\cup \{\bot\}$ where $\Pos(s)$ is the set of positions (that is: vertices) of~$s$ such that the following holds. If $u\in \Pos(s)$ is labeled by~$x$, then  $\gam$ selects a tree $\gam(u)\in \sig(x)$. If $\sig(x)=\es$, then $\gam(u)$ is not defined, which is denoted as $\gam(u)=\bot$.
To each choice function we will associate a Cauchy sequence $\gamma_n(s)$ in some complete metric space $T_\bot(\Sig\cup \cX\cup H)$; and we let $\gamma_\infty(s)=\lim_{n\to \infty}\gamma_n(s)$ be its limit. We can think of the space
$T_\bot(\Sig\cup \cX\cup H)$ as a union of the usual Cantor space
$T(\Sig\cup \cX\cup H)$ together with an isolated point $\bot$ which has distance $1$ to every other point.
Then we define
\begin{equation}\label{eq:firstchoice}
\sio(s) = \set{\gamma_\infty (s)}{\text{ $\gamma$ is a choice function for $s$ and $\gamma_\infty (s)\neq \bot$}}.
\end{equation}
It turns out that this definition  coincides with the natural definition for finite trees, and it satisfies
\prref{eq:introsio}, too.
Another crucial step on the road to show \prref{thm:s1s2main} is
the following result. If $\sig$ and $R$ are regular, then the ``inverse image''
$\oi \sio(R)  = \set{s\in T(\Sig \cup \cX)}{\sio(s) \sse R}$ is a regular set of trees. In order to prove this fact
we use two well-known results. First,
the class of regular tree languages can be characterized by
alternating parity-automata, and the semantics of these automata can be defined by parity-games~\cite{MullerSchupp87tcs,MullerSchupp95tcs,tata2007}.
Second, parity-games are determined and have positional (= memoryless) winning strategies,~\cite{GurevichH82stoc}.

\section{Notation and preliminaries}\label{sec:not}
We let $\N= \os {0,1, \ldots}$ denote the set of natural numbers, $\Nat=\N\sm {\os 0}$, and $\Nat^*$ to be the monoid of finite sequences of positive integers with the operation ``$.$'' and the neutral element~$\epsilon$. % chktex 40
For $r\in\N$ we let $[r]=\os{1\lds r}$. We write $2^S$ for the power set of $S$, and identify every element $x\in S$ with the singleton $\{x\}$.

A \emph{rooted tree} is a nonempty, connected, and directed graph $t=(V,E)$ with vertex set $V$ and without multiple edges such that there is exactly one vertex, the \emph{root}, without any incoming edge and all other vertices have exactly one incoming edge. As a consequence, for every vertex
$v\in V$ there is exactly one directed path from the root to $v$.
Since there are no multiple edges we assume without restriction $E\sse \set{(u,v)}{u,v\in V, u\neq v}$.
If $(u,v)\in E$ is an edge, then we say that $v$ is a \emph{child} of $u$, and $u$ is the \emph{parent} of $v$.
A \emph{leaf} of $t$ is a vertex without children.
Throughout, we restrict ourselves to directed graphs where the set of edges is (at most) countable. Hence, it is possible to encode the vertex set of a rooted tree as
a subset of \emph{positions} $\Pos(t)\sse \Nat^*$ satisfying the following conditions: $\eps\in \Pos(t)$, and
if $u.j\in \Pos(t)$, then both $u\in \Pos(t)$ and $u.i\in \Pos(t)$
for all $1\leq i \leq j$.
Using this, we have $\root(t)=\eps$ and edge set $\set{(u,u.i)}{u,u.i\in \Pos(t)}$. We are mainly interested in
\emph{ordered} trees: these are
rooted trees where the children of a node are equipped with a  ``left-to-right'' ordering. In such a case the ``left-to-right breadth-first'' ordering on positions can be represented by the length-lexicographical ordering on $\Nat^*$.
The \emph{size} of a tree $t$ is the cardinality of $\Pos(t)$.
The \emph{level} of a vertex $u$ is the length of the unique directed path starting at the root and ending in~$u$. Dually, the
\emph{height} of a vertex~$u$ is the length of the longest directed path starting at~$u$. Leaves have height~zero.
The \emph{height of a tree} is the height of its root.

Typically, and actually without restriction, every position in a tree has a label in some set $\Delta$. Such a tree $t$ can be defined therefore
through a mapping $t:\Pos(t)\to \Delta$ where $\Pos(t)$ is the set of positions and $t(u)\in \Delta$ is the label of the position $u$.
The set of all trees with labels in $\Delta$ is denoted by $T(\Delta)$. Its subset of finite trees is denoted by $\Tfin(\Delta)$. More precisely, if $\Gam\sse \Del$, then
$\TfinGam(\Delta)$ is the set of trees where the number of positions with a label in $\Gam$ is finite.

Let $t,t'\in T(\Delta)$ and $u\in \Pos(t)$.
We define the trees $t|_u$ and $t[u\lsa t']$ as usual:
\begin{align*}
 \Pos(t|_u)&= \set{v\in \Nat^*}{u.v\in \Pos(t)} \text{ with labeling }
 t|_u(v)=t(u.v),\\
\Pos(t[u\lsa t'])&= \set{u.u'}{u'\in \Pos(t')}\cup
\set{v \in \Pos(t)}{u\text{ is not a prefix of }v},\\
t[u\lsa t'](u.u') &= t'(u') \text{ and } t[u\lsa t'](v) =t(v)\text{ if $u$ is not a prefix of }v.
\end{align*}
In almost all our cases (there is one exception in the proof of \prref{prop:equivtasl}) the degree of vertices in a tree is bounded by a constant depending on $\Del$, and vertices have finite degree depending on their label.
A \emph{ranked alphabet} is a nonempty finite set of labels $\Del$ with a \emph{rank} function $\rk:\Del \to \N$. Trees over a ranked alphabet have to
satisfy the following additional constraint:
\begin{itemize}
\item $\forall u\in \Pos(t)$, if $t(u)=x$, then $\os{1\lds \rk(x)}=\{i\in\N\mid u.i\in \Pos(t)\}$.
\end{itemize}
Trees
in  $T(\Del)$ are represented by the set of \emph{terms}, too.
These are the ordered trees $t\in T(\Del)$. We adopt the standard notation of terms to denote trees: $x(s_1,\ldots,s_r)$ represents the tree $s$ with $s(\eps)=x$ and $s|_i=s_i$ for all $i\in [\rk(x)]$.
Henceforth, if not otherwise specified, we let $\OO= \cX\cup \Sig\cup H $ be a ranked alphabet consisting of three sets: $\cX$ is the set of \emph{variables}, $\Sigma$ is the set of \emph{function symbols}, $H$ is the set of \emph{holes}. We require
$(\cX\cup \Sig) \cap H = \es$ but $\cX\cap \Sig\neq \es$ is not forbidden.
Symbols $a\in\Sig$ with $\rk(a)=0$ are called \emph{constants}.
Holes are not constants, but they have rank~$0$, too.
For simplicity, we assume  $H=[\rmax]$ where $\rmax\in\N$ satisfies $\rk(x)\leq \rmax$ for all $x\in\Sigma\cup\cX$.
To make the theory nontrivial, we assume that there is some $x\in \Sig\cup \cX$ with $\rk(x)\geq 1$. \Ip $T(\OO)$ contains finite trees where some leaves are labeled by a hole. For $\rk(x)\geq 2$ there are also infinite trees with this property.

Given a tree $t\in T(\OO)$  we denote by
 $\leaf_i(t)$ the set of leaves which are labeled by the hole  $i\in H$.
The length-lexicographical
 ordering of positions induced by $\Nat^*$ is a well-order, which has the type of either a finite ordinal or
 the ordinal $\omega$. This well-order induces a linear order on $\leaf_i(t)$,  which can be represented by a downward-closed subset of $\Nat$. That is, we may write  $\leaf_i(t)=\set{i_j}{1\leq j \leq |\leaf_i(t)|}$. If $|\leaf_i(t)|=\infty$, then this means $\leaf_i(t)=\set{i_j}{j \in\Nat}$.
The term notation is also convenient for a concise notation of infinite trees using fixed point equations. For example, let $f\in \Sig$ be a function symbol of rank $2$, then there is exactly one tree $t\in T(\Sig \cup \os 1)$ which satisfies the
equation $t=f(t,1)$.  It is depicted in
 \prref{fig:infcomb}. The set $\leaf_1(t)$ is the set of all leaves.
 There is no leftmost leaf, but a rightmost leaf which is in turn the first one in the length-lexicographical
 ordering.

 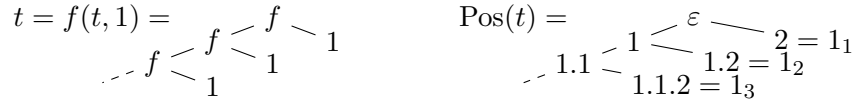
\begin{figure}[h]
\begin{center}
		\begin{tikzpicture}[xscale=0.8,yscale=0.3]
	\node (lrew) at (0,0) {$f$};
	\node (t) at (-3,0) {$t=f(t,1)= $};
				\node (lr1) at (-1,-1) {$f$};
				\node (lr2) at (1,-1) {$1$};
				\draw (lrew) to (lr1);
				\draw (lrew) to (lr2);
				\node (lr11) at (-2,-2) {$f$};
				\node (lr12) at (0,-2) {$1$};
				\draw (lr1) to (lr11);
				\draw (lr1) to (lr12);
				\node (lr111) at (-3,-3) {};
				\node (lr112) at (-1,-3) {$1$};
				\draw [dashed] (lr11) to (lr111);
				\draw  (lr11) to (lr112);
				%%%%%%%% RIGHT POS
			\node (post) at (4,0) {$\Pos(t)= $};
				\node (rew) at (7,0) {$\eps$};
				\node (r1) at (6,-1) {$1$};
				\node (r2) at (9,-1) {$2=1_1$};
				\draw (rew) to (r1);
				\draw (rew) to (r2);
				\node (r11) at (5,-2) {$1.1$};
				\node (r12) at (8,-2) {$1.2= 1_2$};
				\draw (r1) to (r11);
				\draw (r1) to (r12);
				\node (r111) at (4,-3) {};
				\node (r112) at (7,-3) {$1.1.2=1_3$};
				\draw [dashed] (r11) to (r111);
				\draw  (r11) to (r112);

		\end{tikzpicture}
    \end{center}
	\caption{The tree (representing an ``infinite comb'') $t=f(t,1)$ has infinitely many holes labeled by~$1$, but no leftmost hole. The set $\Pos(t)\sse \Nat^*$ is depicted on the right. Following our convention, the  subset $\leaf_1(t)\sse \Pos(t)$ is written as $\{1_1, 1_2, 1_3, \ldots\}$.
	}\label{fig:infcomb}
\end{figure}

 Suppose  that sets $T_x\sse T(\Sig\cup [r])$ and
 $T_{i}\sse T(\Sig)$ for $i\in [r]$  are defined.
 Then we define the set
 $T_x[i_j\lsa T_{i_j}]\sse T(\Sig)$ as the union
 over all trees  $t_x[i_j\lsa t_{i_j}]$ where $t_x\in T_x$ and
 $t_{i_j}\in T_{i_j}$. This is explained in more detail in \prref{sec:substi}.

\subsection{Substitutions: outside-in and inside-out for trees in \texorpdfstring{$\TfinX(\Sig \cup \cX)$}{T\_X-fin(Sigma union X)}}\label{sec:substi}
The ranked alphabet $\ScX$ contains function symbols and variables. As mentioned  in the introduction, we allow
$\Sig \cap \cX \neq \es$.
In the following, if $\sig:\cX\to 2^{T(\Sig \cup H)}$ is a mapping which is specified on the set of variables, then we
extend it to a mapping $\sig:\ScX\to 2^{T(\Sig \cup H)}$
by letting
\begin{align}\label{eq:extsig}
\sig(f)= \os{f(1\lds \rk(f))} &\text{\; for all $f \in H\cup \Sig\sm \cX$}.
\end{align}
We say that a mapping $\sig:\cX\to 2^{T(\Sig \cup H)\sm H}$ is
a \emph{substitution} if $\sig$ satisfies the following additional property
\begin{align}\label{eq:defsig}
\sig(x)\subseteq T(\Sig\cup[\rk(x)])
\setminus H &\text{\; for all $x \in \cX$}.
\end{align}
Note that (\ref{eq:extsig}) and (\ref{eq:defsig}) together imply that $t\in T(\Sig\cup[\rk(x)])\setminus H$ for all $t\in \sig(x)$ and for all $x\in \ScX$. For all elements $x\in \ScX$ of rank zero we have $\sig(x)\sse T(\Sig)$.
The set of \subst{s} is a partial order by letting
 $\sig\leq \sig'$ if $\sig(x)\sse \sig'(x)$ for all $x\in \cX$. A \subst $\sigma$ is called a \emph{homomorphism} (resp.~\emph{partial homomorphism}) if $|\sig(x)|=1$ (resp.~$|\sig(x)|\leq 1$) for all $x\in \cX$. Since we identify elements and singletons
 we can also say that a \hom\footnote{\Homs $\sig$ such that $\sig(x)\notin H$ for all $x \in \cX$ are called \emph{non-erasing} by Courcelle in~\cite[page~117]{courcelle83-tcs}. Thus, there is some difference in notation between our paper and~\cite{courcelle83-tcs}.} is specified by a mapping
$\sig:\cX\to T(\Sig \cup H)\sm H$.
Recall that $t[i_j\lsa t_{i_j}]$ denotes the tree produced from $s$ by
replacing every  position $i_j\in \leaf_i(s)$ with the tree $t_{i_j}$.
According to~\cite{EngelfrietS77,EngelfrietS78} there are two natural ways to extend a substitution $\sig$ to $\Tfin(\Sig\cup\cX)$: \emph{outside-in} (\emph{OI} for short) and \emph{inside-out} (\emph{IO} for short).
The corresponding notations are $\sigma_{\mathrm{oi}}$ and $\sigma_{\mathrm{io}}$ respectively.
Our goal is to extend $\sig$ to extensions $\sio$ and $\soi$ from $T(\Sig\cup H)$ to $2^{T(\Sig \cup H)}$ such that $T(\ScX)$ maps to $2^{T(\Sig)}$. In this section we restrict ourselves to
maps from $\TfinX(\Sig\cup\cX)$ to arbitrary subsets of $T(\Sig \cup H)$. The inside-out-extension of $\sio$ including infinite trees
relies on ``choice functions''.  We postpone the general definition of $\sio$ to \prref{sec:ch}. The corresponding extension of $\soi$
can be found in \prref{sec:oiinf}. It is not used elsewhere.

For a tree $s=x(s_1,\ldots,s_r)\in\TfinX(\Sig\cup \cX)$, the sets of trees  $\sio(s_i)$ and $\soi(s_i)$ are defined by induction on the maximal level of a position labeled by some variable. If $s\in T(\Sig)$, then we let $\soi(s) = \sio(s) = \os s$. Thus, we may assume that some variable occurs in $s$.
For $r=0$ we let
$\soi(s) = \sio(s) = \sig(x) \sse T(\Sig)$.
For $r\geq 1$, the sets $\sio(s_i)\sse \soi(s_i)\sse T(\Sig)$ are defined by induction for all $i\in [r]$. Hence, we can define
\begin{align}%
\label{eq:siodef}
\sio(s)
&=\set{t_x[{i_j}\lsa t_{i}]}{t_x \in \sig(x) \wedge t_{i}\in \sigma_{\mathrm{io}}(s_i) }\\%
\label{eq:soidef}
\soi(s) &=\set{t_x[{i_j}\lsa t_{i_j}]}{t_x \in \sig(x) \wedge t_{i_j}\in \sigma_{\mathrm{oi}}(s_i)}
\end{align}

\begin{rem}\label{rem:siovssioes}
For the interested reader we note that $\sio(s)\neq\es\iff\soi(s)\neq\es$. Indeed, $\sio(s)\neq\es$ implies $\soi(s)\neq\es$ because  $\sio(s) \sse \soi(s)\subseteq 2^{T(\Sig)}$ by definition. The other direction holds for $r=0$. For $r\geq 1$ the induction (on the maximal level) tells us that
$\soi(s_i)\neq\es$ implies $\sio(s_i)\neq\es$.
\qed \end{rem}
There is more flexibility in OI than in IO because positions $i_{j}\neq i_k$ of a hole $i\in H$  may be substituted with $\soi$ by different trees $t_{i_{j}}$ and $t_{i_k}$, whereas with $\sio$ they are substituted by the same tree $t_{i_{j}}= t_{i_k}$. This means the $i$-th child of a position $u$ in $s$ is duplicated if there is some
$t\in \sig(s(u))$ where $\abs{\leaf_i(t)}\geq 2$.
Hence, $\sio(s)\ssneq\soi(s)$ is possible because of ``duplication''.
\prref{fig:iooi} depicts an example for $\sio(s)\neq\soi(s)$.
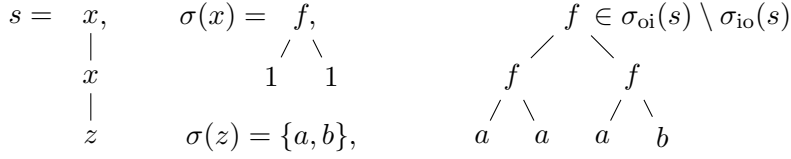
\begin{figure}[t]
\begin{center}
		\begin{tikzpicture}[scale=0.8]
		\node(s) at (-9,0) {$s=$};
				\node (ew) at (-8,0) {$x$};
				\node (ewx) at (-8,-1) {$x$};
				\node (c1) at (-8,0) {$\phantom{f}\;,$};
				\node (1) at (-8,-2) {$z$};
				\draw (ew) to (ewx);
				\draw (ewx) to (1);
			\node(sigX) at (-5.8,0) {$\sig(x)=$};
				\node (rew) at (-4.5,0) {$f$};
				\node (c2) at (-4.5,0) {$\phantom{f}\, ,$};
				\node (r1) at (-5,-1) {$1$};
				\node (r2) at (-4,-1) {$1$};
				\draw (rew) to (r1);
				\draw (rew) to (r2);
			\node(sigx) at (-5,-2) {$\sig(z)=\os {a,b}$, };
t\node(sigoiX) at (2,0) {$\in \soi(s)\sm \sio(s)$};

%				\node (1rew) at (-0.5,0) {$f$};
%					\node (comma) at (0.5,0) {$\phantom{f},$};
%				\node (1r1) at (-0,-1) {$a$};
%				\node (1r2) at (-1,-1) {$a$};
%
%				\draw (1rew) to (1r1);
%				\draw (1rew) to (1r2);
%
%				\node (s2) at (-2,0) {$\sio(x^2(z))= \{$};
%				\node (send) at (2,0) {$\mid c=a \vee c=b\}$};
				\node (rew) at (0,0) {$f$};
				\node (r1) at (-1,-1) {$f$};
				\node (r2) at (1,-1) {$f$};
				\node (r21) at (0.5,-2) {$a$};
				\node (r22) at (1.5,-2) {$b$};
				\node (r11) at (-1.5,-2) {$a$};
				\node (r12) at (-0.5,-2) {$a$};

				\draw (rew) to (r1);
				\draw (rew) to (r2);
				\draw (r1) to (r11);
				\draw (r1) to (r12);
				\draw (r2) to (r21);
				\draw (r2) to (r22);
			\end{tikzpicture}
        \end{center}
	\caption{Let $a\neq b$ be two constants and $n\geq 1$. Duplication of the hole~$1$ yields
	$|\sio(x^n(z))|=2$
	and $|\soi(x^n(z))|= 2^{n+1}$. \Ip $\soi(x^n(z))\sm \sio(x^n(z))$ for all $n\geq 1$.}%
\label{fig:iooi}
\end{figure}
 Here, $x$ is a variable of rank $1$ and $z$ is a variable of rank 0 (playing the role of ``end-of-string''). Note that in this situation (with $\sig(x)=f(1,1)$ and  $\sig(z)=\os{a,b}$) neither
$\sio(x^*(z))$ nor $\soi(x^*(z))$ is regular, since every
tree in $\sio(x^nz)$ and in  $\soi(x^nz)$ is a full binary tree with $2^n$ leaves. Reading the labels of the leaves from left-to-right reveals the difference. For $n\geq 1$ and $a\neq b$
the ``leaf-language'' of $\sio(x^nz)$ is the two-element set
$\{a^{2^n},b^{2^n}\}$ whereas the ``leaf-language'' of $\sio(x^nz)$ has $2^n$ elements. It is equal to $\os{a,\, b}^{2^{n}}$.

Consider a modification of the example in \prref{fig:iooi} by letting $\sig(z)=\es$
(or any other subset of trees in $T(\Sig)$) and $\sig(x)=\os{f(1,1),a}$. This leads to a striking situation where
$\sio(x^*(z))$ is not regular but $\soi(x^*(z))$ is the set of all finite trees over $\os{f,a}$. Indeed, $\sio(x^n(z))$
is the set of all full binary trees
with $2^m$ leaves for all $0\leq m < n$. The set $\sio(x^3(z))$ is depicted in \prref{fig:soiall}.
All leaves are labeled by $a$ and all inner nodes are labeled by~$f$. In contrast,
$\soi(x^n(z))$ is the set of all trees in $T(\os{f,a})$ with height less than $n$.
 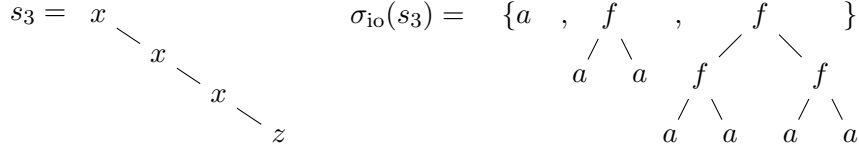
\begin{figure}[h]
	\begin{center}
		\begin{tikzpicture}[scale=0.8]

		\node(s) at (-10,0) {$s_3=$};
				\node (ew) at (-9,0) {$x$};
				\node (1) at (-8,-0.666) {$x$};
				\node (2) at (-7,-1.333) {$x$};
				\node (3) at (-6,-2) {$z$};

				\draw (ew) to (1);
				\draw (1) to (2);
				\draw (2) to (3);

		\node(s) at (-3,0) {$\sio(s_3)= \quad \{a\quad, $};

				\node (rew) at (-0.5,0) {$f$};
					\node (comma) at (0.5,0) {$\phantom{f},$};
				\node (r1) at (-0,-1) {$a$};
				\node (r2) at (-1,-1) {$a$};

				\draw (rew) to (r1);
				\draw (rew) to (r2);

				\node (rew) at (2,0) {$f$};
				\node (r1) at (1,-1) {$f$};
				\node (r2) at (3,-1) {$f$};
				\node (r21) at (2.5,-2) {$a$};
				\node (r22) at (3.5,-2) {$a$};
				\node (r11) at (0.5,-2) {$a$};
				\node (r12) at (1.5,-2) {$a$};

				\draw (rew) to (r1);
				\draw (rew) to (r2);
				\draw (r1) to (r11);
				\draw (r1) to (r12);
				\draw (r2) to (r21);
				\draw (r2) to (r22);
				\node (bracket) at (3.5,0) {$\}$};
			\end{tikzpicture}
        \end{center}
	\vspace{-0.5cm}
	\caption{$L=x^*(z)$ and $\sig(x)=\os{f(1,1),a}$. Then $\sio(L)$ is not regular, but $\soi(L)=T(\os{f,a})$.}\label{fig:soiall}
	\end{figure}

Yet another variant is given by $\Sig= \os{f,a,b}$, $H=\os 1$, and $\cX=\os{x,y}$ where $\rk(f)=2$,  $\rk(x)=\rk(y)=\rk(a)=1$, and $\rk(b)=0$. Let $\sig(x)=t$ where $t=f(t,1)$ as in \prref{fig:infcomb} and
$\sig(y)=a^*(b)$. Then  we obtain the following situation as depicted in \prref{fig:ffa}.
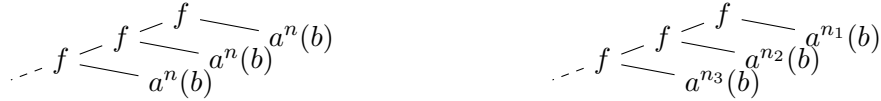
\begin{figure}[t]
\begin{center}
		\begin{tikzpicture}[xscale=0.8,yscale=0.3]
		%\soi on the left
	\node (lrew) at (0,0) {$f$};
				\node (lr1) at (-1,-1) {$f$};
				\node (lr2) at (2,-1) {$a^n(b)$};
				\draw (lrew) to (lr1);
				\draw (lrew) to (lr2);
				\node (lr11) at (-2,-2) {$f$};
				\node (lr12) at (1,-2) {$a^n(b)$};
				\draw (lr1) to (lr11);
				\draw (lr1) to (lr12);
				\node (lr111) at (-3,-3) {};
				\node (lr112) at (0,-3) {$a^n(b)$};
				\draw [dashed] (lr11) to (lr111);
				\draw  (lr11) to (lr112);
			%\si0 on the right

				\node (rew) at (9,0) {$f$};
				\node (r1) at (8,-1) {$f$};
				\node (r2) at (11,-1) {$a^{n_1}(b)$};
				\draw (rew) to (r1);
				\draw (rew) to (r2);
				\node (r11) at (7,-2) {$f$};
				\node (r12) at (10,-2) {$a^{n_2}(b)$};
				\draw (r1) to (r11);
				\draw (r1) to (r12);
				\node (r111) at (6,-3) {};
				\node (r112) at (9,-3) {$a^{n_3}(b)$};
				\draw [dashed] (r11) to (r111);
				\draw  (r11) to (r112);
			\end{tikzpicture}
        \end{center}
	\caption{For all $n\in \N$ there exists $t_n\in \sio(xy(b))$ as shown on the left. For all sequences $(n_i)_{i\in \N}$ there is a tree $t\in \soi(xy(b))$ as on the right. The set $\soi(xy(b))$ is regular, but $\sio(xy(b))$ is not.}%
\label{fig:ffa}
\end{figure}

In the sequel, we show that the domain of $\sio$ can be extended to $T(\Sig\cup\cX)$ such that equation (\ref{eq:siodef}) still holds for trees in $\TfinX(\Sig\cup\cX)$. We shall use two auxiliary concepts: complete metric spaces (\prref{sec:TOObot})  and choice functions~(\prref{sec:ch}).
%%%%%%%%%%%%%%%%%%%%%%%%%

\subsection{Complete metric spaces with \texorpdfstring{$\bot$}{bottom} for ``undefined''}\label{sec:TOObot}
Let us introduce a special constant $\bot\notin \OO$ which plays the role of ``undefined''. The idea is that an empty set $\sio(s)$ is replaced by the singleton $\bot$.
We turn $T(\OO\cup \os\bot)$ into a metric space by defining a metric $d$, which makes sure that  both sets, $\os\bot$ and $T(\OO)$, become clopen (= open and closed) subspaces in $(T(\OO\cup \os\bot),\, d)$. We  let $2^{-\infty}=0$ and define:
 \begin{align*}
 d(s,t)=
 \begin{cases}
 2^{-\inf\set{|u|}{u\in \Pos(s)\cap \Pos(t) \wedge s(u)\neq t(u)}}
 &\text{ if either $s,t\in T(\OO)$ or both, $s,t\notin T(\OO)$,}\\
 1 &\text{ otherwise: $s\in T(\OO) \iff t\notin T(\OO)$.}
\end{cases}
\end{align*}
Note that every constant in $T(\OO\cup \os\bot)$, like the constant $\bot$, has distance $1$ to every other element in $T(\OO\cup \os\bot)$.
As usual, $T(\OO)$ is closed in $T(\OO\cup \os\bot)$ but, thanks to the definition of~$d$, it is also open.
Hence,  each of the subsets $T(\OO)$, $\{\bot\}\cup T(\OO)$, and their complements with respect to $T(\OO\cup \os\bot)$ are clopen subsets in $(T(\OO\cup \os\bot),\, d)$. The restriction to the clopen subspace $\os\bot \cup T(\OO)$ yields a compact ultra-metric space such that
 \begin{align}\label{def:stamet}
 d(s,t)=2^{-\inf\set{|u|}{u\in \Pos(s)\cap \Pos(t) \wedge s(u)\neq t(u)}}
 \end{align}
  holds for all $s,t\in \os\bot \cup T(\OO)$.  The ambient metric space $(T(\OO\cup \os\bot),\, d)$ is not complete and therefore not compact. Indeed, recall that, by our convention, there is some function symbol of rank at least one. Hence, there exists an infinite tree $s$ and a Cauchy sequence $(t_n)_{n\in \N}$ in $T(\OO\cup \os\bot)\sm T(\OO)$ such that $(t_n)_{n\in \N}$ converges to $s$ in the usual Cantor-space, but not in our metric. For example, assume $\rk(a)=1$,
then the Cauchy sequence $(a^n(\bot))_{n\in \N}$ does not have any limit in
 $(T(\OO\cup \os\bot),\, d)$. However,
 as $\bot$ represents ``undefined'', we wish
 that $\lim_{n\to \infty}a^n(\bot)=\bot$. There is a way to achieve that:
 we identify the clopen set $T(\OO\cup \os\bot)\sm T(\OO)$ with the constant $\bot$. Thereby $\bot$ becomes an isolated point. To be precise, let us
 define the equivalence relation $\sim$ on $T(\OO \cup\{\bot\})$
 which is induced by letting $\bot \sim t$ for all  $t\in T(\OO\cup \os \bot)\sm T(\OO)$. Now, we have $a^n(\bot) \sim \bot$ for all~$n$. Hence, the image of the sequence in the quotient space
 $(a^n(\bot))_{n\in \N}$ is equal to the constant sequence $(\bot)_{n\in \N}$ with the obvious limit $\bot$.
 The natural embedding of $\os\bot \cup T(\OO)$ into the quotient space $T(\OO\cup \os \bot)/\sim$ induces an isometry
 between $(T(\OO) \cup \os\bot,d)$ and $(T(\OO\cup \os \bot)/\sim,d_\sim)$
 where $d_\sim$ is the  canonical quotient metric\footnote{The interested reader may consult \prref{sec:pseu} %in the appendix
 for the definition of a \emph{quotient metric} in general topology.\label{foot:pm}}
on $T(\OO\cup \os \bot)/\sim$.

\subsection{Choice functions and the definition of \texorpdfstring{$\sio$}{sigma-io} for infinite trees}\label{sec:ch}
Henceforth, $T_\bot(\Sig\cup H)$ denotes the complete metric space
$(T(\Sig\cup H)\cup \os\bot,d)$ which, by the previous subsection, is identified with the quotient space $(T(\Sig\cup H\cup \os \bot)/\sim,\,d_\sim)$.\\
A \emph{choice function} for $s \in T(\Sig\cup\cX)$ is a mapping
$\gamma : \Pos(s)\to  T_\bot(\Sig \cup  H)$ such that
\begin{align*}
\gam(u)\in \os\bot \cup T(\Sig\cup[\rk(s(u))])\setminus H \text{ with $\gam(u)=f\big(1\lds \rk(f)\big)$ if $f=s(u)\in\Sig \sm\cX$}.
\end{align*}
For $u\in \Pos(s)$ we let $\gam|_u:\Pos(s|_u)\to T_\bot(\Sig \cup  H)$ be the mapping defined by $\gam|_u(u')=\gam(u.u')$. Note that $\gam|_u$ is a choice function for the subtree $s|_u$ whenever $u\in \Pos(s)$ and $\gam$ is a choice function for $s$.
For every $s\in T(\Sig\cup\cX)$ and choice function $\gam$ for $s$, we define the sequence of trees $(\gam_n(s))_{n\in\N}$ in $T(\Sig\cup H\cup\os\bot)$ as follows:
\begin{align}\label{eq:defgamn}
\gam_0(s)=\gam(\eps)\quad \text{ and } \quad \gam_n(s)=\gam(\eps)[i_j\leftarrow (\gam|_i)_{n-1}(s|_i)]\text{ if $n>0$}.
\end{align}
This yields a Cauchy sequence $(\gam_n(s))_{n\in\N}$ in $T_\bot(\Sig\cup H)$ and
therefore, $\lim_{n\to\infty}\gam_n(s)$ exists.
Since choice functions don't map tree positions to holes, we have $\lim_{n\to\infty}\gam_n(s)
\in \os\bot \cup T(\Sig)$.
\begin{defi}\label{def:Gamsio}
Let $\sig:\cX\to 2^{T(\Sig\cup H)\setminus H}$ be a \subst and
$s\in T(\ScX)$. We define:
\begin{enumerate}
\item The tree $\gaminf(s)= \lim_{n\to \infty}\ngam n (s)\in \os\bot \cup T(\Sig)$.
\item The set $\Gam(\sig,s)$ by the set of choice functions for $s$ satisfying for all $u\in \Pos(s)$ with $x=s(u)$ the following condition:
 if $\sig(x)\neq\es$, then $\gam(u)\in \sig(x)$, otherwise $\gam(u)=\bot$.
\item The set $\sio(s)=\{\gam_\infty(s)\mid \gam\in\Gam(\sig,s)\}\sm \os \bot$.
\end{enumerate}
\end{defi}
%%%%%%%%%%%%%%%%%%%%
\begin{prop}\label{prop:gamsio}
Let $\sig:\cX\to 2^{T(\Sig\cup H)\setminus H}$ be a \subst,
$s=x(s_1\lds s_r)\in T(\Sig \cup  \cX)$ be a tree, and  $\gamma \in \Gamma(\sig,s)$ be a choice function.
Then $\gamma_\infty(s)= \gam(\eps)[i_j\lsa (\gam_{|i})_\infty(s_i)]$.
\end{prop}
\begin{proof}
Since $\gamma_\infty(s)=\lim_{n\to \infty}\ngam n (s)$ and
$(\gam_{|i})_\infty(s_i)=\lim_{n\to \infty}(\gam_{|i})_n (s_i)$ for all~$i$, we obtain
\begin{align*}
\gam(\eps)[i_j\lsa (\gam_{|i})_\infty(s_i)]
&=  \gam(\eps)[i_j\lsa \lim_{n\to \infty}(\gam_{|i})_n(s_i)]\\
&=  \gam(\eps)[i_j\lsa \lim_{n\to \infty}(\gam_{|i})_{n-1}(s_i)]\\
&=  \lim_{n\to \infty} \gam(\eps)[i_j\lsa (\gam_{|i})_{n-1}(s_i)]\\
 &= \lim_{n\to \infty}\ngam n (s) = \gamma_\infty(s). \qedhere
\end{align*}
\end{proof}
%%%%%%%%%%%%%%%%%%%%
\begin{cor}\label{cor:gamsio}
Let  $s=x(s_1\lds s_r)\in T(\Sig \cup  \cX)$ be a tree. Then we have
$\sio(s)= \set{t_x[i_j\lsa t_i]}{t_x\in \sig(x) \wedge t_i\in \sio(s_i)}$.
\Ip for finite trees the new definition of $\sio(s)$ in \prref{def:Gamsio} agrees with the earlier one given in \prref{eq:siodef}.
\end{cor}
%%%%%%%%%%
\begin{proof}
We use the notation as given in \prref{def:Gamsio}.
\begin{align*}
\sio(s)&=\bigcup \set {\gam(\eps)[i_j \lsa (\gam_{|i})_\infty (s_i)]}{\gam\in\Gamma(\sig,s)} & \text{ by \prref{prop:gamsio}}\\
&=  \set{t_x[i_j\lsa (\gam_{|i})_\infty(s_i)]}{t_x\in \sig(x),\gam_{|i}\in\Gam(\sig, s_{|i})}&\text{ trivial}\\
&=  \set{t_x[i_j\lsa t_i]}{t_x\in \sig(x) \wedge t_i\in \sio(s_i)} &
 \text{ by \prref{def:Gamsio}}
\end{align*}
Hence, $\sio(s)= \set{t_x[i_j\lsa t_i]}{t_x\in \sig(x) \wedge t_i\in \sio(s_i)}$ as desired.
\end{proof}
The following examples show that
$\sio(s)$ is not closed in $T_\bot(\Sig)$, in general.
Every (infinite) tree $s\in T(\Sig)$ can be written as a limit $\lim_{n\to \infty} s_n$ where
$s_n$ are finite trees in $T(\Sig)$ (assuming, as usual, that $\Sig$ contains a constant). Then there are Cauchy sequences $(t_n)_{n\in\N}$ with $t_n\in\sio(s_n)$ which do not converge to any tree in $\sio(s)$. In these examples $x$ is a variable of rank one, $\Sig=\os{f,a,b}$ with a constant~$b$,  $\rk(a)=1$, $\rk(f)=2$, and {where only the first hole $1\in H$ is used.}
In the first example we have $s=s_n$ for all $n$.
\begin{exa}\label{ex:limit2}
Consider $s=x(b)$ and $\sig(x)=\{a^n (1)\mid n\in\N\}$. Then $s=\lim_{n\to\infty}s_n$ where $s_n=s$.
We define $t_n=\gam^{(n)}(s)=a^n(b)$ where $\gam^{(n)}\in\Gam(s,\sig)$, $\gam^{(n)}(\eps)=a^n(1)$, and note that $(t_n)_{n\in\N}$ is a Cauchy sequence with every $t_n\in \sio(s_n)$, but \[\lim_{n\to\infty}t_n=a^\oo\notin\sio(s)=\{a^n(b)\in\N\}.\]
\end{exa}
The next example is a modification of \prref{ex:limit2} such that  $s$ becomes infinite and the sequence $s_n$ is increasing
such that $s_{n+1}=s_n[1\lsa x(1)]$.
\begin{exa}\label{ex:limit}
 Let $s=f(a^\oo, x(b))$ and $\sig$ be defined by $\sig(x)=\set{a^n(1)}{n\in\N}$. Then $s=\lim_{n\to \infty} s_n$
where $s_n=f(a^n(1), x(b))$. We have $t_n=f(a^n(1),a^n(b))\in\sio(s_n)$ since for each $n$ we may use the choice function $\gam^{(n)}$ with $\gam^{(n)}(u)= a^n(b)$ where $u=2\in \Pos(s)$. The $t_n$'s form a Cauchy sequence with $t=f(a^\oo,a^\oo)=\lim_{n\to \infty} t_n$ but $t\notin \sio(s) = \set{f(a^\oo,a^n(b))}{n\in\N}$.
\end{exa}
\section{Regular tree languages}\label{sec:regis}
There are several equivalent definitions for regular languages of finite and infinite trees, see~\cite{tata2007,LNCS2500automata,rab69,tho90handbook}.
Here, we will consider regular languages of finite and infinite trees from $T(\Sig\cup\cX\cup H \cup \os \bot)$.
Note that for $\Gam\sse \Del \sse \Sig\cup\cX\cup H \cup \os \bot$, the sets $\Tfin(\Del)$ (more general: $\TfinGam(\Del)$) and~$T(\Del)$ are regular subsets of $T(\Sig\cup\cX\cup H \cup \os \bot)$.
%VD Not NEEDED AT THIS PLACE: Moreover, there are isometrical embeddings of
%$T(\Del)\sse T(\Sig\cup\cX\cup H \cup \os \bot)$ and
%$T_\bot(\Del)\sse T_\bot(\Sig\cup\cX\cup H)$. The element $\bot\in T_\bot(\Sig\cup\cX\cup H)$ is isolated since it has distance at least $1$ to any other element.

\subsection{Nondeterministic tree automata with parity acceptance}\label{sec:parnta}
We use parity-NTAs (nondeterministic tree automata with a parity acceptance condition) as the basic reference  to characterize regular tree languages.
The parity condition is used to accept infinite trees. In our definition a parity-NTA accepts sets of finite and infinite trees.
Alternating tree automata with a parity acceptance condition will be considered later.

Let $\Del$ be any ranked alphabet (finite as usual) with a rank function $\rk:\Del\to \N$.
A \emph{parity-NTA} over $\Del$ is specified by
a tuple $A=(Q,\Del,\del,\chi)$ where
$Q$ is a nonempty finite set of states,
$\del$ is the \emph{\tra relation}, and
$\chi:Q\to C$ is a coloring with $C=\os{1,\lds \abs C}$. Without restriction, we assume that $\abs C$ is odd.
We have
\begin{equation}\label{eq:ntadel}
\del \sse \bigcup_{f\in \Del}Q\times \os{f}\times Q^{\rk(f)}.
\end{equation}
Thus,  $\del$ is a set of tuples
$
\big(p,f,p_{1}\lds p_{r})
$
where $r=\rk(f)$. The acceptance condition is defined as follows.
\begin{defi}\label{def:parityacc}
Let $A$ be a parity-NTA and
let $t\in T(\Del)$.
\begin{itemize}
\item A \emph{run $\rho$ of $t$} is a relabeling of the positions of $t$ by states which is consistent with the \tra{s}. That is, $\rho:\Pos(t)\to Q$ is a mapping such that for all $u\in \Pos(t)$ with $t(u)=f$ there is a \tra $(p,f,p_{1}\lds p_{\rk(f)})\in \del$ such that $\rho(u)=p$ and $\rho(u.j)=p_j$ for all $j\in [\rk(f)]$. (See~\cite{LNCS2500automata} for more details.)\\
If $\rho(\eps)=p$, then we say that $\rho$  is a \emph{run of $t$ at state $p$}.
\item Let $\rho$ be a run and $(u_0,u_1,u_2,\ldots)$ be an infinite path in $\Pos(t)$ such that $u_{j+1}$ is a child of $u_j$. The path is \emph{accepting} if the number
\[\liminf_{k\to \infty}(\chi\rho(u_k))=\max\set{\min\set{\chi\rho(u_i)\in C}{i\geq k}}{k\in \N}\]
is even. This means that the minimal color appearing infinitely often on this path is even. The run $\rho$ is \emph{accepting} if all its all infinite directed paths are accepting.
\item By $\RUN_A(t,p)$ we denote the set of \textbf{accepting} runs of $t$ at  state $p$. (If the context to $A$ is clear, we might simply write $\RUN(t,p)$.)
\item If $p\in Q$, then we let $L(A,p) = \set{t\in T(\Del)}{\RUN_A(t, p) \neq \es}$. It is the \emph{accepted language of $A$ at state $p$}.
\item For $P\sse Q$ we define $L(A,P)$ by
\begin{equation}\label{eq:defLAP}
L(A,P)=\bigcap\set{L(A,p)}{p\in P}.
\end{equation}
The set $L(A,P)\sse T(\Del)$ is called the \emph{accepted language at the set $P$}.
The convention is as usual: $L(A,\es)= T(\Del)$.
\end{itemize}
\end{defi}

\noindent
For $\rk(f)=0$ there are no children; and we accept $f$ at state $p$ \IFF
$(p,f)\in \del$.
Therefore, no final states appear in the specification of parity-NTA\@. For a finite tree $t\in T(\Del)$ we have
$t\in L(A,p)$ \IFF there exists a run $\rho$ of $t$ with $\rho(\eps)=p$ since the parity condition holds vacuously. Thus, for trees in $\Tfin(\Del)$ no other accepting condition is required than the existence of a run.
The following fact is well-known, see for example~\cite{LNCS2500automata}.
We shall use %the characterization in
\prref{prop:regispnta}
as the working  definition for the present paper.
\begin{prop}\label{prop:regispnta}
A language $L\sse T(\Del)$
is regular \IFF there is a parity-NTA $A=(Q,\Del,\del,\chi)$ and a state $p\in Q$ such that $L= L(A,p)$.
\end{prop}
It is  a well-known classical fact that the class of regular tree languages forms an effective Boolean algebra~\cite{rab69}. This means that first, there is a parity-NTA accepting all trees in $T(\Del)$ and second, given two parity NTAs $A_1$ and $A_2$ with states $p_1$ and $p_2$ respectively, we can effectively construct a parity NTA accepting $L(A_1,p_1)\sm L(A_2,p_2)$.
\Ip
\prref{prop:regispnta} implies that
for each parity-NTA $A$ and all subsets $P,P'\sse Q$ we can effectively construct a parity-NTA $B$ with a single initial state $q$ such that
$L(B,q)= L(A,P)\sm L(A,P')$.

\section{Tasks and profiles}\label{sec:tapro}
In this section we introduce the notions of a \emph{task} and a \emph{profile}. Throughout, we let $H=\os{ 1\lds \abs H}$ be a set of holes and $\Sig$ be a ranked alphabet such that $\rk(f)\leq \abs H$ for all $f\in \Sig$. We also fix a parity-NTA $B=(Q,\Sig,\del,\chi)$ such that
$L(B,p)\sse  T(\Sig)$ for all $p\in Q$.
Every such $B$ is extended to a parity-NTA $B_H$ by adding more \tra{s} ensuring that every hole $i\in H$ is accepted at all states. Thus, we let $B_H=(Q,\Sig\cup H,\del_H,\chi)$ where $\del_H= \del\cup (Q\times H)$.
Clearly, $L(B_H,p)$ is regular and therefore  the complement
$T(\Sig\cup H)\sm L(B_H,p)$ is regular, too. For the rest of this section, a run $\rho$ of a tree $t$ refers to a tree $t\in  T(\Sig\cup H)$ and the
NTA $B_H$. Thus $\RUN(t,p)=\RUN_{B_H}(t,p)$ if not stated explicitly otherwise.
Let $\rho$ be any run of a tree $t:\Pos(t)\to \Sig$ (not necessarily accepting) and let $\alp= (u_0,u_1, \ldots)$, $\alp'= (u'_0,u'_1, \ldots)$ be infinite directed paths in $\Pos(t)$ with $\eps= u_0=u'_0$. We concentrate on infinite paths because these paths decide whether $\rho$ is accepting.
Since directed paths in trees are directed away from the root,
$u_{i+1}$ is a child of $u_{i}$ and $u'_{i+1}$ is a child of $u'_{i}$ for all $i\in \N$.
The paths $\alp$ and $\alp'$ define infinite words over $\Sig$ and
the run $\rho$ defines infinite words $\rho(\alp)$ and $\rho(\alp')$ over $Q$.
Suppose each path is cut into infinitely many finite and nonempty pieces
$w_i$ and $w_i'$ such that we can write
$\rho(\alp)=\rho(w_0)\rho(w_1) \cdots$
and $\rho(\alp')=\rho(w'_0)\rho(w'_1) \cdots$
where all $\rho(w_i)$ and $\rho(w'_i)$ are in  $Q^+$. Assume that
$\alp'$ satisfies the parity condition.
We aim for a condition based on the factors $\rho(w_i)$ and $\rho(w'_i)$ which is strong enough to imply that the other infinite word
$\alp$ satisfies the parity condition, too.\footnote{To find such a sufficient condition we just mimic the standard concept of a strongly recognizing morphism in the theory of $\oo$-regular words as exposed for example in the appendix, \prref{sec:oldconwords}.}

Naturally, such a condition is related to colors which appear in $\chi\rho(w_i)$ and $\chi\rho(w'_i)$. More precisely,
let $c_i$ (resp.~$c'_i$) denote the minimal color
in $\chi\rho(w_i)$ (resp.~$\chi\rho(w'_i)$) for $i\in \N$.
We obtain two infinite words $(c_0,c_1, \ldots)$ and
$(c'_0,c'_1, \ldots)$ over the alphabet $C$ of colors.
If we compare  locally each color $c_i$ with the corresponding
color $c'_i$, then intuitively, with respect to the parity condition, an even color is better than an odd color.
A small even color is better than a large even color. However, a large odd color is better than a small odd color. Based on that intuition, let us introduce the \emph{best-ordering}
$\prebest$ on the set $\os{0\lds \abs C}$.
Note that we explicitly include $0$ which is not in $\chi(Q)$
 and that $\abs C$ is odd by our convention. We let
\begin{align}\label{eq:prebest}
0\prebest 2 \prebest \cdots \prebest \abs C-1 \prebest\abs C \prebest \abs C-2 \prebest \cdots \prebest 3 \prebest 1
\end{align}
Indeed, in the linear order $\prebest$ all even numbers are ``better'' than the odd numbers.
 Among even numbers smaller is better. Among odd numbers larger is better.

As a consequence, if $c_i\prebest c_i'$ for all $i\in \N$ and if
the parity condition holds for $\alp'$, then it holds for $\alp$.

\begin{defi}\label{def:task}
A \emph{task} is a tuple $(p,\psi_1\lds \psi_{|H|})$ with $p\in Q$ and $\psi_i : Q\to \os{0}\cup \chi(Q)$ for $1\leq i \leq |H|$.
A \emph{profile} is a set of tasks.
\end{defi}
The range of $\psi_i$ in \prref{def:task} is a subset of the set
 $\os{0\lds \abs C}$ where $0$ does not belong to $\chi(Q)$, but it is the best number in the linear order $\prebest$. As we will see, for ``satisfying'' a task the value $\psi_i(q)=0$ is better than any value in $\chi(Q)$, The best (resp.~worst) value in $\chi(Q)$ is the smallest even (resp.~odd) number.
The best-ordering defines a  partial order on the set of tasks:
we write $(p,\psi_1\lds \psi_{|H|})\leq (p',\psi'_1\lds \psi'_{|H|})$
if first, $p=p'$ and second, $\psi_i(q)\prebest \psi'_i(q)$ for all $q\in Q$ and
$i\in H$.
\begin{defi}\label{def:runmodtask}
Let $\tau=(p,\psi_1\lds \psi_{|H|})$ be a task and
{$t,t'\in T(\Sig \cup  H \cup \os\bot)$} be trees.
\begin{enumerate}
\item A run $\rho\in \RUN_{B_H}(t,p)$ \emph{satisfies $\tau$} if the following condition holds for all leaves $i_j\in \leaf_i(t)$:
\begin{itemize}
\item If $c_\rho(i_j)$ denotes the minimal
color (\wrt to the natural order $\leq$ for natural numbers)  on the path from the root of $t$ to position $i_j\in \leaf_i(t)$ in the tree $\chi\rho:\Pos(t)\to C$, then we have
$c_\rho(i_j) \prebest\psi_i(\rho(i_j))$.
\end{itemize}
We also write $\rho\models \tau$ in that case. Moreover, we write
$t\models \tau$ if there exists some $\rho\in \RUN_{B_H}(t,p)$ such that $\rho\models \tau$. Note that a tree containing the symbol $\bot$ cannot satisfy any task because there is no run on such a tree.

\item A run $\rho\in \RUN_{B_H}(t,p)$ \emph{defines $\tau$} if
first, $\rho\models \tau$ and second, for all $q\in Q$ and
$i\in H$ we have $\psi_i(q)\geq 1 \iff \exists i_j\in \leaf_i(t)
: \psi_i(q) = c_\rho (i_j)$.
\item The set of all tasks $\tau$  such that there is a tree
$t\in T(\Sig\cup H\cup \{\bot\})$ with $t\models \tau $ is denoted by
$\cT = \cT_B$.
\item By $\pi(t)$ we denote the set of tasks $\tau$ such that $t\models \tau$. \Ip $\tau \in \pi(t)$ implies $\tau \in \cT$.
\item The set of all profiles $\pi$ such that there is a tree
$t\in T(\Sig \cup  H \cup \os\bot)$ with $\pi=\pi(t)$ is denoted by $\Pro = \Pro_B$.
\item By ${\equiv_B}$ we denote the equivalence relation which is defined by $t\equiv_B t'\iff \pi(t)= \pi(t')$.
\end{enumerate}
\end{defi}

\noindent
Note that $\pi(t)=\es$ \IFF $\RUN_{B_H}(t,p)=\es$ for all $p\in Q$. Since $\bot\notin T(\Sig \cup H)$, we have $\pi(\bot)=\es$.
The index of ${\equiv_B}$ is finite since it is bounded by $\abs \Pro$.
More precisely,
\begin{equation}\label{eq:numprof}
\abs{\set{\set{t\in T(\Sig\cup H)}{t\equiv_B t'}}{t'\in T(\Sig\cup H)}} \leq
\abs{\Pro_B} \leq 2^{\abs Q\cdot (\abs C +1)^{\abs{H\times Q}}}.
\end{equation}
The value $0$ in the range of $\psi_i$ plays the following role:
Let $q\in Q$ and $\rho\models (p,\psi_1\lds \psi_{|H|})$ with $\rho\in \Run p t$. Then $\psi_i(q)=0$ implies   $\set{i_j\in \leaf_i(t)}{\rho(i_j)=q}=\es$. The condition is vacuously true if $\leaf_i(t)=\es$.
 The following proposition says that $\set{t\in T(\Sig \cup H)}{t\models \tau}$ is effectively regular for every task $\tau\in \cT_B$.
The proof of \prref{prop:treg} is based on a product construction where the first component simulates the NTA $B_H$ and the second component  remembers the minimal color appearing on paths in runs $\rho$ of $B_H$ at states $p$.
\begin{prop}\label{prop:treg} Let $\tau= (p,\psi_1\lds \psi_{|H|})\in \cT_B$ be any task.
Then the set of trees $\set{t\in T(\Sig \cup H)}{t\models \tau}$ is effectively regular.  More precisely,
$\set{t\in T(\Sig \cup H)}{t\models \tau}=$\\ $L(B_\tau,(p,\chi(p)))$ where
$B_\tau= (Q\times C, \Sig \cup H, \del_\tau, \chi_\tau)$ is the following parity-NTA\@.

For a color $c\in C$, a function symbol $f\in \Sig$ with $r=\rk(f)$, a hole $i\in H$, and states $q, p_1\lds p_r\in Q$ we define  $\del_\tau$ by the following equivalences:
\begin{align}\label{eq:lem101}
\hspace{-0.6cm}
\big((q,c) ,f, (p_1,\min\os{c,\chi(p_1)})\lds
(p_r,\min\os{c,\chi(p_r) })\big)\in\del_\tau &\iff (q,f, p_1\lds p_{r}) \in \del_B\\
\big((q,c),i\big)\in \del_\tau &\iff c\prebest \psi_i(q)\label{eq:lem102}
\end{align}
The color of a pair $(q,c)$ is defined by
$\chi_\tau(q,c) = \chi(q)$.
\end{prop}
\begin{proof}
Let $\pr1:Q\times C\to Q, (p,c)\mapsto p$ be the projection onto the first component. Let $t\in T(\Sig \cup H)$. The aim is to show that $\pr1$ defines for $t$ a canonical bijection between accepting runs $\rho$ at a state $p$ in $B_H$ satisfying~$\tau$ and  the set of accepting runs $\rho_\tau$ at the state~$(p,\chi(p))$ in~$B_\tau$.
This implies $\set{t\in T(\Sig \cup H)}{t\models \tau}=L(B_\tau,(p,\chi(p)))$ and hence, $\set{t\in T(\Sig \cup H)}{t\models \tau}$ is effectively regular.

The definition of $\del_\tau$ says
that $\pr1 \rho':\Pos(t)\to Q$ is a run in $\RUN_{B_H}(t,p)$
 for every $t\in T(\Sig \cup H)$ and every run
$\rho'\in \RUN_{B_\tau}(t,(p,\chi(p)))$. Here, as usual, $\pr1\rho'=\pr1\circ \rho'$ denotes the composition. Consider any nonempty directed path $\eps,u_1\lds u_k$ in $t$, then, by induction on $k$, we see that $\rho'(u_k)=(q_k,c_k)$ satisfies $c_k=\min\os{c_1\lds c_k}$. We claim that $\pr1\rho' \models \tau$.  Indeed, the minimal color seen on every infinite directed path of the tree $\rho't$ with $\rho'(\eps)=(p,\chi(p))$ is even.
Since $\chi_\tau(q,c) = \chi(q)$, the same holds for $\pr1 \rho'$.
Together with the definition of $\del_\tau((q,c),i)$ we obtain the claim that
$\pr1\rho' \models \tau$.
Thus, the projection onto the first component defines a mapping
\begin{align}\label{eq:projrun}
\pr1: \RUN_{B_\tau}(t,(p,\chi(p)))\to \set{\rho\in \RUN_{B_H}(t,p)}{\rho \models \tau}, \; \rho'\mapsto \pr1(\rho') = \pr1\rho'.
\end{align}
We wish a bijection. Therefore, let us also define a mapping $\rho\mapsto \rho_\tau$ in the other direction from $\set{\rho\in \RUN_{B_H}(t,p)}{\rho \models \tau}$
to $\RUN_{B_\tau}(t,(p,\chi(p)))$. Given $\rho\in \RUN_{B_H}(t,p)$ such that $\rho \models \tau$, we define the run $\rho_\tau$ top-down beginning at the root with $\rho_\tau(\eps)=(p,\chi(p))$.
Now, assume that $\rho_\tau(u)=(q,c)$ is defined for $u\in \Pos(t)$ such that $\rho(u)=q$. Suppose $t(u)=f$ and $\rk(f)=r$
Then, thanks to $\rho\in \RUN_{B_H}(t,p)$, there is some $(q,f,p_1\lds p_r)\in \del_{B_H}$ such that
$\rho(u.i)=p_i$ for all $1\leq i \leq r$. However, since
$\rho_\tau(u)=(q,c)$ is fixed there is exactly one corresponding \tra $\big((q,c) ,f, (p_1,\min\os{c,\chi(p_1)})\lds
(p_r,\min\os{c,\chi(p_r) })\big)\in\del_\tau$ and we obtain
 $\rho_\tau(u.i)=\min\os{c,\chi(p_i)}$ for all $1\leq i \leq r$.
 This is clear for $f\in \Sig$ by (\ref{eq:lem101}) and
 $f=i\in H$ by (\ref{eq:lem102}) because $\rho\models \tau$.
 Moreover, since every infinite path in the tree $\rho$
 satisfies the parity condition, the same is true for the run
 $\rho_\tau$ because the color of a state $\chi(q,c)$ is defined by the first component: $\chi(q,c)=\chi(q)$. Thus, $\rho_\tau\in
 \RUN_{B_\tau}(t,(p,\chi(p)))$ as desired.
 A straightforward inspection shows $(\pr1\rho')_\tau = \rho'$
 for all $\rho'\in
 \RUN_{B_\tau}(t,(p,\chi(p)))$ and $\pr1(\rho_\tau) = \rho$
 for all $\rho\in \RUN_{B_H}(t,p)$ where $\rho \models \tau$.
  This shows that the mapping $\pr1$  defined in (\ref{eq:projrun}) is bijective. We conclude
  $\set{t\in T(\Sig \cup H)}{t\models \tau}= L(B_\tau,(p,\chi(p)))$. \end{proof}

\begin{cor}\label{cor:pireg}
Let $\pi$ be a profile, then $\set{t\in T(\Sig \cup H \cup \{\bot\})}{\pi(t)=\pi}$ is effectively regular.
\end{cor}

\begin{proof}
Let $\pi=\es$, first. Then
$\set{t\in T(\Sig \cup H \cup \{\bot\})}{\pi(t)=\pi}$ is the disjoint union of $T(\Sig \cup H \cup \{\bot\})\sm T(\Sig \cup H)$
(which is regular) and the set $\set{t\in T(\Sig \cup H)}{\pi(t)=\pi}$. For $\pi\neq \es$, we have $\set{t\in T(\Sig \cup H \cup \{\bot\})}{\pi(t)=\pi}\sse \set{t\in T(\Sig \cup H )}{\pi(t)=\pi}$.
Thus, it is enough to show that $\set{t\in T(\Sig \cup H)}{\pi(t)=\pi}$ is regular.
This set coincides with
\begin{align}\label{eq:piisreg}
\bigcap_{\tau\in \pi}\set{t\in T(\Sig \cup H)}{t\models\tau} \cap \bigcap_{\tau\notin \pi}
\set{t\in T(\Sig \cup H)}{t\not\models\tau}
\end{align}
which is a finite intersection of languages, or complements of languages, from the set of languages
$\cL=\set{\set{t\in T(\Sig \cup H)}{t\models\tau}}{\tau\in\tau_B}$. By \prref{prop:treg}, the family $\cL$ consists of regular tree languages. Since the class of regular tree languages in $T(\Sig \cup H\cup \{\bot\})$ is an effective Boolean algebra~\cite{rab69}, we conclude that $\set{t\in T(\Sig \cup H \cup \{\bot\})}{\pi(t)=\pi}$ is effectively regular, too.
\end{proof}

\begin{prop}\label{prop:equivtasl}
Let $B$ be a parity NTA with state set $Q$ such that $L(B,p)\sse T(\Sig)$ for all $p\in Q$ and $s \in  T(\Sig \cup \cX)$ be a tree. If $\gam,\gam':\Pos(s) \to T_\bot(\Sig \cup H)$ are two choice functions for $s$
such that $\gam'(u) \equiv_B \gam(u)$ for all $u\in \Pos(s)$, then
$\gamma_\infty(s)\in  L(B,p_0)\iff \gamma_\infty'(s)\in  L(B,p_0)$
for all states $p_0$ of $B$.
\end{prop}
In the following proof we encounter trees with infinite directed paths and where nodes may have an infinite degree:
the $\bet$-sequences. This requires some attention.
%%%%%%%%%%%%%%%%%%%%
\begin{proof}Let $s=x(s_1\lds s_r)$ and $\gam_\infty(s)\in  L(B,p_0)$. By symmetry, it is enough to show $\gamma'_\infty(s)\in  L(B,p_0)$. Recall that
$\gam_\infty(s)$ is the limit of trees $\gam_n(s)$.
\prref{fig:pgnps} depicts the basic relationship between paths in $\gam_n(s)$ and
paths in $s$. For the path in $s$ it shows  positions $u\in \Pos(s)$ which can be identified with positions on the corresponding path in $\gam_n(s)$.
Let us explain the relationship of the depicted paths in more detail.
We begin with a directed path $p$ in $\gam_n(s)$ starting at
the root $\eps$ which defines $t_0=\gam(\eps)$. The path $p$ stays inside $t_0$ or it leaves $t_0$ using a hole labeled by some $h_1\in H$.
Assuming the second case, we let $u_0=\eps$ and $h_1=u_1\in \Pos(s)\cap \Nat$. Hence, $t_1= \gam(u_1)$ exists\footnote{Observe that
$\Pos(s)\cap H=[r]$. Thus, the choice function determines $t_1$.}. The path $p$ continues in the tree $\gam_n(s)$. Again, there are two possibilities: it stays inside $t_1$ or it leaves $t_1$ using a hole labeled by some $h_2\in H$. If it leaves, we let $u_2=h_1\,{.}\,h_2 \in \Pos(s)\cap \Nat^2$.
If possible, we continue by defining $u_3\in \Pos(s)$, $t_3=\gam(u_3)$, and $h_3\in H$, etc. {}From the very beginning we face two cases: either the path stops at some leaf of $\gam_n(s)$
(which is also a leaf of $t_n$) labeled by some $h_{n+1}\in H$
or there is some $m\leq n$ such that
the path stays in $t_m$.  \prref{fig:pgnps} shows a situation where  $p$ ends in a leaf of $\gam_n(s)$. This implies that there exists $u_n{.} h_{n+1}\in \Pos(s)$.
Assume $p$ stays in $t_m$ for some $m\leq n$. This path can be finite or it can be infinite. Moreover, it may possibly stay inside $t_m$ although, perhaps, for all $h\in H$ there exist positions $u_m{.} h\in \Pos(s)$. However, in that case, we do not define any position $u_{m+1}$.
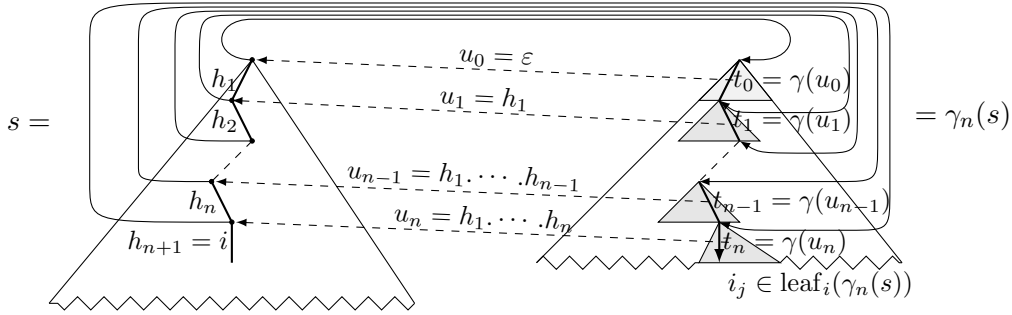
\begin{figure}
\begin{center}
\begin{tikzpicture}[inner sep=1pt,scale=0.54,>=latex]
\node[above=3pt] at (-11.5,1.5) {$s =$};
\node[above=3pt] at (11.5,1.5) {$= \gam_n(s)$};
\draw[fill] (-6,3.5) circle (1.5pt) (-6.5,2.5) circle (1.5pt) (-6,1.5) circle (1.5pt) (-7,.5) circle (1.5pt) (-6.5,-.5) circle (1.5pt);
\draw[thick]  (-6,3.5) -- node[left]{\small $h_1$}(-6.5,2.5);
\draw[thick]  (-6.5,2.5) -- node[left]{\small $h_2$}(-6,1.5);
\draw[dashed] (-6,1.5) -- (-7,0.5);
\draw[thick] (-7,0.5) --  node[left]{\small $h_{n}$} (-6.5,-.5);
\draw[thick] (-6.5,-.5) --  node[left]{\small $h_{n+1}=i$} (-6.5,-1.5);
\draw[fill=gray!10] (6,3.5) -- (5,2.5) -- (6.8,2.5) -- (6,3.5);
\draw[fill=gray!20] (5.5,2.5) -- (4.5,1.5) -- (6.5,1.5) -- (5.5,2.5);
\draw[fill=gray!20] (5,.5) -- (4,-.5) -- (6,-.5) -- (5,.5);
\draw[fill=gray!20] (5.5,-.5) -- (5,-1.5) -- (7,-1.5) --  (5.5,-.5);
\draw[thick,->] (6,3.5) -- node[right]{\small $t_0=\gam(u_0)$} (5.5,2.5) -- node[right]{\small $t_1=\gam(u_1)$} (6,1.5) (5,.5) -- node[right]
{\small $t_{n-1}=\gam(u_{n-1})$}(5.5,-.5) -- node[right]{\small $t_{n}=\gam(u_n)$} (5.5,-1.5);
\node[below right=2pt] at (5.5,-1.5) {\small $i_j\in \leaf_i(\gam_n(s))$};
\draw  (-2,-2.5) -- (-6,3.5) -- (-11,-2.5);
\draw (6,3.5) -- (1,-1.5) (10,-1.5) -- (6,3.5);% (0,-3) -- (-4,3) -- (-9,-3);
\draw decorate [decoration=zigzag] {(1,-1.5) -- (5,-1.5)};
\draw decorate [decoration=zigzag] {(7,-1.5)  -- (10,-1.5)};
\draw decorate [decoration=zigzag] {(-2,-2.5) -- (-11,-2.5)};
\draw[dashed] (6,1.5) -- (5,.5);
\draw[dashed,<-] (-5.9,3.5) -- node[sloped,above]{\small $u_0=\eps$} (5.9,3.);
\draw[dashed,<-] (-6.5,2.5) -- node[sloped,above]{\small $u_1=h_1$} (5.9,1.9);
\draw[dashed,<-] (-6.9,0.5) -- node[sloped,above]{\small $u_{n-1}=h_1.\cdots.h_{n-1}$}(5.3,0);
\draw[dashed,<-] (-6.4,-.5) -- node[sloped,above]{\small $u_n=h_1.\cdots.h_{n}$} (5.8,-1.0);
\draw[->] (-6,3.5)   .. controls (-7,3.5) and (-7,4.5) .. (-6,4.5) --
               (6.5,4.5) .. controls (7.5,4.5) and (7.5,3.5) .. (6.5,3.5) -- (6,3.5);
\draw[->] (-6.5,2.5) .. controls (-7.3,2.5) and (-7.3,3.2) ..
               (-7.3,3.4) .. controls (-7.3,4.6) .. (-6.3,4.6) -- (7.5,4.6)
               .. controls (8.7,4.6)  .. (8.7,3.4) -- (8.7,3.2) .. controls (8.7,2.2) .. (7.7,2.2) -- (6.7,2.2)
                .. controls (5.9,2.2) .. (5.5,2.5);
\draw[->] (-6,1.5) -- (-6.5,1.5) .. controls (-7.9,1.5) .. (-7.9,2.5) -- (-7.9,3.7)
               .. controls (-7.9,4.7) .. (-6.5,4.7) -- (7.9,4.7) .. controls (8.9,4.7) .. (8.9,3.7)
               -- (8.9,2.2) .. controls (8.9,1.2) .. (7.9,1.2) .. controls (6.5,1.2) .. (6,1.5);
\draw[->] (-7,.5)  .. controls (-8.8,.5) .. (-8.8,.9) -- (-8.8,3.8)
               .. controls (-8.8,4.8) .. (-6.7,4.8) -- (8.1,4.8) .. controls (9.4,4.8) .. (9.4,3.8)
               -- (9.4,1.5) .. controls (9.4,.5) .. (8.4,.5) -- (5,.5);
\draw[->] (-6.5,-.5) -- (-7,-.5) .. controls (-10,-.5) .. (-10,.5) -- (-10,3.9)
               .. controls (-10,4.9) .. (-6.6,4.9) -- (8.3,4.9) .. controls (9.7,4.9) .. (9.7,3.9)
               -- (9.7,.3) .. controls (9.7,-.7).. (8.7,-.7) -- (7.3,-.7) .. controls (6.1,-.7) .. (5.5,-.5);
\end{tikzpicture}
\end{center}
\caption{A directed path  in $\gam_n(s)$ starting at the root maps to a unique sequence of positions $\eps, h_1, h_1.h_2 \ldots$ in  $s$ (indicated by dashed arrows). These positions of the sequence can be identified with positions in $\gam_n(s)$: the top positions of the small triangles (indicated by plain arrows from positions in $s$ to positions in $\gam_n(s)$). The depicted directed path ends in a leaf $i_j$, but it could also stay inside $t_n$ or end at a leaf of $t_n$ which is a constant.}\label{fig:pgnps}
\end{figure}

Now, consider
 $u\in \Pos(s)$ and $p\in Q$ such that $\RUN_{B_H}(\gam(u),p)\neq \es$. Then each run $\rho\in  \RUN_{B_H}(\gam(u),p)$
defines a task $\tau_\rho= (p,\psi_1\lds \psi_{|H|})$ by\begin{equation}\label{eq:mintaurp}
\psi_i(q) =
 \text{sup}_{\prebest}\set{c_\rho(i_j) \in \os{0\lds \abs C}}{\exists i_j\, : q =  \rho(i_j)}.
\end{equation}
The existence of $\rho$ implies that
 $\gam(u)\models \tau_\rho$. The task $\tau_\rho$ is minimal\footnote{Minimal \wrt  the natural partial order on tasks defined by $\prebest$} among all tasks $\tau$ such that $\gam(u)\models \tau$. The minimality of $\tau_\rho$ implies that there is some leaf $i_j$ in $\gam(u)$ where $q =  \rho(i_j)$ \IFF  $\psi_i(q) \geq 1$.
 Moreover,  Since $\gam(u)\equiv_B\gam'(u)$, there exists
 some $\rho'\in \RUN_{B_H}(\gam'(u),p)$ such that $\rho'\models \tau_\rho$. It follows that there is a mapping
\begin{align}\label{eq:rho2'}
\theta:\RUN_{B_H}(\gam(u),p) \to \RUN_{B_H}(\gam'(u),p), \quad\rho \mapsto\rho'
\end{align}
such that $\rho'=\theta(\rho)$ satisfies $\rho'\models \tau_\rho$.

The definition of $\theta$ is crucial to prove $\gamma'_\infty(s)\in  L(B,p_0)$. Note that the shape of $t=\gam(u)$ and $t'=\gam'(u)$ might be very different. For example, $t$ can be finite whereas $t'$ is infinite, or vice versa. The trees share  however $\leaf_i(t)\neq\es \iff \leaf_i(t')\neq\es$ because $\gam(u)\equiv_B\gam'(u)$. Still, the cardinalities of $\leaf_i(t)$ and $\leaf_i(t')$
might differ drastically.

The rest of the proof uses an additional concept of an \emph{IO-prefix}. Let $t,t_\eps, t_i\in T_\bot(\Sig \cup H)$ be trees for each
$i\in H$ where $\leaf_i(t_\eps)\neq \es$.
We say that \emph{$t_\eps$ is an IO-prefix of $t$} if we can write
$t=t_\eps[i_j\lsa t_i]$.

In order to see $\gamma'_\infty(s)\in  L(B,p_0)$ we fix a run
$\rho\in \Run{\gamma_\infty(s)}{p_0}$ witnessing $\gamma_\infty(s)\in  L(B,p_0)$. Using this fixed $\rho$ we
construct a set of (finite or infinite) sequences
$\alp= (u_0,p_0,\rho_0) \cdots (u_k,p_k,\rho_k)$ and
$\bet= (u_0,p_0,\rho'_0) \cdots (u_k,p_k,\rho'_k)$ such that
for all $i$ we have $u_i\in \Pos(s)$, $p_i\in Q$, and
$\rho_i\in \Run{\gam(u_i)}{p_i}$ (resp.~$\rho'_i\in \Run{\gam'(u_i)}{p_i}$). See \prref{fig:pgnps} for an intuition about the positions $u_i$.
We require that  $u_{j+1}$ is child of $u_j$ for all $0\leq j<k$.
The set of sequences $\alp$ (resp.~$\bet$) form themselves a rooted tree where the
root is the empty sequence. The parent of a nonempty sequence $w_0 \cdots w_{k-1}w_k$ of length $k+1$.
As we will see later, each node in the tree defined by an $\alp$-sequence has at most $\abs{H\times C}$ children whereas a node in the tree defined by a $\bet$-sequence may have infinitely many children. The number of children is actually equal to the size of the index sets $I_\alp$ resp.~$I_\bet$ as defined in (\ref{eq:Ialp}) and (\ref{eq:Ibet})
below.

Every such sequence $\alp= (u_0,p_0,\rho_0) \cdots (u_k,p_k,\rho_k)$  will define below a unique position $\nu(\alp)$ in $ \gam_\infty(s)$ such that the subtree of $\gam_\infty(s)$ at position $\nu(\alp)$ has the tree $\gam(u_k)$ as an IO-prefix as defined above. The same will be true for a sequence $\bet$ with respect to  $\gam'_\infty(s)$. The position of the empty sequence $\nu(\eps)$ is $\eps$ which is the root of $\gam_\infty(s)$ and $\gam'_\infty(s)$.
We define a sequence $\alp_1=(u_0,p_0,\rho_0)= (\eps,p_0,\rho_0)$ such that
$\rho_0\in \Run{\gam(\eps)}{p_0}$ is induced by $\rho\in \Run{\gam_\infty(s)}{p_0}$.
This yields  $\bet_1=(\eps,p_0,\rho'_0)$ with $\rho'_0= \theta(\rho_0)$ as it was defined by (\ref{eq:rho2'}).

Now, let $k\in \N$ and assume that a sequence $\alp= (u_0,p_0,\rho_0) \cdots (u_{k},p_{k},\rho_{k})$
is already defined.
By induction, we also know a position $\nu(\alp')\in \Pos(\gaminf(s))$ where $\alp'$ is defined by the equation $\alp'(u_{k},p_{k},\rho_{k})= \alp$.
Let
\begin{equation}\label{eq:Ialp}
I_\alp=\set{(i,q)\in H\times Q}{\exists i_j \in \leaf_i(\gam(u_k)): \rho_k(i_j)= q}.
\end{equation}
For each $(i,q)\in I_\alp$ we define a triple
$(u_k{.}i,q,\rho(i,q))$. Here, $u_k.i$ is the $i$-th child of $u_k\in \Pos(s)$. The children of $\alp$ are defined by the sequences
$\alp \cdot (u_k.i,q,\rho(i,q))$. Thus, there are at most $|H\times C|$ children in $\alp$-sequences. Before we define the run $\rho(i,q)$, let us define the position $\nu(\alp)\in \Pos(\gaminf(s))$. By induction, the subtree at position $\nu(\alp')\in \Pos(\gaminf(s))$ has $\gam(u_k)$ as an IO-prefix; and every leaf $i_\ell\in \leaf_i(\gam(u_k))$ defines a unique position in $\gaminf(s)$
in the subtree at position $\nu(\alp')$.
We choose $\nu(\alp)$ as a leaf in $\gam(u_k)$. For that we choose any leaf
$i_j\in \leaf_i(\gam(u_k))$ such that $c_{\rho_k}(i_j)=\psi_i(q)$
where $\psi_i(q)$ refers to the task $\tau_{\rho_k}$ defined in
(\ref{eq:mintaurp}).
%For example, we can choose $i_j$ to be the first one in the length-lexicographical ordering of leaves.
Hence, the minimal color on the path from the root in $\gam(u_k)$ to the chosen leaf $i_j$ induced by the run $\rho_k$  is the greatest among all these colors among all $i_\ell\in \leaf_i(\gam(u_k))$ with respect to best ordering. Since the root
of $\gam(u_k)$ is $\nu(\alp')$ we can define $\nu(\alp)$ by the corresponding position of the chosen position $i_j$.
The subtree of $\gam_\infty(s)$ at position $\nu(\alp)$ has the tree $\gam(u_k.i)$ as an IO-prefix. Hence, $\rho$ induces
a run $\rho(i,q)\in \Run{\gam(u_k.i)}{q}$; and $(u_k.i,q,\rho(i,q))$ is defined by $(i,q)\in I_\alp$. Note that $\alp$ has a child \IFF
$I_\alp\neq \es$. If $I_\alp= \es$ the sequence $\alp$ is a leaf in the corresponding tree.

The definition of the tree of $\bet$-sequences is defined analogously.
We let $k\in \N$ and assume that a sequence $\bet= (u_0,p_0,\rho'_0) \cdots (u_k,p_k,\rho'_k)$
is already defined such that $\rho'_i=\theta(\rho_i)$ for all $1\leq i \leq k$. The mapping $\theta$ was defined in (\ref{eq:rho2'}). By induction, we also know a position $\nu(\bet)\in \Pos(\gam'_\infty(s))$.
It might be that $\bet$ has infinitely many children because we do not choose a particular leaf (as we did for $\alp$-sequences) but have to consider all leaves $i_j$. More precisely,
we define
\begin{equation}\label{eq:Ibet}
I_\bet=\set{(i,q,i_j)\in H\times Q \times \leaf_i(\gam'(u_k))}{i_j \in \leaf_i(\gam'(u_k))\wedge \rho'_k(i_j)= q}.
\end{equation}
The set $I_\bet$ is in a canonical bijection with the union
$\bigcup\set{\leaf_i(\gam'(u_k))}{i\in H}$.
For each $(i,q,i_j)\in I_\bet$ we define the run $\rho'(i,q)= \theta(\rho(i,q))$. (The run $\rho(i,q)$ was defined above for the $\alp$-sequence by choosing a particular leaf.)
Thereby we obtain a triple
$(u_k.i,q,\rho'(i,q))$. For $\bet'= \bet \cdot (u_k.i,q,\rho'(i,q))$ the position $\nu(\bet')\in \Pos(\gam'_\infty(s))$ is given in analogy to the
position $\nu(\alp')$. That is, we consider the position  $\nu(\bet)\in \Pos(\gam'_\infty(s))$. The subtree $\nu(\bet)$ has $\gam'(u_k)$ as an IO-prefix; and we let $\nu(\bet')=\nu(\bet (i,q,i_j))$ be the position which corresponds to the leaf $i_j$.
Since $I_\bet$ is in a canonical bijection with the set of all leaves in $\gam'(u_k)$ labeled by some hole $i\in H$, we see that we actually have defined runs $\rho_n':\Pos(\gam'_n(s)) \to Q$ for all $n\in \N$, hence a run
 $\rho':\Pos(\gam'_\infty(s)) \to Q$.

It remains to show that the run $\rho'$ is accepting. Since it is a run, it is enough to consider infinite paths
$w'$. These paths result from a finite or infinite sequence
$\bet= (u_0,p_0,\rho'_0)(u_1,p_1,\rho'_1)\cdots$.
If the infinite path $w'$ is defined by some finite prefix of $\bet$, then it is accepting\footnote{Accepting paths were defined in \prref{def:parityacc}.} because every finite prefix of $\bet$ defines an accepting path of $\gam'_\infty(s)$ at $p_0$.
So, we may assume that  $\bet$ is infinite, too. Hence, $w'$ can be cut into infinitely many finite paths $w'=w'_0 w'_1 \cdots$ where each $w'_i$ is a finite path inside
$\gam'(u_i)$.
Each finite path belongs to a run $\rho_{v'}'$ of $\gam'(u)$ at some position $v'\in \Pos(\gam'_\infty(s))$. It corresponds to the
some run $\rho_v$ of $\gam(u)$ at some position $v\in \Pos(\gam_\infty(s))$. In this way we obtain an infinite
path in $w=w_0 w_1\cdots$ in $\gam_\infty(s)$ where each
$w_i$ is a path in $\gam(u_i)$. Recall that $\bet$ was defined together with a corresponding sequence
$\alp= (u_0,p_0,\rho_0)(u_1,p_1,\rho_1)\cdots$
such that $\rho'_i=\theta(\rho_i)$. Every path in tree $\rho':\Pos(\gam'_\infty(s))$ can be mapped to
some accepting path in $\rho:\Pos(\gam_\infty(s))$. This situation is depicted in \prref{fig:runprime}.
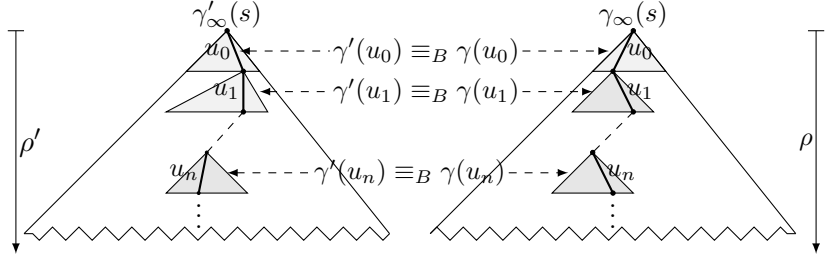
\begin{figure}
\begin{center}
\begin{tikzpicture}[inner sep=1pt,scale=0.54,>=latex]
\draw[fill=gray!10] (3,3.5) -- (2,2.5) -- (3.8,2.5) -- (3,3.5);
\draw[fill=gray!20] (2.5,2.5) -- (1.5,1.5) -- (3.5,1.5) -- (2.5,2.5);
\draw[fill=gray!20] (2,.5) -- (1,-.5) -- (3,-.5) -- (2,.5);
\draw[fill=gray!10] (-7,3.5) -- (-8,2.5) -- (-6.2,2.5) -- (-7,3.5);
\draw[fill=gray!10] (-6.6,2.5) -- (-8.5,1.5) -- (-6,1.5) -- (-6.6,2.5);
\draw[fill=gray!20] (-7.5,.5) -- (-8.5,-.5) -- (-6.5,-.5) -- (-7.5,.5);
\draw[fill] (3,3.5) circle (1.5pt) (2.5,2.5) circle (1.5pt) (3,1.5) circle (1.5pt) (2,.5) circle (1.5pt) (2.5,-.5) circle (1.5pt)
(-7,3.5) circle (1.5pt) (-6.6,2.5) circle (1.5pt) (-6.6,1.5) circle (1.5pt) (-7.5,.5) circle (1pt) (-7.7,-.5) circle (1pt);
\node[below] at (2.5,-.3) {\small $\vdots$};
\draw[thick] (3,3.5) -- node[right]{\small $u_0$} (2.5,2.5) -- node[right]{\small $u_1$} (3,1.5) (2,.5) -- node[right]{\small $u_{n}$}( 2.5,-.5)
          (-7,3.5) -- node[left]{\small $u_0$} (-6.6,2.5) -- node[left]{\small $u_1$} (-6.6,1.5)
          (-7.5,.5) -- node[left]{\small $u_{n}$} (-7.7,-.5);
\node[below] at (-7.7,-.3) {\small $\vdots$};
\node[above] at (3,3.5) {\small $\gam_\infty(s)$};
\node[above] at (-7,3.5) {\small $\gam'_\infty(s)$};
\draw (7,-1.5) -- (3,3.5) -- (-2,-1.5)
          (-3,-1.5) -- (-7,3.5) -- (-12,-1.5);
\draw decorate [decoration=zigzag] {(-2,-1.5) -- (7,-1.5)};
\draw decorate [decoration=zigzag] {(-3,-1.5) -- (-12,-1.5)};
\draw[dashed] (3,1.5) -- (2,.5)  (-7.5,.5) -- (-6.6,1.5);

\node (k1) at (-2.1,3) {\small  $\gam'(u_0)\equiv_B\gam(u_0)$};
\node (k2) at (-2.1,2) {\small  $\gam'(u_1)\equiv_B\gam(u_1)$};
\node (k3) at (-2.5,0) {\small  $\gam'(u_{n})\equiv_B\gam(u_{n})$};

\draw[dashed,<-] (-6.7,3) -- (k1);
\draw[dashed,->] (k1) -- (2.4,3);

\draw[dashed,<-] (-6.2,2) -- (k2);
\draw[dashed,->] (k2) -- (2,2);

\draw[dashed,<-] (-7,0) -- (k3);
\draw[dashed,->] (k3) -- (1.4,0);
\draw (7.3,3.5) -- (7.7,3.5) (-12,3.5) -- (-12.4,3.5);
\draw[->] (7.5,3.5) -- node[left] {\small $\rho$} (7.5,-2);
\draw[->] (-12.2,3.5) -- node[right] {$\rho'$}(-12.2,-2);
\end{tikzpicture}
\end{center}
\caption{The $u_i\in \Pos(s)$ are the top positions of the small triangles
according to \prref{fig:pgnps}. The run $\rho$ defines a run $\rho'$. It satisfies the parity condition because every path in tree $\rho$ has this property.}\label{fig:runprime}
\end{figure}
The task $\tau_v$ defined by $\rho_v$ is satisfied by $\rho_{v'}'$.
Hence, the minimal color seen on the path from the
root in $\gam'(u)$ to a leaf $i_j$ is never greater in the best-ordering
than the minimal color seen on the path from the
root in $\gam(u)$ to a leaf $i_j$. As a consequence, the
minimal color seen infinitely often on the run defined by the
path $w'$ is not greater in the best-ordering than the one defined by $w$. Since the minimal color seen infinitely often defined by $w$ is even, the same is true for the minimal color seen infinitely often defined by $w'$.

Let's repeat the essential steps in the proof  by looking at
\prref{fig:runprime}. A high-level interpretation of the picture shows us the top levels of two trees of possible infinite trees where the  degree of nodes on the right is finite whereas the degree on of nodes on left can be infinite. On the other hand, $\gaminf'(s)$ is just a tree where every vertex has at most $\rmax$ children. In order to see whether $\gaminf'(s)\in L(B,p_0)$, we have to consider all directed infinite paths in $\gaminf'(s)$. The left tree in the picture shows such a path $p'$. This path maps to a unique path in
$\gaminf(s)$ which is chosen such that acceptance is made most difficult, but still the acceptance condition holds because we
$\gaminf(s)$ has an accepting run $\rho$. We use the small triangles on the right to define local runs on the small triangles on the left.
This is possible because $\gam(u)$  and  $\gam'(u)$ satisfy the same set of tasks. The final step is to understand that the local runs on the
left glue together to an accepting run of $\gaminf'(s)$. This was the goal.
\end{proof}

\subsection{The saturation of \texorpdfstring{$\sig$}{sigma}}\label{sec:satsio}
We continue with parity-NTAs
$B=(Q,\Sig,\del,\chi)$ and $B_H= (Q,\Sig \cup H,\del_H,\chi)$
 where, as above, $\del_H=\del\cup (Q\times H)$.

\begin{defi}\label{def:sat}
Let  $\sig:\cX\to 2^{T(\Sig\cup H)\sm H}$ be a \subst.
The  \emph{saturation of  $\sig$} (\wrt~$B$) is the \subst $\wh\sig:\cX\to
2^{T(\Sig\cup H)\sm H}$ defined by
\[\wh\sig(x) = \set{t'\in T(\Sig \cup \cX)\setminus H}{\exists t\in \sig(x): t'\equiv_B t}.\]
\end{defi}
Cleary, $\sig\leq \wh \sig$ %(for the natural partial order by set inclusion for components)
 and $\sig(x)\neq \es\iff\wh \sig(x)\neq \es$.
%The next statement is asymmetric in the NTAs $B_H$ and $B$.
\begin{prop}\label{prop:satsig}
Let $q_0\in Q$ be a state of the NTA $B$ and $L\sse T(\Sig\cup\cX)$ be any subset. Then $\wh\sig(x)$
is regular for all $x\in \cX$. Moreover,
$\sio(L) \sse L(B_H,q_0)\iff
\whsio(L) \sse L(B,q_0)$.
\end{prop}

\begin{proof}
We have $\wh\sig(x) = \bigcup_{\pi\in\set{\pi(t)}{t\in \sig(x)}}\set{t'\in T(\Sig \cup H)}{\pi(t')=\pi}$. This is a finite union. Thus, $\wh\sig(x)$ is regular by \prref{cor:pireg}.
Moreover, if $\whsio(L) \sse L(B,q_0)$, then $\sio(L) \sse \whsio(L) \sse L(B,q_0) \sse L(B_H,q_0)$. Hence, let $\sio(L) \sse L(B_H,q_0)$. Since $\sio(L) \sse T(\Sig)$, we may
suppose $\sio(L) \sse L(B,q_0)$. It suffices to show that, if $s\in T(\Sig\cup\cX)$ and $\sio(s)\sse L(B,q_0)$, then $\whsio(s)\in L(B,q_0)$, that is, $\gam'_\infty(s)\in L(B,q_0)$ for all $\gam'\in \Gam(s,\wh\sig)$.
Let $\gam'\in \Gam(s,\wh \sig)$. Then, by definition of $\wh\sig$, for every
 $u\in \Pos(s)$, where $\sig(s(u))\neq \es$, we can choose a tree $\gam(u)=t\in \sig(s(u))$ such that
 $\gam'(u)\equiv_B \gam(u)$. Thus, $\gam\in\Gam(s,\sig)$ and $\gam_\infty(s)\in L(B,q_0)$ because $\sio(s)\subseteq L(B,q_0)$. By \prref{prop:treg}, we have $\gam'_\infty(s)\in L(B,q_0)$ too.
 Hence, $\whsio(s) \sse L(B,q_0)$ as desired.
 \end{proof}
%The following corollary is another crucial step towards \prref{thm:s1s2main}.
\begin{cor}\label{cor:maxioreg}
Let $\sig_1,\sig_2:\cX \to 2^{T(\Sig\cup H)\sm H}$ be regular \subst{s} such that $\sig_1\leq \sig_2$,
$R\sse T(\Sig)$ be a regular tree language, and $L\sse T(\Sig\cup\cX)$ be any subset. Then there is an effectively computable finite set $S_2$ of regular \subst{s} such that the following property holds. For every $\sig:\cX\to 2^{T(\Sig\cup H)\sm H}$ satisfying both,  $\sig_1\leq \sig\leq \sig_2$ and $\sio(L)\sse  R$ (resp.~$\sio(L)= R$), there is some maximal \subst
$\sig' \in S_2$ such that $\sio'(L)\sse R$
(resp.~$\sio'(L)= R$) and $\sig\leq \sig'\leq \sig_2$.
\end{cor}

\begin{proof}We may assume that  $R=L(B,q_0)$. Throughout the proof,
given any \subst $\sig$, the notation $\wh\sig$ refers to its saturation according to \prref{def:sat}. Moreover, by \prref{prop:satsig}, every  $\wh\sig$ is a regular \subst and
if $\sio(L)\sse R$, then $\sio(L)\sse\whsio(L)\sse R$. Based on these facts, we define in a first step the following set of \subst{s}.
\begin{align}\label{eq:S_0fini}
S_0=\{\wh\sig\mid \sig:\cX\to 2^{T(\Sig\cup H)\sm H}
\text{ is a \subst}\}.
\end{align}
The set $S_0$ is  finite and effectively computable. Its cardinality is bounded by $2^{|\cP_B\times \cX|}$. Since $\sig_1$ and $\wh \sig$ are regular for every
$\wh\sig\in S_0$, we can effectively compute the following subset $S_1$ of $S_0$, too.
\begin{align}\label{eq:S_1eff}
S_1=\{\wh \sig\in S_0\mid \sig:\cX\to 2^{T(\Sig\cup H)\sm H} \wedge \sig_1\leq \sig\}
= \set{\wh \sig\in S_0}{\sig_1\leq \wh \sig}.
\end{align}
The second equation in (\ref{eq:S_1eff})  uses $\wh{\wh\sig}= \wh \sig$. Now, for every  \subst $\sig:\cX\to 2^{T(\Sig\cup H)\sm H}$ satisfying  $\sig_1\leq \sig\leq \sig_2$  we have $\wh\sig\in S_1$. But $\wh \sig\leq \sig_2$ might fail.
However, since $\sig_2$ is regular, we can  calculate for each $\wh\sig\in S_1$ the \subst $\sig'$ such that $\sig'(x) = \wh \sig(x)\cap \sig_2(x)$ for all $x\in \cX$. Hence, the following finite set is effectively computable:
\begin{align}\label{eq:S_2eff}
S_2=\{\sig':\cX\to 2^{T(\Sig\cup H)\sm H}\mid \exists \,\wh\sig \in S_1 \,\forall x\in \cX: \sig'(x)=\wh\sig(x) \cap \sig_2(x)\}.
\end{align}
An easy reflection shows that for all $\sig$ such that $\sig_1\leq \sig\leq \sig_2$ and $\sio(L)\sse  R$ there is some \subst $\sig'\in S_2$ such that $\sig\leq \sig'\leq \sig_2$ and $\sio'(L)\sse  R$.
Since $S_2$ is finite, there is a maximal element $\sig'$ in $S_2$
such that $\sig\leq \sig'\leq \sig_2$ and $\sio'(L)\sse  R$.
The assertion of \prref{cor:maxioreg} follows.
\end{proof}

\begin{rem}\label{rem:mainres}
Note that \prref{cor:maxioreg} does not tell us how to decide
whether there exists any \subst $\sig$ such that $\sio(L)\sse R$ or $\sio(L)=R$.
However, it tells us that if there exists such a \subst $\sig$, then there there exists such a \subst in an effectively computable and finite set of regular \subst{s}. If this set is empty, then $\sio(L)\sse R$ does not hold.
\qed \end{rem}

\subsection{The specialization of \texorpdfstring{$\sig$}{sigma}}\label{sec:specsio}
The purpose of this section is to add for each profile $\pi=\pi(t)$ a specific
finite tree $t_\pi$ which satisfies the same profile and which depends on $\pi$, only.
The roadmap for that is as follows. We  enlarge the set $\Sig$ by various new function symbols in order to obtain a larger ranked alphabet $\Sig_\cP$. \Ip for each $\pi\in \cP_B$ there is a function symbol $f_\pi$ and a  finite tree $t_\pi$ with $t_\pi(\eps)=f_\pi$. The tree $t_\pi$ is either a constant or it has height  $2$.
Simultaneously, we define an extension $B_\cP$ of $B_H$ such that for each $t\in T(\Sig\cup H)$ with $\pi=\pi(t)$ (\wrt the NTA $B_H$) we have
$\pi= \pi(t)=\pi(t_\pi)$ where $\pi(t)$ and $\pi(t_\pi)$ are defined with respect to $B_\cP$.

The application is for deciding whether a regular \subst $\sig: \cX\to 2^{T(\Sig\cup H)\sm H}$ satisfies $\sio(L)\sse L(B,q_0)$.
In the decision procedure, we compute, in a preprocessing phase, for each $x$ the set of profiles
$\Pi_x=\set{\pi(x)}{t\in \sig(x)}$. The set $\Pi_x$ is defined \wrt  $B_H$, but this is the same as being defined \wrt  $B_\cP$. Hence, we can also write
$\Pi_x=\set{\pi(x)}{t\in \check\sig(x)}$ where $\check\sig(x)$ is a finite set of finite trees $t_\pi$ of height at most $2$. The \subst $\check\sig$ yields the
``specialization'' of~$\sig$; and it is enough to decide $\check\sig_\mio(L)\sse L(B,q_0)$.

The formal definitions of $f_\pi$, $\Sig_\cP$ and $B_\cP$ are in \prref{def:trip} and \prref{def:alpProf}. Since every symbol $f_\pi$ has a rank which only depends on $\pi$, but which will be defined through
a tree $t$ where $\pi=\pi(t)$, it is important to note that
$\es\neq \pi(t)=\pi(t')$ implies $\leaf_i(t)\neq \es \iff \leaf_i(t')\neq \es$ for all holes $i\in H$. This follows because
if  $\es\neq \pi(t)$, then $\leaf_i(t)\neq \es$ \IFF for all tasks
$(p,\psi_1\lds \psi_{|H|})\in \pi$ there is some $q\in Q$ such that $\psi_i(q)>0$. Indeed, let $\rho\in \RUN(t,p)$ such that
$\rho \models \tau$. Then for all $i_j\in \leaf_i(t)$ the run $\rho$ yields some state $\rho(i_j)\in Q$. Since $\chi\rho(i_j)>0$
we must have $\psi_i(\rho(i_j))>0$.
\Ip $\es\neq \pi(t)=\pi(t')$ implies $H(t)=H(t')$. Thus, for
$\pi\in \cP$: if $\pi = \pi(t)$, then
$H(t)$ depends on $\pi$ but not on the chosen $t$.
\begin{defi}\label{def:trip}
\begin{enumerate}
\item For each $t\in T(\Sig\cup H)$ we denote by $H(t)$ the subset of holes
$i\in H$ such that $\leaf_i(t)\neq \es$.
\item The alphabet $\Sig_{\Pro}$ contains $\Sig$ and it  contains, in addition,
first a new function symbol $\$_i$ of rank~$1$ for each  $i\in H$, and second, for each   $t\in T(\Sig\cup H)$ and each profile
$\es \neq \pi=\pi(t)\in \cP$ there is a new function symbol $f_\pi$ of rank $\abs{Q\times H(t)}$.
\Ip if $t$ is without holes, then $f_\pi$ is a constant.
We also include a new constant $f_\es$ in case $\es\in \cP_B$ (but we will make sure that
no run at any state of $B_\cP$ accepts $f_\es$).
\item For each $\es\neq \pi\in \cP$ we choose any tree $t\in T(\Sig \cup H)$ such that $\pi=\pi(t)$. Then
we define the finite tree
$t_{\pi}$ which has the following structure. Its root is labeled by the symbol $f_\pi$,
If~$t$ has no holes, then~$f_\pi$ is a constant. In the other case it has height $2$. The children of the root are labeled by $\$_i$ where $i\in H(t)$; and if a vertex is labeled
by  $\$_i$, then its children are holes $i_j$ (and thus labeled by $i\in H(t)$).
There are $k= \abs{Q\times H(t)}$ holes, all of them belong to
$H(t)$, and each hole $i\in H(t)$ appears exactly $\abs{Q}$ times (in some fixed order). The corresponding tree is depicted in \prref{fig:RunProf}. For $H(t)\neq \es$ we have $k\geq 1$. For $H(t)= \es$ we have $k=0$ and then the figure shows the constant $f_\pi$.
\end{enumerate}
\end{defi}

\noindent
We are ready to define an extension $B_\cP$ of the automaton $B_H$
such that for all $t\in T(\Sig \cup H)$ the finite tree $t_\pi$ satisfies $t\equiv_{B_\cP} t_\pi$ \IFF $\pi(t) =\pi$.
\begin{defi}\label{def:alpProf}
We let $B_{\Pro}= (Q',\Sig_{\Pro}\cup H, \del_{\Pro},\chi')$ denote the following parity-NTA\@. It is an extension of  $B_{H}$.
 The set of states is $Q'= Q\cup (C\times Q \times H)$.
 The coloring $\chi$ is extended to states in $C\times Q \times H$ by $\chi'((c,q,h)) = c$.

The set of \tra{s} $\del_\cP$ of $B_{\Pro}$ contains all \tra{s} from $B_H$ (that is: $\del_{H}= \del \cup Q\times H$) and,  in addition, {for each $\es \neq \pi(t)\in \cP_B$} the following set of tuples.
\begin{itemize}
\item Transitions $\big(p,f_{\pi},(p_{s,j})_{(s,j)\in Q \times H(t)}\big)$
for $\tau= (p,\psi_{1}\lds \psi_{|H|})\in \pi$  such that
\begin{align}\label{eq:psj}
\set{p_{s,j}}{(s,j)\in Q \times H(t)}= \set{(\psi_i(q),q,i)}{q\in Q, i \in H(t), \psi_i(q)\neq 0}.
\end{align}
The sequence $(p_{s,j})_{(s,j)\in Q \times H(t)}$ is a tuple of length $\rk(f_\pi)$ and each state $p_{s,j}$ is a triple such that \prref{eq:psj} holds. The tuple allows repetitions of
entries $p_{s,j}$.
 \item Transitions $\big((c,q,i),\$_i,q\big)$ for all $(c,q,i)\in C \times Q \times H(t)$.
\end{itemize}
\end{defi}
\begin{figure}[h]
\centerline{\begin{tikzpicture}[node distance=8mm]
		\node[] (tpi) at (-3,-0.6) {$t_{\pi}$};
		\node[] (tfpi) at (-2.5,-0.6) {$=$};
		\node[] (fpi) at (0,-0.6) {$f_{\pi}$};
		%\node[] (gtau) at (0,-1) {$g_{\tau}$};
		\node[] (v1) at (-2,-1.5)
		{$\$_{h_1}$}; %{$c_{(\psi_{h_{1}}(q_{1}),q_{1},h_{1})}$};
		\node[] (v) at (0,-1.5) {$\cdots$};
		\node[] (vk) at (2,-1.5)
		{$\$_{h_k}$}; %{$c_{(\psi_{h_{k}}(q_{k}),q_{k},h_{k})}$};
		\node[] (i1) at (-2,-2.5) {\qquad $h_{1}\in H(t)$};
		\node[] (ik) at (2,-2.5) {\qquad $h_{k}\in H(t)$};
		\node[] (w) at (0,-2.4) {$\cdots$};
		%\draw [thick,-, >=latex] (fpi) to (gtau);
		\draw [thick,-, >=latex] (fpi) to (v1);
		\draw [thick,-, >=latex] (fpi) to (vk);
		\draw [thick,-, >=latex] (v1) to (i1);
		\draw [thick,-, >=latex] (vk) to (ik);
\end{tikzpicture}
\qquad
\begin{tikzpicture}[node distance=8mm]
		\node[] (fpi) at (0,-0.6) {$p$};
		\node[right=3pt] at (fpi) {$\in Q$};
		%\node[] (gtau) at (0,-1) {$\tau$};
		\node[] (v1) at (-2,-1.5)
		{$(\psi_{h_{1}}(q_{1}),q_{1},h_{1})$};
		\node[] (v) at (0,-1.5) {$\cdots$};
		\node[] (vk) at (2,-1.5)
		{$(\psi_{h_{k}}(q_{k}),q_{k},h_{k})$};
		\node[] (i1) at (-2,-2.5) {\qquad $q_1\in Q$};
		\node[] (ik) at (2,-2.5) {\qquad $q_k\in Q$};
		\node[] (w) at (0,-2.4) {$\cdots$};
		\draw [thick,-, >=latex] (fpi) to (v1);
		\draw [thick,-, >=latex] (fpi) to (vk);
		\draw [thick,-, >=latex] (v1) to (i1);
		\draw [thick,-, >=latex] (vk) to (ik);
\end{tikzpicture}
}
	\caption{The tree $t_{\pi}$ where $k=\rk(f_\pi)$ and an accepting run $\rho$ on  $t_{\pi}$
	for $\tau= (p,\psi_{1}\lds \psi_{|H|})\in \pi$.}\label{fig:RunProf}
\end{figure}
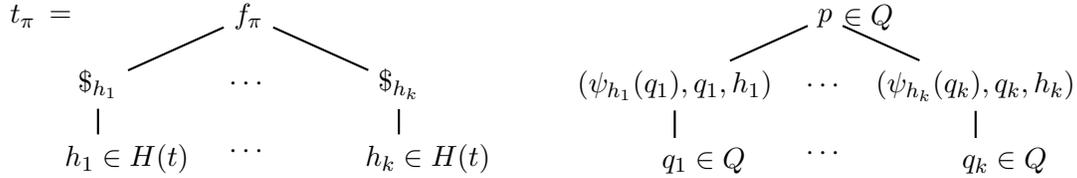
%%%%%%%%%%%%

\noindent
The cardinality of the set $\set{(\psi_i(q),q,i)}{q\in Q, i \in H(t), \psi_i(q)\neq 0}$ in (\ref{eq:psj}) varies.  It can be any number
between $\abs{H(t)}$ and $\rk(f_\pi)$; and it is not necessarily the same for all $\tau\in \pi$. This is the reason, why
triples $(\psi_i(q),q,i)$ may occur several times. Note that for every $\tau\in\pi$ there exists some \tra
$\big(p,f_{\pi},(p_{s,j})_{(s,j)\in Q \times H(t)}\big)$ in $\del_\cP$.
 \begin{lem}\label{lem:BProfPi}
Let $t\in T(\Sig \cup H)$,
$\pi=\pi(t)\in \cP_B$, and $t_\pi$ as defined above. Then
for all tasks $\tau\in \cT_{B_\cP}$ we have
$t_{\pi}\models \tau \iff \tau \in \pi$. \Ip
$t\equiv_{B_\cP} t_\pi$ and $\pi(t_\pi)=\pi$.
\end{lem}
\begin{proof}
For each $\tau= (p,\psi_{1}\lds \psi_{|H|}) \in \pi$ the NTA $B_\cP$ admits for $t_\pi$ a successful run at $p$. This is clear
for $H(t)=\es$ because then $t_\pi$ is a constant and
$(p,t_\pi) \in \del_\cP$. For $H(t)\neq \es$ the accepting run is depicted on the right of
 \prref{fig:RunProf}. It is accepting because $Q\times H \sse \del_H \sse \del_\cP$. Hence,
 $\tau \in \pi$
implies $t_{\pi}\models \tau$. The converse follows from the definition of $\del_\cP$. \Ip $t_{\pi}\models \tau$ implies
$\tau \in \cT_B$.
Since $\Sig \sse \Sig_\cP$ and $Q\sse Q'$, we have  $t\in T(\Sig_\cP\cup H)$ and $\cT_B\sse \cT_{B_\cP}$ and therefore,
$t\equiv_{B_\cP} t_\pi$.
\end{proof}
\prref{prop:redBcP} reduces the computation of $\sig_{\mathrm{io}}^{-1}(R)$ to the computation of  $\check\sig_{\mathrm{io}}^{-1}(R)$.
Here, $\check\sig$ is the specialization of $\sig$ according to
\prref{def:specsat} such that $\check\sig(x)$ is a subset of the finite
set of finite trees $\set{t_\pi}{\pi \in \cP_B}$.
\begin{defi}\label{def:specsat}
 Let $\sig:\cX\to 2^{T(\Sig \cup H)}$ be a \subst and $\pi\in \Pro$ be a profile. The \emph{specialization} $\check \sig:\cX \to 2^{T(\Sig_{\Pro} \cup H)}$ \wrt $B_{\Pro}$ is defined by the \subst
\begin{equation}\label{eq:spec}
\check \sig(x)=\set{t_{\pi} \in T(\Sig_{\Pro}\cup H)}{\exists t \in \sig(x):\, \pi=\pi(t)}.
\end{equation}
Here, $t_\pi$ is the tree defined according to \prref{def:trip}.
That is, if $\pi=\pi(t)$ for some $t$ without holes, then $t_\pi = f_\pi$. Otherwise $k\geq 1$ and $t_\pi$ has height two according
to \prref{fig:RunProf} on the left-side.
\end{defi}
Note that $\sig(x)=\es \iff \check\sig(x)=\es$ and $t_\es\in \check \sig(x) \iff \exists t \in \sig(x):\, t\models \es$.

\begin{prop}\label{prop:redBcP}
Let $R=L(B,q_{0})\sse T(\Sig)$ and let
$\Sig_\cP$ and $B_\cP=(Q',\Sig_\cP\cup H, \del_\cP, \chi')$ according to \prref{def:alpProf}.
Let
$\sig: \cX \to 2^{T(\Sig\cup H)\sm H}$ be any  \subst and
$\check \sig: \cX \to 2^{T(\Sig_\cP\cup H)}$ its specialization according to \prref{def:specsat}.
Then we have
\begin{equation}\label{eq:checksoi}
\set{s\in T(\Sig\cup \cX)}{\sio(s)\sse L(B,q_{0})}
=\set{s\in T(\Sig\cup \cX)}{\vsig_{\mathrm{io}}(s)\sse L(B_{\Pro},q_{0})}.
\end{equation}
\end{prop}
\begin{proof}
The equality
\[\set{s\in T(\Sig\cup \cX)}{\sio(s)\sse L(B,q_{0})}= \set{s\in T(\Sig\cup \cX)}{\whsio(s)\sse L(B_{\Pro},q_{0})}\] follows directly
{}from the definition of $B_{\Pro}$. We apply
\prref{prop:satsig} twice, once to $B$ and once to~$B_{\Pro}$. We obtain
$\set{s\in T(\Sig\cup \cX)}{\sio(s)\sse L(B_{\Pro},q_{0})}=
\set{s\in T(\Sig\cup \cX)}{\whsio(s)\sse L(B_{\Pro},q_{0})}$ as well as
$\set{s\in T(\Sig\cup \cX)}{\wh\vsig_{\mathrm{io}}(s)\sse L(B_{\Pro},q_{0})}
=\set{s\in T(\Sig\cup \cX)}{\vsig_{\mathrm{io}}(s)\sse L(B_{\Pro},q_{0})}$.

Note that $\whsio(s)$ refers to $\wh\sig(x)= \set{t'\in T(\Sig_{\Pro}\cup H)}{\exists t \in \sig(x): t'\equiv_{B_{\Pro}}t}$. That is to
the saturation with respect to
terms in $T(\Sig_{\Pro}\cup H)$.
By \prref{lem:BProfPi}, we have $t_{\pi} \equiv_{B_\cP} t$ for every $t\in \sig(x)$ and  profile $\pi=\pi(t)$. Hence,
$\set{s\in T(\Sig\cup \cX)}{\whsio(s)\sse L(B_{\Pro},q_{0})}
=\set{s\in T(\Sig\cup \cX)}{\wh\vsig_{\mathrm{io}}(s)\sse L(B_{\Pro},q_{0})}$. Putting everything together we see
\begin{align*}
\set{s\in T(\Sig\cup \cX)}{\sio(s)\sse L(B,q_{0})}&
=\set{s\in T(\Sig\cup \cX)}{\whsio(s)\sse L(B_{\Pro},q_{0})}\\
=\set{s\in T(\Sig\cup \cX)}{\wh\vsig_{\mathrm{io}}(s)\sse L(B_{\Pro},q_{0})}
&=\set{s\in T(\Sig\cup \cX)}{\vsig_{\mathrm{io}}(s)\sse L(B_{\Pro},q_{0})}.
\end{align*}
The last line uses \prref{prop:equivtasl}. The result follows.
\end{proof}

\section{Parity games}\label{sec:pg}
A \emph{directed (multi-)graph} is a tuple $G=(V,E,\sr,\tr)$ where $V$ is the set of vertices and
the sets $\sr$ and $\tr$ are functions from $E$ to $V$. Here, $\sr(e)$ denotes the source of $e$ and $\tr(e)$ denotes the target of $e$.
In our paper an \emph{arena} is any tuple $A=(V_0,V_1,E,\sr,\tr)$ such that $(V,E,\sr,\tr)$ is a directed graph where  $V=V_0\cup V_1$ and $V_0,V_1$ are disjoint.

Phrased differently, every vertex is in exactly one subset $V_i$ of $V$ and we allow multiple edges between vertices.
The maps $\sr$ and $\tr$ are extended to all paths of the graph $A$ in the natural way:
if $w= p_0,e_{1},p_{1}, \ldots,e_{m},p_{m}$ is a path of $A$ where $m\geq 0$ and $p_{i-1}=\sr(e_i)$, $p_{i}=\tr(e_i)$ for all $1\leq i \leq m$, then we let $\sr(w) = p_0$ and $\tr(w) = p_m$.

A \emph{(board of a) parity game}  is defined by a pair $(A,\chi)$ where $A$ is an arena and
$\chi:V\to C$ is a mapping to a set of \emph{colors} $C$.

Without restriction, we always assume that
$C=\os{1\lds |C|}$ and that $\abs C$ is odd. Let $p_0\in V$. {A \emph{game at $p_0$} is a finite or infinite sequence
$p_0,e_1,p_1,e_2 \ldots$ such that
$p_{i-1}=\sr(e_1)$, $p_{i}=\tr(e_i)$ and $e_i\in E$ for $i\geq 1$. }Moreover,  we require that a game is infinite unless it ends in a sink. A \emph{sink} is a vertex without outgoing edges\footnote{Some papers require that an arena has no sinks.}.
There are two players: $P_0$ (\emph{Prover}) and $P_1$ (\emph{Spoiler}). The rules of the game are as follows. It starts in the vertex $p_0$. Let  $m\geq 0$ such that a path $p_0,e_{1},p_{1}, \ldots,e_{m},p_{m}$ is already defined with $p_m\in V_i$. If $p_m$ is a sink, then player $P_i$ loses. In the other case
 player $P_i$ chooses an edge  $e_{m+1} \in E$ such that
 $p_{m}= \sr(e_{m+1})$. Let $p_{m+1}=\tr(e_{m+1})$, then the game continues with
 $p_0, \ldots,p_m,e_{m+1},p_{m+1}$.

If the game does not end in a sink, then the mutual choices define an infinite sequence. Prover $P_0$ wins an infinite game if the least color which appears infinitely often in $\chi(p_0),\chi(p_1), \ldots$  is even. Otherwise, Spoiler $P_1$ wins that game.

A  \emph{positional strategy} for player $P_i$ is a subset $E_i\sse E$ such that for each $u\in V_i$ there is at most one edge $e\in E_i$ with $\sr(e)=u$.
Hence, every  $e\in E_i$ is an edge of the arena. However, the property that $E_i$ is
\emph{winning} depends on the pair
$(\sr(e),\tr(e))$, only.  Therefore,  we can assume $E_i\sse V_i\times V$.
Each positional strategy defines a subarena $A_i= (V, E_i\cup \set{e\in E}{\sr(e)\in V_{1-i}})$.
In the arena $A_i$ player  $P_{1-i}$ wins at a vertex $p$ \IFF there exists some path starting at $p$ which satisfies his winning condition defined above.
Note that $W_i$ depends on the positional strategy $E_i$.
The set $W_i$ contains all vertices where player~$P_i$ wins positional (also called \emph{memoryless}) by choosing $E_i$. The set $W_i$ is a set of \emph{winning positions} for $P_i$ in the original arena $A$,
because player $P_i$ is able to win, no matter how $P_{1-i}$ decides for $p_{m} \in V_{1-i}$ on the next outgoing edge $e \in E$ with $\sr(e)=p_{m}$.

\begin{thmC}[\cite{GurevichH82stoc}]\label{thm:GH82}
There exist positional strategies $E_i= V_i\times V\sse E$ for both players $P_i$ such that their sets
of winning positions $W_i\sse V$  (\wrt the subarenas $A_i$)
satisfy $W_{1-i}=V \sm W_{i}$. That is, $V$ is the disjoint union of $W_{1-i}$ and $W_{i}$.
\end{thmC}
The theorem implies that for parity games there is no better strategy than a positional one. \prref{thm:GH82} is due to Gurevich and Harrington. Simplified proofs are \eg in~\cite{LNCS2500automata,zielonka98tcs}.

\subsection{Alternating tree automata}%
\label{sec:rtlata}
Let $\Del$ be a finite ranked alphabet with the rank function $\rk:\Del\to \N$. Nondeterministic tree automata are special instances of alternating tree automata. ATAs were introduced in~\cite{MullerSchupp87tcs} and further investigated in~\cite{MullerSchupp95tcs}.
A \emph{parity-ATA} for $\Del$ is a tuple $A=(Q,\Del,\del,\chi)$ where
$Q$ is a finite set of states,
$\del$ is the \emph{\tra function}, and
$\chi:Q\to C$ is a coloring with $C=\os{1,\lds \abs C}$ and where, without restriction, $\abs C$ is odd. For each $(q,f)\in Q\times \Del$ there is exactly one \tra which has the form $\big(q,f,\Phi)$ where
$\Phi$ is a positive Boolean formula over the set $[\rk(f)]\times Q$ and, moreover, $\Phi$  is
written in disjunctive normal form\footnote{Let $S$
be any (finite) set. Then the set of \emph{positive Boolean formulae} $\B_+(S)$ is defined inductively. 1.~The symbols $\bot$ and $\top$ and
all $s\in S$ belongs to $\B_+(S)$. (The elements of $S$ are the \emph{atomic propositions}.) 2.~If $\phi,\psi\in \B_+(S)$, then $(\phi\vee\psi)\in \B_+(S)$
and $(\phi\wedge\psi)\in \B_+(S)$.}.
This means that $\big(q,f,\Phi)\in \del$ is written as
\begin{align}\label{eq:ataPhi}
\big(q,f,\bigvee_{j\in J}\;\bigwedge_{k\in K_{j}} (d_{k},p_k)\big)
\end{align}
where $J$ and $K_{j}$ are finite  sets and $(d_{k},p_k)\in [\rk(x)]\times Q$. The first component $d_k\in [\rk(f)]$ is also called a \emph{direction}. With each $j\in J$ there is an associated finite set of indices $K_j$.
We use a syntax for Boolean formulae defined by a context-free grammar, but in the notation redundant brackets are typically removed as we did in (\ref{eq:ataPhi}).
The syntax still allows repetitions of pairs $(d_{k},p_k)$. For example, we may have $K_j=\os{k,\ell}$ with a conjunction
$\big((d_{k},p_k) \wedge (d_{\ell},p_\ell)\big)$ where $(d_{k},p_k) =(d_{\ell},p_\ell)$. This will have no influence on the semantics which we define next, but it will introduce
``multiple edges''  in the corresponding game-arena,
which will make the set of choices for Spoiler larger.\footnote{This larger set of choices is explicitly used in the definition of $w_g$ in the proof of Lemma~\ref{lem:parthom}.}

By definition, let $p\in Q$, then $L(A,p)$ is the set of trees $s\in T(\Del)$ such that Prover $P_{0}$ has a winning strategy
for the following parity-game  in the arena $G(A,s)$ at vertex $(\eps,p)$.

The  set of vertices belonging to the Prover of the arena $G(A,s)$ is $V_0=\Pos(s)\times Q$.
The color of $(u,q)$ is $\chi(q)$.
Next we define the set $E_0$ of outgoing edges at $(u,q)\in V_0$ by the set of all
$(u,q,j)$ with $j\in J$ and where $J$ appears under the disjunction in (\ref{eq:ataPhi}) of the unique \tra $\big(q,s(u),\bigvee_{j\in J}\;\bigwedge_{k\in K_{j}} (d_{k},p_k)\big)\in \del$.
We let $\sr(u,q,j)=(u,q)$ and we let $\tr(u,q,j)=(u,q,j)$.
We define the set of positions belonging to
Spoiler by  $V_1=\tr(E_0)$ and we let $\chi(u,q,j)=\chi(q)$. Thus,  the target function $\tr:E_0 \to V_1$ is surjective by definition and every vertex in $V_1$ has an incoming edge.
 Note that for $(u,q)\in V_0$ and any $v\in V_1$ there is at most one edge $e\in E$ with $\sr(e)=(u,q)$ and $\tr(e)= v$. No multiple edges occur here. This changes when we define the outgoing edges
for vertices  $(u,q,j)\in V_1$. The outgoing edges
are the quadruples $(u,q,j,k)$ where $k\in K_j$ for index sets under the conjunction
in the \tra $\big(q,s(u),\bigvee_{j\in J}\;\bigwedge_{k\in K_{j}} (d_{k},p_k)\big)\in \del$. We define $\sr(u,q,j,k)=(u,q,j)$.
 The index $k$ defines a pair
$(d_k,p_k)$ and then we let $\tr(u,q,j,k)=(u.{d_k},p_k)$. Thus, there can be many edges with the
source $(u,q,j)\in V_1$ and target $(u.{d_k},p_k)$. Thanks to the condition $(d_{k},p_k)\in [\rk(x)]\times Q$ we are sure that the
position $u.{d_k}$ exists as soon as  $K_j\neq \es$.

Let us phrase the definition of the arena from the perspective of the players if the game starts at some vertex $(u,q)\in V_0$.
Then Prover chooses, if possible,  an index $j\in J$ according to set $E_0$.
If $J$ is empty, then Prover loses. (An empty disjunction is ``false''.) In the other case,  it is the turn of Spoiler~$P_1$, and the game continues at the vertex $(u,q,j)$ belonging to $P_{1}$ with color $\chi(q)$. If $K_{j}$ is empty, then Spoiler loses. (An empty conjunction is ``true''.) Otherwise, Spoiler chooses an index $k\in K_{j}$ and the game continues at the vertex $(u.d_{k},p_{k}) \in V_0= \Pos(s)\times Q$ (which exists). That is, the game continues at the position of the $d_{k}$-th child of $u$ at state $p_{k}$. Prover wins an infinite game \IFF the least color occurring infinitely often is even.

A \emph{parity-NTA} $A$ is a special instance of an alternating automaton, where every \tra as in \prref{eq:ataPhi} has the special form $\big(q,f,\bigvee_{j\in J}\;\bigwedge_{i\in [\rk(f)]} (i,p_i)\big)$.
For convenience we use only the traditional syntax that $\del$ is a set of tuples
\begin{equation}\label{eq:nta}
\big(q,f,p_{1}\lds p_{\rk(f)}).
\end{equation}

The main results of alternating tree automata can be formulated as follows.

\begin{thmC}[\cite{MullerSchupp87tcs,MullerSchupp95tcs}]\label{thm:MS8795}
Let $A$ be a parity-NTA as defined in  \prref{sec:parnta} and $p$ be a state. Then the following assertions hold.
\begin{enumerate}
\item Viewing $A$ as a special instance of a parity-ATA using
(\ref{eq:ataPhi})
or using the traditional syntax (\ref{eq:nta}) and its semantics according to \prref{def:parityacc} yields the same set $L(A,p)$.
\item Parity-ATAs characterize the class of regular tree languages: if $L(A,p)$ is defined by a parity-ATA $A$ at a state $p$, then  we can construct effectively a
parity-NTA $B$ and a state $q$ such that
$L(A,p)= L(B,q)$.
\end{enumerate}
\end{thmC}

\subsection{Alternating tree automata for languages of type \texorpdfstring{$\oi \sio(R)$}{sigma-io\^-1(R)}}\label{sec:altfin}
Thanks to \prref{prop:redBcP} we know that $\sio(L)\sse L(B,q_0) \iff \check\sig_\mio(L)\sse L(B,q_0)$ where $\check\sig_\mio$ is the specialization of $\sig$. Since
$\check\sig_\mio(x)$ is, by construction, a finite set of finite trees,
it is enough to consider the concept of alternating tree automata in the setting where there is a finite set $T\sse \Tfin(\Sig\cup H)\sm H$ of finite trees
such that $\bigcup\set{\sig(x)}{x\in \cX}\sse T$. This is not essential but allows to use the traditional  definition in
(\ref{eq:ataPhi}) for \tra{s} where the index sets $J$ and $K_j$ are assumed to be finite.

\begin{rem}\label{rem:homchoice}
Recall that a choice function is defined for a \subst and yields an element in the cartesian product
$\prod_{s\in T(\Sig \cup \cX)}T_\bot(\Sig \cup H)^{\Pos(s)}$.
Let us show that the set of  mappings $\cX\to T(\Sig \cup H)$ can be viewed as a subset of $\prod_{s\in T(\Sig \cup \cX)}T_\bot(\Sig \cup H)^{\Pos(s)}$ since there is a natural embedding
\begin{equation}\label{eq:homchoice}
\iota: T(\Sig \cup H)^\cX \to
\prod_{s\in T(\Sig \cup \cX)}T(\Sig \cup H)^{\Pos(s)}
\end{equation}
Indeed, if $\phi\in T(\Sig \cup H)^\cX$ and $x=s(u)$, then $\iota(\phi)(x)$
is defined canonically by $\iota(\phi)(s(u))= \phi(x)$ for
$x\in \cX$ and by $\iota(\phi)(s(u))= f(1\lds \rk(f))$ for
$x=f\in \Sig\sm \cX$. It is also clear that
$\iota$ is an embedding since $x(1\lds \rk(x))\in T(\Sig \cup \cX)$.
Moreover, if  $\phi\in T(\Sig \cup H)^\cX$ is a \hom and
$\gam=\iota(\phi)$, then $\phi_\mio(s)= \os{\gaminf(s)}$ is a singleton.
\qed
\end{rem}
The proof of \prref{lem:parthom} is rather involved, but it does not use any result which was not stated or proved above. The lemma is the cornerstone to understand and to show (via \prref{lem:SigProf}) our main result: \prref{thm:s1s2main}.
\begin{lem}\label{lem:parthom}
Let $R\sse T(\Sig)$ be a regular tree language and $\phi: \cX\to {\Tfin(\Sig\cup H)\sm H}$ be a \hom to the set of finite trees.
Let $\gam=\iota(\phi)$ be its canonically associated choice function as defined
in (\ref{eq:homchoice}).
Then
the set of trees
\[\phi_{\mathrm{io}}^{-1}(R)=\set{s\in T(\Sig\cup \cX)}{\gaminf(s)\in R}\]
is effectively regular.
\end{lem}
\begin{proof}
We have $R =L(B,q_0)$ for some parity-NTA $B=(Q,\Sig,\del,\chi)$ and $q_0\in Q$.
Thus, it is enough to show that  the set $\set{s\in T(\Sig\cup\cX)}{\gaminf(s)\in L(B,p)}$
is effectively regular for all $p\in Q$.
As usual, we assume that $\chi:Q\to \os{1\lds|C|}$ where $|C|$ is odd. Let us define another \pATA $A=(Q\times C,\Sig\cup\cX,\del_A,\chi_A)$ (where
$\chi_A(q,c)=\pr2(q,c)=c$) such that
for all $p\in Q$ we have
\begin{align}\label{eq:AltB}
\set{s\in T(\Sig\cup\cX)}{\gaminf(s)\in L(B,p)}=L(A,(p,\chi(p))).
\end{align}
The set  $\del_A$ is defined by the following \tras
\begin{align}\label{eq:ataph}
\big((p,c),x,\bigvee_{\rho\in \RUN_{B_H}(\gam(x),p)}\;
\bigwedge_{i_j\in K_{\rho}}
\big(i,(\rho(i_j),c_\rho(i_j))\big)\big)
\end{align}
The notation in (\ref{eq:ataph}) is as follows: $x\in \Sig \cup \cX$ with  $\gam(x)=x(1\lds\rk(x))$ and
\begin{align}\label{eq:Krho}
K_\rho= \{i_j\in \Pos(\gam(x))\mid \exists \, i\in H:\, i_j\in \leaf_i(\gam(x))\}
\end{align}
Recall that for a tree $t\in T(\Sig \cup H)$ and a run
$\rho\in \RUN_{B_H}(t,p)$,
the color $c_\rho(i_j)$ denotes the minimal color in the tree
$\rho(t)$ which appears on the unique path from the root $\eps$ to the leaf~$i_j$. If $\RUN_{B_H}(\gam(x),p)=\es$, then
the disjunction is over the empty set, hence
the \tra in (\ref{eq:ataph}) becomes $(p,x,\textrm{false})$.
If $\RUN_{B_H}(\gam(x),p)\neq\es$ but $\gam(x)$ is a tree without holes, then the \tra in (\ref{eq:ataph}) becomes $(p,x,\textrm{true})$.

Using \prref{thm:MS8795} it is enough to show that for all $p\in Q$ we have
$\gaminf(s)\in L(B,p)$ \IFF
Prover $P_0$ has a winning strategy in the associated
parity game for the arena $G(A,s)$ at vertex $(\eps,(p,\chi(p)))$ where $A$ is the \pATA above.
Note that the index set $K_\rho$ allows Spoiler to pick any hole which appears in $\gam(x)$.

One direction is easy. If $\gaminf(s)\in L(B,p)$, then there exists an accepting run $\rho_0\in \RUN_{B}(\gaminf(s),p)$ and Prover can react to all choices of Spoiler
by choosing  $\rho=\rho_{x,p}$ in the disjunction in (\ref{eq:ataph}) as parts of the global run $\rho$.
Note that due to the second component $c_\rho(i_j)$ in the
states, every infinite game $\alp$ corresponds to a unique infinite directed path
in $\rho_0(\gaminf(s))$ such that
the minimal color occurring infinitely often on that infinite path is the same color as the minimal color occurring infinitely often in game $\alp$. Prover's winning strategy can also be explained by looking at the right tree depicted in \prref{fig:runprime}. Prover begins by choosing the partial run of $\rho_0$ on the top triangle. Spoiler has no other option than to choose some leaf $i_j$. This leads to the second triangle, and Prover chooses again the partial run according to $\rho_0$ etc. Since $\gaminf(s)\in L(B,p)$, Spoiler cannot win.

For the other direction, let us assume that Prover has a winning strategy in the arena $G(A,s)$ at vertex $(\eps,(p,\chi(p)))$.
Then Prover~$P_0$ has a positional winning strategy by~\cite{GurevichH82stoc}.

We will show that this winning strategy of $P_0$ defines an accepting run
\[
    \rho_0\in \RUN_{B}(\gaminf(s),p).
\]
The positional strategy chosen by Prover $P_0$
defines for each vertex belonging to $P_0$ at most one
outgoing edge, all other outgoing edges at that vertex are deleted from the arena. After the deletion of edges, Prover is released and without any further interaction in the game. The game becomes a solitaire game where the outcome  depends only on the choices of Spoiler.

Henceforth, all games are defined in a subarena $G'(A,s)$ where
all games at vertex $(\eps,(p,\chi(p)))$ are won by Prover. Moreover, we may assume that
all vertices in $G'(A,s)$ are reachable from the vertex $(\eps,(p,\chi(p)))$.
 Thus, without restriction, every vertex in $G'(A,s)$
belonging to Prover has exactly one outgoing edge. \Ip since all nodes are reachable from $(\eps,(p,\chi(p)))$, all finite games starting at any vertex in $G'(A,s)$ are won by Prover, and if $\alp$  is any directed infinite path starting at any vertex in $G'(A,s)$, then the minimal color occurring infinitely often at vertices of $\alp$ is even.

Whenever a game reaches a vertex $(u,(q,c))$  via a finite directed path in $G'(A,s)$ starting at $(\eps,(p,\chi(p)))$, then $\RUN_{B_H}(\gam(s(u)),q)\neq \es$, and the unique outgoing edge defines a run $\rho\in \RUN_{B_H}(\gam(s(u)),q)\neq \es$ which depends on $(u,q,c)$, only. For better reading we denote the run~$\rho$ as $\rho(u,q,c)$, too\footnote{Thus, the same letter $\rho$ denotes a function or the specific value of that function. The context makes clear what we mean.}.
Note that, if $(q,c)= (\rho(i_j),c_\rho(i_j))$, then the run $\rho(u,q,c)$ does not reveal the position of the leaf $i_j$, in general. That implies that $P_0$ choses $\rho(u,q,c)$ without knowing $i$ or $i_j$, in general.
The outgoing edge (which defines $\rho$) ends at the vertex $(u,(q,c),\rho)$ belonging to spoiler.

In the following we assume without restriction that all games $\alp$ start at vertices belonging to $P_0$. Moreover, if $\alp$ is finite, then  $\alp$ ends in a vertex belonging to Spoiler.

Our aim is construct an accepting run of $\gaminfs$. The run is constructed top-down by induction.
We use the following notation.
Let $u\in \Pos(s)$. We say that
$\rho\in \RUN_{B_H}(\gam(s(u)),q)$ is a \emph{partial run} of the
subtree $\gaminf(s|_u)$ if the run labels a prefix closed subset
of positions in $\gaminf(s|_u)$ with states. The partial run labels the root with $q$. (Recall that $(s|_u)$ denotes the subtree of $s$ rooted at $u$.)
A \emph{decision sequence} for Spoiler is a finite or infinite sequence
$g(\alp)= (g_1,g_2,\ldots)$ where each $g_\ell$ is of the form
$g_\ell= ({h_\ell})_{j_\ell}$ with $h_\ell\in H$ for all $\ell \geq 1$ such that  $g(\alp)$ records all choices of Spoiler in the game $\alp$. The full information about the solitaire game $\alp$ in $G'(A,s)$ starting at $(\eps,(p,\chi(p)))$ is encoded in the sequence $g(\alp)$.

Let $\cD(s)$ be the set of all decision sequences of Spoiler.
In the next step, we define four items for every finite prefix $g=(g_1,\ldots, g_k)$ of a decision sequence $g(\alp)$ in $\cD(s)$.
\begin{itemize}
\item A path $w_g$ in $\gaminf(s)$.
The other items are defined through the path $w_g$.
\item A position $v(w_g)$ in $\gaminf(s)$.
\item A ``terminal'' position $\tr(w_g)=(u,(q,c))$ in $V_0=\Pos(s)\times (Q\times C)$ such that the tree $\gam(s(u))$ appears
as an IO-prefix of the subtree $\gaminf(s|_u)$.
\item For $\tr(w_g)=(u,(q,c))$ a partial run $\rho_g\in \RUN_{B_H}(\gam(s(u)),q)$ of the subtree $\gaminf(s|_u)$.
\end{itemize}
If $g$ is the empty sequence, then no choice of Spoiler has been recorded. We let $w_g$ be the empty path, $v(w_g)=\eps \in \Pos(\gaminf(s))$, and $\tr(w_g)= (\eps,(p,\chi(p)))\in V_0$.
The game $\alp$ starts at that vertex $(\eps,(p,\chi(p)))$ belonging to $P_0$. The game defines a unique run
$\rho_1\in \RUN_{B_H}(\gam(s(\eps)),p)$.
If no hole appears in $\gam(s(\eps))$, then $\gaminf(s)$ is equal to $\gam(s(\eps))$.
We are done in this case: knowing $\rho_\eps=\rho_1$, all four items are defined if $g$ is the empty sequence.

In the other case there exist some $i\in H$ and $i_j\in \leaf_i(\gam(s(u)))$. Thus, the game $\alp$ is not finished; and  Spoiler decides on some leaf labeled by a hole.
Let $g'=(g_1,\ldots, g_{k+1})$ be the corresponding prefix.
By induction,  $w_g$ is defined for  $g= (g_1,\ldots, g_k)$. Let $\tr(w_g)=(u,(q,c))$, and $x=s(u)$.
Hence, there is a tree $t_u=\gam(x)$, and, by induction, $t_u$ is an IO-prefix of the subtree $\gaminf(s)|_v$ of $\gaminf(s)$ rooted at
the position  $v=v(w_g)$ in $\gaminf(s)$. Note that $u\in \Pos(s)$ and $v\in \Pos(\gaminf(s))$. \Ip since $t_u$ is an IO-prefix, every position $v$ in $t_u$ can be identified with a unique position $\pos(v)\in \Pos(\gaminf(s))$ by the endpoint of the path $w_g.w_u(v)$. Here, $w_u(v)$ denotes the unique directed path in $t_u$ from the root to the position $v$ in $t_u$.
The choice $g_{k+1}$ of spoiler defines a position $i_j\in \leaf_i(t_u)$ for some $i\in [\rk(x)]$. We let $w_{g'}$ be the unique path which starts at  the root, has $w_g$ as prefix and stops in $\Pos(i_j)\in \Pos(\gaminf(s))$. Then we define
$v(w_{g'}) = \Pos(i_j)$. It is the endpoint of the path $w_{g'}$.
The unique outgoing edge at $\tr(w_g)=(u,(q,c))$ defines a run $\rho_{g'}\in
\RUN_{B_H}(\gam(s(u)),q)$. Since $g_{k+1}$ is defined, the position
$u.i\in \Pos(s)$ exists, and we let $\tr(w_{g'})= (u.i,(\rho'(i_j),c_{\rho'}(i_j)))$ where $\rho'=\rho_{g'}$.
Note that the tree
$\gam(s(u.i))$ appears as an IO-prefix of the subtree $\gaminf(s|_{u.i})$. Thus, all desired items are defined for the sequence $g'$.
%Moreover, the root position is labeled by some element in $\Sig$.

See \prref{fig:pgnps} how a finite directed path $w$ in $\gam^s_{\infty}(s)$ starting at the root defines a corresponding path in $s$. The figure also shows
 that $w$ defines a unique position in $v(w)\in \Pos(\gam_{\infty}(s))$ via the arrows which start in $\gam_{\infty}(s)$, leave the  tree on the left, and make a clockwise turn upwards to enter the  tree from the right.
 This visualizes the formal definition of a path $w_g$ and its position $v(w_g)$ in $\gaminfs$.

Suppose that $\tr(w)= (u,(q,c))$ and $\tr(w')=(u,(q',c'))$ where
$w=w_g$ and $w'=w_{g'}$ are finite prefixes of decision sequences for Spoiler. We claim that if $(q,c)\neq (q',c')$, then the positions of $v(w)$ and $v(w')$ are incomparable. That is, there is no directed path between them.
Indeed, if $(q,c)= (\rho(i_j),c_\rho(i_j))$ and  $(q',c')= (\rho(i'_{j'}),c_\rho(i'_{j'}))$, then we must have $i'_{j'}\neq i_{j}$.
The claim follows by induction on $\abs w$. Details are left to the reader.
Thus, the runs $\rho(u,q,c)\in \RUN_{B_H}(\gam(s(u)),q)$ and
$\rho(u,q',c')\in \RUN_{B_H}(\gam(s(u)),q)$ (chosen by $P_0$ in his positional strategy) might be different, for example because $c\neq c'$. Since the positions $v(w)$ and $v(w')$ are incomparable, we can use $\rho(u,q,c)$ to define a partial run of the subtree rooted at
$v(w)\in \Pos(\gaminfs)$ and we can use $\rho(u,q',c')$ to define a partial run of the subtree rooted at $v(w')$.

Let $n\in \N$. Consider all finite directed paths $w$ in $\gaminfs$ beginning at the root such that $\abs w\leq n$, and among them those which are induced by finite prefixes of decision sequences $g$ of Spoiler. Let's call them $w_g$. Then the runs $\rho_g$ constructed above for $w_g$ yield  partial runs $\wt\rho_{n}$ of $\gaminf(s)$.
Note that these partial runs are compatible: for every common positions of two such partial runs, the labels (being states) are the same.
{This follows because $w_g= w_{g'}$ implies $\rho_g(v(w_g))=\rho_{g'}(v(w_{g'}))$.}
Moreover, the sequence $(\wt\rho_{n})_{n\in \N}$ has a well-defined limit
$\rho_0= \lim_{n\to \infty}\wt \rho_n$. Indeed, for every
$k\geq 0$ there is some $n_k$ such that $\wt \rho_{n}(v)=\wt \rho_{n_k}(v)$ for all $n\geq n_k$ and all positions $v\in \Pos(\gaminf(s))$ having a distance  less than $k$ to the root.

Therefore, the set of all directed paths in
the arena $G'(A,s)$ starting at $(\eps,(p,\chi(p)))$ yields $\rho_0$ as a (totally defined) run
$\rho_0:\Pos(\gaminf)\to Q$ such that $\rho_0(\eps) =p$.
(We write
$\rho_0$ because its definition depends on the positional winning strategy
of~$P_0$.)

We have to show that the run $\rho_0$ is accepting. For that we have to consider all directed paths in the tree defined by $\rho_0$. If a path ends at a leaf, then it is accepting because $\rho_0$ is a run.
Every infinite path in $\rho_0$ is due to some infinite game $\alp$ starting at $(\eps,(p,\chi(p)))$. The game  defines a decision sequence $g(\alp)$ and the run $\rho_0$ labels all vertices on the unique infinite directed path $\bet$ which visits all positions
$v(w_{g_k})$ where $g_k$ runs over all finite prefixes of $g$.  By construction, the minimal color $c(\alp)$ appearing infinitely often in the game $\alp$ is the same color as occurring infinitely often on $\bet$. Since Prover wins all games, $c(\alp)$ is even. Hence, the run
$\rho_0$ is  accepting.
Thus, the assertion we were looking for.
\end{proof}

\section{Main results: The inclusion problem into regular sets}\label{sec:mainres}
Henceforth, if $\sig: \cX\to 2^{T(\Sig\cup H)\sm H}$ is a \subst
and $R\sse T(\Sig)$ is a subset, then we let
\begin{align}\label{eq:defsR}
\sig_{\mathrm{io}}^{-1}(R)=\set{s\in T(\Sig\cup \cX)}{\sio(s)\subseteq R}.
\end{align}
\prref{eq:defsR} extends the earlier notation used for \hom{s};
and the following lemma generalizes \prref{lem:parthom} by switching from a \hom $\gam$  to a regular \subst $\sig$ to sets of finite and infinite trees.
%%%%%%%%%%%%%%%%%%%%%%%%%%%%%%
\begin{lem}\label{lem:SigProf}
Let  $R\sse T(\Sig)$ be regular and $\sig: \cX\to 2^{T(\Sig\cup H)\sm H}$ be a regular \subst. Then the set
$\sig_{\mathrm{io}}^{-1}(R)=\set{s\in T(\Sig\cup \cX)}{\sio(s)\subseteq R}$
is effectively regular.
\end{lem}
\begin{proof}
 Let us begin to prove the lemma with the notation
$\Sig'$, $R'$, and  $\sig'$.
Thus, formally, we begin with $R'\sse T(\Sig')$ and $\sig': \cX\to 2^{T(\Sig'\cup H)\sm H}$. The aim is to show that ${\sig'}\phantom{}_{\mathrm{io}}^{-1}(R')=\set{s\in T(\Sig'\cup \cX)}{\sig'_\mio(s)\subseteq R'}$ is effectively regular. For that we introduce a new constant~$a$ and we define a larger alphabet than $\Sig'$ by letting $\Sig$ to be the disjoint union of
$\Sig'$ and $\os a$. We have $R' \cup\sig'(x) \sse T(\Sig'\cup H)\sm H$ for all $x\in \cX$.
Define $\sig(x) = \sig'(x)\cup (T(\Sig)\sm T(\Sig'))$ for all $x\in \cX$ and $R= R'\cup (T(\Sig)\sm T(\Sig')) \sse
T(\Sig)$. Then $\sig$ and $R$ are regular. Moreover,
$R$ contains all trees which have a leaf labeled by $a$. For all $s\in T(\Sig')$ we have
\begin{align*}
\sio'(s) \sse R'\iff \sio(s) \sse R\iff s\in \sig_{\mathrm{io}}^{-1}(R).
\end{align*}
 Moreover, if $\sig_{\mathrm{io}}^{-1}(R)$ is regular, then
$\sig_{\mathrm{io}}^{-1}(R) \cap T(\Sig')$ is regular, too.
Thus,
it is enough to prove the lemma
under the assumptions that first, $\sig(x)\neq \es$ for all $x\in \cX$
and second, there is a constant $a\in\Sig$ such that $T(\Sig)\sm T(\Sig\sm \os a)\sse R$.
Let $R=L(B,q_0)$ for some \pNTA $B$. According to \prref{def:alpProf} in \prref{sec:specsio} we embed the NTA $B$ first into $B_H$ and then $B_H$ into the \pNTA $B_\cP$. We also embed  $T(\Sig)$ into
$T(\Sig_\cP)$ such that both, $R\sse L(B_\cP,q_0)$ and for every profile $\pi=\pi(t)\in \cP_B$ there is some
finite tree $t_\pi\in \Tfin(\Sig_\cP)$ satisfying $\pi(t_\pi)=\pi$. Since every tree with a leaf labeled by $a$ belongs to $R$,
we may assume that $a \in L(B,p)$ for all states $p$. Hence, if  $t\models \tau \in\cT_B$, then
$t[i_j\la a]\models \tau \in\cT_B$, too.

In the following let $\wt R= L(B_\cP,q_0)$. Since $\sig(x)$ is regular for all $x\in \cX$ we can compute
the specialization $\vsig$ of $\sig$ as defined in \prref{def:specsat}. Using \prref{eq:checksoi} in
\prref{prop:redBcP} we can state
\begin{align}\label{eq:RBcP}
\forall s\in T(\ScX):\; \sio(s)\sse R \iff \sio(s)\sse \wt R\iff \check\sig_{\mathrm{io}}(s)\sse \wt R
\end{align}
%Then, using \prref{prop:equivtasl}, the line (\ref{eq:RBcP}) implies
%\begin{align}\label{eq:wtR}
%\forall s\in T(\ScX):\; \sio(s)\sse R \iff \vsig_{\mathrm{io}}(s)\sse \wt R
%\end{align}
Define a new set of variables by
\begin{align}\label{eq:largeX}
\cX'&=\set{(x,\pi)\in (\Sig\cup \cX)\times \cP_B}
{\exists t\in \sig(x):\; \pi=\pi(t)} \text{ with }
\rk(x,\pi)=\rk(x).
\end{align}
Note that for all $f\in \Sig$ there exists a variable $\big(f,\pi(f(1\lds \rk(f)))\big)\in \cX'$. That is why, below, it is enough to work with $T(\cX')$ rather than with $T(\Sig \cup\cX')$.
%We also write  $x_\pi$ to denote the variable $(x,\pi)\in \cX'$ and we let $\rk(x_\pi)=\rk(x)$.
We use the second component of $(x,\pi)\in \cX'$ to define a \hom
$\gam:\cX'\to \Tfin(\Sig_\cP\cup H)\sm H$ by  $\gam(x,\pi) =t_\pi$  where $t_\pi$ was defined in \prref{fig:RunProf}.
(Recall that $t_\pi$ satisfies $\pi=\pi(t_\pi)$.)
Then, by \prref{eq:spec}, for every $x\in X$ the following equation holds.
\begin{equation}\label{eq:gamxpi}
\vsig(x)=\set{t_{\pi} \in T(\Sig_{\Pro}\cup H)}{\exists t \in \sig(x):\, \pi=\pi(t)} = \set{\gam(x,\pi)}{\exists t\in \sig(x): \pi=\pi(t)}.
\end{equation}
Since $\sig(x)\neq \es$, the projection onto the first component $(x,\pi) \mapsto x$ defines a surjective
mapping $\psi_\infty: T(\cX')\to T(\Sig)$ by the \hom
\begin{align}\label{eq:h}
\psi:\cX'\to T(\Sig \cup \cX\cup H)\sm H,\quad (x,\pi) \mapsto x(1\lds \rk(x)).
\end{align}
A top-down induction shows $\Pos(s')=\Pos(\psi_\infty(s'))$ for all $s'\in T(\cX')$. Also note that $\vsig_\mio(s)=\gam_\infty(\psi_\infty^{-1}(s))$ for all $s\in T(\ScX)$. Indeed, consider any $u\in \Pos(s)$ with $x=s(u)$. Then $\oi h(x) =\set{(x,\pi)}{\exists t\in \sig(x): t\models \pi}$. Hence, $\gam(\oi h(x))= \set{t_\pi}{\exists t\in \sig(x): t\models \pi}= \vsig(x)$ where the last equation follows by \prref{eq:gamxpi}.
Therefore,
\begin{align*}
\sio^{-1}(R)&=\vsio^{-1}(\wt R)&\text{by (\ref{eq:RBcP})}\\
&=\set{s\in T(\Sig \cup \cX)}{\gam_\infty(\psi_\infty^{-1}(s))\sse \wt R}\\
&=\set{s\in T(\Sig \cup \cX)}{\psi_\infty^{-1}(s)\sse \gam_\infty^{-1}(\wt R)}.
\end{align*}
It enough to show that
$T(\Sig \cup \cX)\sm \sio^{-1}(R)=
\set{s\in T(\Sig \cup \cX)}{\psi_\infty^{-1}(s)\cap \gam_\infty^{-1}(\wt R)\neq \es}= \psi_\infty\big(T(\cX') \sm \gam_\infty^{-1}(\wt R)\big)$ is regular. The set $\gam_\infty^{-1}(\wt R)$ is regular by \prref{lem:parthom}. We can write $T(\Sig \cup \cX') \sm \gam_\infty^{-1}(\wt R)=L(A',p_0)$ for a \pNTA $A'=(Q,\Sig \cup\cX',\del',\chi)$. It is therefore enough to construct a
\pNTA $A=(Q,\ScX,\del,\chi)$
such that $L(A,p_0)= \psi_\infty(L(A',p_0))$.
%
%
%Thus, $\sio^{-1}(R)$ is regular \IFF the set $S=\{s\in T(\ScX)\mid h^{-1}(s)\in \gam_\infty^{-1}(\wt R)\}$ is regular.
%Let $\ov S= T(\Sig \cup \cX) \sm S$ and $\ov R= T(\Sig ) \sm\wt R$.
%A purely set theoretical consideration shows that
%\begin{equation}\label{eq:ovSovR}
%\ov S= \{s\in T(\Sig \cup \cX)\mid \oi h(s)\cap \gam_\infty^{-1}(\ov R)\neq \es\}
%= h(\gam_\infty^{-1}(\ov R)) = h(T(\cX') \sm \gam_\infty^{-1}(\wt R)).
%\end{equation}
%The set $T(\cX') \sm \gio^{-1}(\wt R)$ is effectively regular
%because $\gio^{-1}(\wt R)$ has this property.
%Thus, we can write $T(\cX') \sm \gio^{-1}(\wt R)=L(A',p_0)$ for a \pNTA $A'=(Q,\cX',\del',\chi)$. It is therefore enough to construct a
%\pNTA $A=(Q,\ScX,\del,\chi)$
%such that $L(A,p_0)= h(L(A',p_0)).$
The construction of $A$ is straightforward by using another \tra relation.
We define $\del$ by the following equivalence where $x\in \Sig\cup \cX$:
\begin{align}\label{eq:Aprime}
(p,x,q_1\lds q_r)\in \del \iff \exists \pi\in \cP_B : \big(p,(x,\pi),q_1\lds q_r\big)\in \del'.
\end{align}
The assertion follows.
\end{proof}
By a \emph{class of tree languages} we mean a family
of sets ${\cC}(\Del)$, indexed by all finite ranked alphabets $\Del$, such that first, we have $\cC(\Del) \subseteq T(\Del)$ and second, for every rank-preserving inclusion
$\phi:\Del \to \Gam$, the tree-homomorphism defined by $\phi$ induces an inclusion $\cC(\Del) \to\cC(\Gam)$.
\begin{defi}\label{def:inreg}
We say that a  class of tree languages $\cC$ satisfies the property \emph{INREG}, if the following  \emph{inclusion problem into regular sets} is decidable:
For every finite ranked alphabet $\Del$, on input $L\in \cC(\Del)$ (given in some effective way) and a regular tree language $K\sse T(\Del)$ (given, say, by some \pNTA) the problem ``$L\sse K$?'' is decidable.
\end{defi}
\begin{thm}\label{thm:s1s2main}Let $\cC$ be  a class of tree languages fulfilling the property INREG of \prref{def:inreg}.
Then the following decision problem is decidable.
\begin{itemize}
\item Input: Regular \subst{s} $\sig_1,\sig_2$ and tree languages $L\sse T(\Sig\cup\cX)$,
$R\sse T(\Sig)$ such that $R$ is regular and $L\in \cC(\Sig \cup \cX)$.
\item Question: Is there some \subst $\sig: \cX\to 2^{T(\Sig\cup H)\sm H}$ satisfying both, $\sio(L) \sse R$  and  $\sig_1\leq \sig\leq \sig_2$?
\end{itemize}
Moreover, we can effectively compute the set of maximal \subst{s} $\sig$ satisfying $\sio(L) \sse R$  and  $\sig_1\leq \sig(x) \leq \sig_2$. It is a finite set of regular \subst{s}.
\end{thm}
\begin{proof}
We let $R=L(B,q_0)$ for some parity NTA $B$.
First, we check that $\sig_1(x)\sse\sig_2(x)$ for all $x\in \cX$ because otherwise there is nothing to do. By definition, we have $\sig_1(f)=\sig_2(f)=f(1\lds [\rk(f)])$ for all $f\in \Sig\sm \cX$.
Thus, without restriction, $\sig_1(x)$ and $\sig_2(x)$ are defined as regular sets for all $x\in \Sig\cup \cX$ with $\sig_1\leq \sig_2$. If there is any \subst $\sig: \cX\to 2^{T(\Sig\cup H)\sm H}$ satisfying both, $\sio(L) \sse R$  and  $\sig_1\leq \sig\leq \sig_2$, then $\sig_1$ is the unique \textbf{minimal} \subst satisfying that property.
The \subst $\sig'$ defined by $\wh \sig_1(x)\cap\sig_2(x)$ satisfies that property, too.
It belongs to the following effectively computable finite set  $S_2$
 of  regular \subst{s}
 \begin{align}\label{eq:S_2main}
S_2=\{\sig':\cX\to 2^{T(\Sig\cup H)\sm H}\mid \exists \,\sig
\,\forall x\in \cX: \sig_1(x) \sse \sig'(x) = \wh\sig(x) \cap \sig_2(x)\}.
\end{align}
The set $S_2$ is the same as defined above in (\ref{eq:S_2eff}) in the proof of \prref{cor:maxioreg}. The notation in \prref{eq:S_2main} refers to \subst{s} $\sig:\cX\to 2^{T(\Sig\cup H)\sm H}$ and their saturations $\wh \sig$ \wrt the NTA $B$.
Thus, if the set $S_2$ is empty, then we can stop. The answer to the question in \prref{thm:s1s2main} is negative; and the set of maximal \subst{s} $\sig$ satisfying $\sio(L) \sse R$  and  $\sig_1\leq \sig(x) \leq \sig_2$ is empty.

Thus, we may assume $S_2\neq \es$, which is computable by \prref{cor:maxioreg}.
Next, in a second phase we consider each element $\sig \in S_2$, one after another.
By \prref{lem:SigProf} we know that for each $\sig \in S_2$, the
set $\sig^{-1}_{\mathrm{io}}(R)$ is an effectively regular set. The assertions $\sio(L)\sse R$ and $L\sse \sig^{-1}_{\mathrm{io}}(R)$ are equivalent.
Since $L\in \cC$ we can check $L\sse \sig^{-1}_{\mathrm{io}}(R)$. Thus, we end up with an effectively computable subset $S_2'$ of $S_2$ such that
there is some \subst $\sig: \cX\to 2^{T(\Sig\cup H)\sm H}$ satisfying both, $\sio(L) \sse R$  and  $\sig_1(x)\sse \sig(x) \sse \sig_2(x)$ for all $x\in \cX'$ \IFF $S_2'\neq\es$. In case
$S_2'=\es$ the answer to the question in \prref{thm:s1s2main} is negative, again. Hence, without $S_2'\neq\es$.
  The maximal \subst{s} $\sig$  satisfying both, $\sio(L) \sse R$  and  $\sig_1(x)\sse \sig(x) \sse \sig_2(x)$ for all $x\in \cX'$  are in $S_2'$. Since $S_2'$ is a nonempty, finite, and computable set of regular substitutions, we can compute all maximal element(s) in $S_2'$. %We are done.
\end{proof}
Note that, in the above theorem, $\cC$ can be chosen strictly larger than the class of all regular languages
and still fulfill the hypothesis INREG\@.
For example we can define $\cC(\Del)$ as consisting of all the regular languages in $T(\Del)$ augmented with all context-free languages over $\Tfin(\Del)$ as defined in~\cite{Courcelle78-tcs2,EngelfrietS77,Guessarian83,Schimpf-Gallier85}.

\begin{rem}\label{rem:cL}
The hypothesis INREG cannot be removed in \prref{thm:s1s2main}.
For a counter-example we can choose the family $\cC$ such that $\cC(\Del)$ consists of all regular subsets of $T(\Del)$
together with all the subsets $T(\Del)\sm K$ where $K$ is a context-free subset of $T_{fin}(\Del)$.\\
The inclusion problem: ``$R \subseteq L$?'' for $R$ regular and $L$ context-free, reduces to  the
problem:\\
$\exists \sigma: \cX\to 2^{T(\Del\cup H)\sm H}:\, \sigma(T(\Del) \sm L) \subseteq (T(\Del) \sm R)$ where
$\sig_1=\sig_2=H= \es$.
But the inclusion of a regular language into a context-free language is undecidable, showing that the substitution-inclusion problem for the class $\cC$ is undecidable.
\qed
\end{rem}
The following corollary considers the special case where we demand that \subst{s}
map variables to nonempty subsets of trees. This setting is rather natural and
it is also the setting when Conway studied the problem for finite words.
%\prref{cor:main} is a straightforward consequence of \prref{thm:s1s2main} and the fact that saturated \solu{s} are regular.

\begin{cor}\label{cor:main}Let $\cC$ be  a class of tree languages  fulfilling the property INREG as in \prref{thm:s1s2main}.
Then the following decision problem is decidable.
\begin{itemize}
\item Input: Tree languages $L\sse T(\Sig\cup\cX)$,
$R\sse T(\Sig)$ such that $R$ is regular and $L\in \cC(\Sig\cup\cX)$.
\item Question: Is there some \subst $\sig: \cX\to 2^{T(\Sig\cup H)\sm H}$ satisfying both, $\sio(L) \sse R$  and  $\sig(x)\neq \es$ for all $x\in \cX$?
\end{itemize}
Moreover, we can effectively compute the set of maximal \subst{s} $\sig$ satisfying $\sio(L) \sse R$  and  $\sig(x)\neq \es$ for all $x\in \cX$. It is a finite set of regular \subst{s}.
\end{cor}
\begin{proof}
By definition of regularity, we have $R=L(B,q_0)$ for some \pNTA $B$.
We run the following nondeterministic decision procedure.
For each $x\in \cX$ we guess a profile $\pi_x\in \Pro_{B}$.
Then we check that there is some $t\in T(\Sig \cup [\rk(x)])\sm H$ such that $\pi_x=\pi(t)$. If there is no such $t$, then the assertion in the corollary has a negative answer for that guess.  If there exists such $t$, then we let $\sig_1(x)= \set{t\in T(\Sig \cup [\rk(x)])\sm H}{\pi=\pi(t)}$ and $\sig_2(x)=T(\Sig \cup [\rk(x)])\sm H$. We have $\wh\sig_1=\sig_1$ and therefore $\sig_1$ is regular by \prref{prop:satsig}.
We now apply \prref{thm:s1s2main}. The assertion in the corollary
is positive \IFF for some guess the application of \prref{thm:s1s2main} yields a positive answer
with a set of maximal \solu{s}. From these data we can compute the
set of maximal \solu{s} where $\sig(x)\neq \es$ for all $x\in \cX$.
Indeed, let $\sig(x)\neq \es$ for all $x\in \cX$.
Then there is some $t_x\in \sig(x)$ and we can define $\pi_x=\pi(t_x)$. The non-deterministic procedure can guess the
mapping $x\mapsto \pi_x$ and we obtain the maximal \solu{s} for that guess. Simulating all guesses and collecting the maximal \solu{s} for each guess yields the result.
\end{proof}
\medskip
\section{Conclusion and open problems}%
\label{sec:op}
The main result of the paper is \prref{thm:s1s2main}.
We included the special case \prref{cor:main} because the assertion in the corollary corresponds exactly to the results for regular languages over finite words due to Conway as stated in the introduction, \prref{sec:intro}.
The assertions of \prref{thm:s1s2main} and
\prref{cor:main} are positive decidability results, but we don't know any
 (matching) lower and upper complexity bounds except for a few special cases.

Various other natural problems are open, too. For example, a remaining question is whether it is possible to derive positive results for the outside-in extension $\soi$. Another puzzling problem is that we don't know how to decide for regular tree languages $L$ and $R$
whether there exists a \subst $\sig$ such that $\sio(L)=R$.
We have seen that if such a \subst $\sig$ exists, then $\sig$ is in a finite and effectively computable set of regular \subst{s}. Thus, the underlying problem has no existential quantifier: decide ``$\sio(L)=R$?'' on input
$L$, $R$, and $\sig$ where $L$, $R$, and $\sig$ are regular. If the answer is \emph{yes}: $\sio(L)=R$, then $\sio(L)$ is regular. So, one could try to solve that problem, first. Recall that the problem is decidable in the setting of finite trees and \hom{s} by~\cite{CreusGGR16SIAMCOMP,GodoyG13JACM}.

Actually, we don't know how to  decide the problem ``$\sio(L)=T(\Sig)$?'' where $L\sse T(\cX)$ and  $\sig(x)= T(\Sig\cup [\rk(x)])$ for  all $x\in \cX$.
Both problems ``$\exists \sig: \sio(L)=R$?'' and ``$\exists \sig: \soi(L)\sse R$?'' remain open, even if we restrict ourselves to deal with finite trees, only.

It is also open whether better results are possible if we restrict $R$ (or $L$ and $R$) to smaller classes of regular tree languages like the class of languages with B\"uchi acceptance.

\bibliography{traces}
\bibliographystyle{alphaurl}

\appendix
\section{Some additional material}\label{sec:append}
As noticed in the preamble to the present paper:  it is not necessary to read anything in the appendix to understand the main results.
L'annexe, c'est l'art pour l'art: ars gratia artis.
\subsection{Conway's result for finite and infinite words}\label{sec:oldconwords}
This section is meant for readers who are interested to see
  proofs of our main results in the special (and simpler) case of finite and infinite words.

Let $\Sigma$ and  $\cX$ be  finite alphabets and let $\Sig^\infty$ denote the set of finite and infinite words. That is $\Sig^\infty = \Sig^*\cup \Sig^\oo$ where $\Sig^\oo$ is the set of infinite words. The aim is to give an essentially self-contained proof for (a generalization of) Conway's result with respect to $\#\in \os{\sse,=}$ for subsets of $\Sig^\infty$. For that we consider ``regular constraints'' in the spirit of \prref{thm:s1s2main}. The proof for $\Sig^*$ is given with all details. Next, we explain that the same approach works for infinite words, too. However, we are more sketchy and our proof uses the well-known result that regular $\oo$-languages form an effective Boolean algebra~\cite{buc62}. This means that on input $\oo$-languages $L_1,L_2$ (for example, specified by B\"uchi automata) we can construct B\"uchi automata for $\Sig^\oo$,
$L_1\cup L_2$, and $\Sig^\oo\sm L_1$. There is a way to avoid an explicit complementation, for example by using ultimately period words as a witness to show that $L_1\sm L_2\neq \es$. \begin{prop}\label{prop:mainword}
Let $\sig_1,\sig_2$ be mappings from $\ScX$ to $2^{\Sig^+}$ such that for $i=1,2$ first, $\sig_1(a)=\os{a}$ for every $a\in \Sig\sm \cX$ and second,
$\sig_i(x)$ is a regular language for every $x\in \cX$.
Then for regular languages $L\sse (\Sig\cup \cX)^\infty$ and  $R\sse \Sig^\infty$
the following assertions hold for $\#\in \os{\sse,=}$.
\begin{enumerate}
\item The following decision problem is decidable.
\begin{align}\label{eq:Confinword}
\text{``$\exists \sig: \sig(L)\mathop{\#} R\wedge \forall x: \sig_1(x)\sse \sig(x) \sse \sig_2(x)$?''}
\end{align}
\item  Define $\sig\leq \sig'$
	by $\sig(x)\sse \sig'(x)$ for all $x\in \cX$. Then every solution $\sig$ of (\ref{eq:Confinword}) is bounded from above by  a maximal \solu; and the number of maximal solutions is finite.
Moreover, if $\sig$ is maximal, then $\sig(x)$ is regular for all $x\in \cX$; and all maximal \solu{s} are effectively computable.
\end{enumerate}
\end{prop}

\noindent
Recall that a \solu $\wt \sig$ of (\ref{eq:Confinword}) is maximal and above a \solu
$\sig$ \IFF first, $\sig \leq \wt \sig$ and second, if $\sig'$ is any \solu of (\ref{eq:Confinword}) with $\wt\sig \leq \sig'$, then $\wt\sig = \sig'$.
\begin{proof}
 Clearly, we may assume without restriction that $\sig_1(x)\sse \sig_2(x)$ for all $x$ because otherwise there are no \solu{s}.

\noindent{\textbf{Finite words:}} The easiest situation is the original setting of Conway: $L\sse (\Sig \cup \cX)^*$ is arbitrary and $R\sse \Sig^*$ is regular.
In that case let $B=(Q,\Sig,\del,I,F)$ be an NFA such that
 $\del\sse Q\times \Sig \times Q$ and $R=L(B)$.
Without restriction, $Q=\os{1\lds n}$ with $n\geq 1$.
For $p,q\in Q$ we denote by $L[p,q]$ the set of words accepted by the NFA $B_{p,q}= (Q,\Sig,\del,\os p,\os q)$. Next, we consider the
semiring of Boolean $n\times n$ matrices $\B^{n\times n}$ where, as usual, $\B=(\os {0,1},\max,\min)$. We use the multiplication of matrices to define $\B^{n\times n}$ as a monoid where the unit matrix is the neutral element.
For a letter $a\in \Sig$ we denote by $M_a$ the matrix
such that for all $p,q$ we have $M_a(p,q) =1\iff a\in L[p,q]$. Since $\Sig^*$ is a free monoid, the $M_a$'s
define a \hom
$\mu:\Sig^*\to \B^{n\times n}$ such that for all $w\in \Sig^*$ we have $M_w(p,q) =1\iff w\in L[p,q]$ where $M_w = \mu(w)$.
We have $\oi \mu(\mu(R))= R$. Indeed, for $u\in \Sig^*$ let $[u]$ be the set of words $v$ such that $\mu(u)=\mu(v)$. Then we have
\begin{align}\label{eq:finuni}
R= \bigcup\set{u\in L[p,q]}{p\in I, q\in F}= \bigcup_{p\in I, q\in F}
\bigcup\set{[u]\sse \Sig^*}{u\in L[p,q]}.
\end{align}
The verification of (\ref{eq:finuni}) is straightforward from the definition of $\mu$. The crucial observation is that the union in (\ref{eq:finuni}) is finite because the set $\set{[u]\sse \Sig^*}{u\in \Sig^*}$ is finite. Its cardinality is less than $2^{2^{n^2}}$.
\bigskip

\noindent{\textbf{Final Steps.}}
We are almost done with the proof for finite words. We need a few more steps.
Let $\sig:\cX\to 2^{\Sig^*}$ be any \subst. Define its \emph{saturation} $\wh \sig:\cX\to 2^{\Sig^*}$
by
$ %\[
\wh \sig(x) = \bigcup\set{[u]\sse \Sig^*}{u\in \sig(x)}.
$ %\]
Then $\sig\leq \wh \sig$, the set $\wh \sig(x)$ is regular (being a finite union of regular sets), and if $\sig(L)\mathop{\#} R$, then it holds $\wh \sig(L)\mathop{\#} R$, too. Moreover, if
$\sig$ is regular, then we can calculate $\wh \sig$ because we can decide for
all $m\in \B^{n\times n}$ and $x\in \cX$ whether $\oi\mu(m)\sse \sig(x)$.
Let us call a \subst $\sig'$  be
a \emph{candidate} as soon as first, $\sig_1(x)\sse \sig'(x)$ for all $x\in \cX$ and second, $u\in \sig'(x)$ implies $[u]\sse \sig'(x)$.
Since $\set{[u]\sse \Sig^*}{u\in \Sig^*}$ is an effective list of finitely many regular sets, the finite list $\cC$ of candidates is computable: there are at most $2^{2^{n^2}|\cX|}$ candidates.
Since we assume $\sig_1\leq \sig_2$, we compute for each candidate $\sig'\in \cC$ another regular
\subst $\wt\sig$ by $\wt\sig(x) = \sig'(x)\cap \sig_2(x)$.
The point is that whenever the problem in (\ref{eq:Confinword}) has any \solu $\sig:\cX \to 2^{\Sig^*}$, then there is a maximal \solu  of the form $\wt\sig$
 which satisfies $\wt\sig(L)\mathop{\#} R$.
The class of  regular languages is closed under \subst of letters by regular sets. This is a standard exercise in formal language theory.
Hence,  $\wt \sig(L)$ is regular; and we can decide  $\wt \sig(L)\mathop{\#} R$ for every $\wt\sig$.  This is the same as to decide $\wt\sig(L)\sse R$. Thus, we have to decide $\wt\sig(L)\cap (\Sig^\oo\sm R)=\es$. The test involves the complementation\footnote{It is here where our exposition uses for $\oo$-regular languages B\"uchi's help.\label{foot:1}} of $R$.

So, we end up with a nonempty list $\cL$ of \subst{s} $\wt \sig$ satisfying $\wt\sig(L)\mathop{\#} R$.
The list $\cL$ is finite and effectively given. Hence, we can compute all maximal elements. We are done.
\bigskip

\noindent{\textbf{Infinite words:}}
Let us show that (using~\cite{buc62}) the case of infinite words can be explained in a similar fashion. For simplicity we restrict ourselves to \subst{s} where $\sig(x)$ is a non-empty subset of $\Sig^+$. (The other cases are not harder, but need more case distinctions.)
The starting point are two $\oo$-regular languages
$L\sse (\Sig\cup \cX)^\oo$, and  $R\sse \Sig^\oo$. We use the fact that every $\oo$-regular language can be accepted by a
nondeterministic B\"uchi automaton. The syntax of a B\"uchi automaton is as for an NFA:\@ $B=(Q,\Sig,\del,I,F)$ where $\del\sse Q\times \Sig \times Q$.

A word $w\in  \Sig^\oo$ is accepted if there are $p\in I$ and $q\in F$ such that $B$ allows  an infinite path labeled by $w$ which begins in $p$ and visits the state $q$ infinitely often.

Instead of working over the Booleans $\B$, we consider the three-element commutative idempotent semiring $S=(\os{0,1,2},+,\cdot)$ where
$+=\max$ and
$x\cdot y= 0$ if $x=0$ or $y=0$ and otherwise  $x\cdot y=\max\os{x,y}$.
Note that $0$ is a zero and $1$ is neutral in $(S,\cdot)$. Let us define for
$Q=\os{1\lds n}$ and $a\in \Sig$ the matrix $M_a\in S^{n \times n}$ by
\begin{align*}
 M_a(p,q) =
\begin{cases}
 \text{$0$ \hspace{0cm} if $(p,a,q)\notin \del$,}\\
 \text{$1$ \hspace{0cm} if $(p,a,q)\in \del$ but $\os{p,q} \cap F= \es$,}\\
 \text{$2$ otherwise: if $(p,a,q)\in \del$ and $\os{p,q} \cap F\neq \es$.}
\end{cases}
\end{align*}
The multiplicative structure $(S^{n \times n},\cdot)$ yields a finite monoid with
$3^{n^2}$ elements.
The matrices $M_a\in S^{n \times n}$ define a \hom
$\mu:\Sig^*\to S^{n \times n}$. For all $p,q\in Q$ and $w\in \Sig^*$ the interpretation for $M_w=\mu(w)$ is as follows.
We have $M_w(p,q)\neq 0$ \IFF there is a path labeled by $w$ from state $p$ to $q$. Moreover, $M_w(p,q)=2$ \IFF there is a path labeled by $w$ from state $p$ to $q$ which visits a final state.
Let  $T\sse S^{n \times n}$ be any subset, then $\oi\mu(T)$ is a regular language of finite words: $S^{n \times n}$ is the state set of a (deterministic!) NFA
accepting $\oi\mu(T)$. In the terminology of $\oo$-languages:  $\mu$ \emph{strongly recognizes} every $L= L(Q,\Sig,\del,\os p,\os q)\sse \Sig^\oo$
where $p,q\in Q$. Strong recognition means that
if $u= x_0x_1\cdots$ and $v= y_0y_1\cdots$ are infinite sequences of finite words such that $\mu(x_i)=\mu(y_i)$ for all $i\in \N$, then $u\in L\iff v\in L$.
This property is crucial in the following. The verification of this property is straightforward and left to the reader.

We proceed as in the case of finite words.  We let $[u]=\set{v\in \Sig^+}{\mu(u)=\mu(v)}$. Now, consider any \subst $\sig:\cX\to 2^{\Sig^+}\sm \es$. Let us define
its saturation $\wh\sig$ by $\wh\sig(x) = \bigcup\set{[u]\in \Sig^+}{u\in \sig(x)}$. The saturation $\wh\sig$ is regular. Moreover, if $\sig$ itself is regular, then we can compute $\wh\sig$ effectively. (Note that this involves sets of finite words, only.) The number of saturated
\subst{s} is less than $2^{3^{n^2}|\cX|}$. We can calculate a list all saturated
\subst{s}.

We also use the fact that if $L\sse (\Sig\cup\cX)$ is $\oo$-regular and if  $\sig:\cX\to 2^{\Sig^+}$ is a regular  \subst, then $\sig(L)$ is effectively $\oo$-regular. Again, this fact is rather easy to see and it does not rely on the results of~\cite{buc62}.
It is well-known that a language $L$ of infinite words is $\oo$-regular
\IFF $L$ is a finite union of languages $UV^\oo$ where
$U$ and $V$ are regular languages of nonempty finite words. This fact can be derived in essentially the same way as the corresponding statement for finite words showing that NFAs have the same expressive power as regular expressions.
Now, we can use our knowledge on finite words: if $L=\bigcup UV^\oo$, then $\sig(L)=\bigcup \sig(U)\sig(V)^\oo$. If $\sig$ is a regular \subst such that the empty word is not in any $\sig(x)$, then $\sig(L)\sse \Sig^\oo$ is $\oo$-regular because the class of regular subsets in $\Sig^*$ is closed under regular \subst{s}. More precisely, let $w= x_0x_1\cdots \in L\sse (\Sig\cup \cX)^\oo$ such that $\sig(w)\in L(B)$. Then $\sig(w)$ is the set of all $\oo$-words
which have a factorization $u= u_0u_1\cdots$ where for each $i\in \N$ we have
$u_i\in \sig(x_i)$.  Consider any $\oo$-word
$v=v_0v_1\cdots \in \Sig^\oo$ such that $v_i\in [u_i]$ for all $i\in \N$. Since $\mu$ it is strongly recognizable, as explained above, $\sig(w)\sse R$ implies $v\in R$. As a consequence, if $\sig(L)\#R$, then we have $\wh\sig(L) \# R$, too.

We now have all ingredients together to argue  as in paragraph above ``\textbf{Final~Steps}'' for finite words. However, to make the statements effective we use the fact that the family of $\oo$-regular languages is an effective Boolean algebra. It is in this part where (according to our approach) B\"uchi's result in~\cite{buc62} enters the scene: the complement of an $\oo$-regular language is effectively $\oo$-regular. See also the corresponding footnote (\ref{foot:1}) above where we speak about finite words and where we point to complementation.
\end{proof}
For more details about regular languages over infinite words (in the spirit of this section) we refer to the textbook~\cite{edam16}. \Ip the book presents
the equivalence between $\oo$-regular expressions and Büchi automata. It also shows how Büchi used a simple argument from Ramsey theory to derive his results in~\cite{buc62}.
Of course, other textbooks or survey papers than~\cite{edam16} do the same.

\subsection{About the existence of maximal solutions}\label{sec:maxsol}
The aim of this subsection is to
show  that the \emph{regularity} hypothesis on $R\sse T(\Sig)$ cannot be removed from our main results. This is done in \prref{ex:zorn}. The example uses the special case of infinite words, only.

Remember that $\TfinX(\Sig\cup\cX)$ (resp.~$\TfinH(\Sig\cup H)$) denote the set of trees  with only a finite number of occurrences of variables (resp.~holes).
Moreover, in \prref{sec:substi} we used induction on the maximal level of a variable and \prref{eq:siodef} (resp.~\prref{eq:soidef}) to define $\sio(s)$ (resp.~$\soi(s)$)
for $s\in \TfinX(\Sig\cup\cX)$
(with the additional exceptional case of a term without variable).
Such inductions can be viewed, as inductions on the following notion of \emph{norm} of a term
\begin{equation}\label{eq:sizes}
\|s\| = \sup \{ |u|+1 \mid u \in \pos(s) \wedge s(u) \in {\cX}\}.
\end{equation}
%(The supremum is taken in the ordered set $(\N,\leq))$; it is an integer, which is null exactly when $s$ has no occurrence of variable).
%%%%%%%%%%%%%%%%%%%
\begin{prop}\label{prop:Zorn}
Let  $L\sse \TfinX(\Sig\cup\cX)$ and $R\sse T(\Sig)$ be arbitrary subsets, $\sig:\cX\to 2^{T(\Sig\cup H)\sm H}$ be a \subst, and
$\#\in \os{{\sse}, {=}}$. Then the following holds.
\begin{itemize}
\item There exists a maximal \subst $\sig'$ such that $\sig\leq \sig'$ and $\sio'(L)\,\#R$.
\item If, in addition, $\sig(x)\sse \TfinH(\Sig\cup H)$, then
exists a maximal \subst $\sig'$ such that $\sig\leq \sig'$ and $\soi'(L)\,\#R$, too.
\end{itemize}
\end{prop}
\begin{proof}
For the proof let $e\in \os{\mio,\, \moi}$ one of the two possible extensions.
It is enough to see that $\leq$ is an \emph{inductive ordering} on the set of \subst{s} $\sig:\cX\to 2^{T(\Sig\cup H)\sm H}$ such that
$\sext(L)\,\#\, R$ because then, every solution is bounded from above by some maximal solution thanks to Zorn's lemma.
For that we have to consider nonempty totally ordered subsets
\begin{align*}
\{\sig^{(k)}:\cX\to 2^{T(\Sig\cup H)\sm H}\mid k \in K\}\sse
\{\sig:\cX\to 2^{T(\Sig\cup H)\sm H}\mid \Nat(L)\,\#\, R\}.
\end{align*}
Here, the index set $K$ is totally ordered satisfying
$\sig^{(k)}\leq \sig^{(\ell)} \iff k\leq \ell$.
For $x\in \cX$ we define $\sig'(x)=\bigcup \set{\sig^{(k)}}{k\in K}$.
This yields a \subst $\sig': \cX\to 2^{T(\Sig\cup H)\sm H}$ such that  $\sig^{(k)}\leq \sig'$.
We show by induction on $\Abs s$
that
if $s\in L$ and $t \in \sig'(s)$, then $t  \in \sig^{(k)}(s)$
for some $k \in K$.
For that we fix any index $k_0\in K$. Now, let $s=x(s_1\lds s_r)\in L$ and
$t=t_x[i_j\lsa t_{i}]\in \sio'(s)$ (resp.~$t=t_x[i_j\lsa t_{i_j}]\in \soi'(s)$) such that
$t_x\in \sig'(x)$ and $t_{i}\in \sio'(s_i)$ (resp.~$t_{i_j}\in \sio'(s_i)$).  For $\Abs s =0$ and both possibilities  $e\in \os{\mio,\, \moi}$,  we have
$\sext^{(k_0)}(s)= \sext'(s)=s$. We are done for $\Abs s = 0$. In the other case,  there are indices $k_i$ (resp.~$k_{ij}$) by induction such that
$t_{i}\in \sio^{(k_i)}(s_i)$ (resp.~$t_{i_j}\in \soi^{(k_{ij})}(s_i))$ for all $1\leq i \leq \rk(x)$. (This is true for $\rk(x)=0$.) Since $t_x\in \sig'(x)$, there is some index $k_x$ such that $t_x\in \sig^{(k_x)}(x)$.
Let $k_t=\max\set{k_x,k_i}{0\leq i \leq \rk(x)}$
(resp.~$k_t=\max\set{k_x,k_{ij}}{0\leq i \leq \rk(x) \wedge i_j\in \leaf{t_x}}$). The maximum exists for both $e=\mio$ and $e=\moi$.
We have $k_t\geq k_0$ and
$t \in\sext^{(k_t)}(s)$: the induction step is achieved.
Therefore,
$\sext'(s)\sse R$.

Finally, if $\sext^{(k)}(L)= R$ for all $k\in K$, then for each $t\in R$ there is some $s\in L$ such that $t\in \sext^{(k_t)}(s)$.
This implies
$ %\begin{align}\label{eq:Es}
R\sse \bigcup\{\sext^{(k)}(s)\mid s \in L \wedge k\in K\}
=\bigcup\{\sext'(s)\mid s \in L\}\sse R.
$ %\end{align}
Hence, $\bigcup\{\sext'(s)\mid s \in L\}=R$: every solution is bounded from above by some maximal solution.
\end{proof}

%%%%%%%%%%

\prref{ex:zorn} shows that the statement of \prref{prop:Zorn} might fail, when $L$ contains some tree with infinitely many occurrences of some variable. We give such an example in the case where $R$ is  not regular and $L$ is a set of  $\oo$-words. It fails for both extensions $\os{\mio,\moi}$ and
for both comparison relations ${\#}\in \os{{=},\, {\sse}}$.
%%%%%%%
\begin{exa}\label{ex:zorn}
Consider $\Sigma= \{a,b\}$ and $\cX =\{x\}$
where all three symbols have rank one. Hence, $H=\os 1$ and an infinite tree
over $\ScX$ encodes an infinite word in  $(\ScX)^\oo$ and vice versa.
We let  $L=L' \cup R$ where $L'=\os{(ax)^{\omega}}$ and $R =\set{u \in \Sigma^\omega}{\exists k \geq 1, b^k \text{ is no factor of }u}$. The singleton $L'$ is regular, but  neither $L$ nor $R$ is not regular. Note that we have $\sio(L)=\soi(L)$ for every
\subst $\sig(x) \sse T(\Sig \cup H)\sm H$. Clearly,
$\sio(L)= R\iff \sio(L)\sse R\iff \sio(L')\sse R$.
There are many \subst{s} $\sig$ such that $\sio(L')\sse R$. For example, let $\sig(x)= a(1)$, then $\sio(L')=a^\oo$. But there is no maximal \subst $\sig$ such that $\sio(L)\sse R$. Indeed,
given any \subst  $\sigma(x)\sse \Tfin(\os{a,b,1})\sm H$ such that $\sio(L')\sse R$, we can define $n= \max\{k \geq 0 \mid b^k \text{ is a factor of }\sigma(x)\}$. Then
$\sigma'(x) = \sigma(x) \cup \{b^{n+1}(1)\}$ is
strictly larger than $\sigma$, and it satisfies
$\sio'(L)\sse R$, too. For example, if $\sio(L)=a^\oo$, then  we obtain $\sigma'(x)= \os{a(1), b(1)}$ and $\sio'(L)= (ab)^\oo$.
\end{exa}
Comparing the two items  in \prref{prop:Zorn} we see that the ``outside-in''-version for $\moi$ requires an additional assumption: every
term $\sig(x)$  has finitely many holes, only. The next example shows that we cannot drop that hypothesis.
\begin{exa}\label{ex:zornig}
Consider $\Sigma= \{f,a,b\}$ and $\cX =\{x,y\}$ where
$\rk(f)=2$, $\rk(x)=\rk(y)=\rk(a)=1$, and $\rk(b)=0$. We are interested
in the finite term $x(y(b))$.
Let %$L$ be the singleton $L=\os{xy(b)}$ and
$\sig(x)=t$ where
$t=f(t,1)$ be infinite comb as in  \prref{fig:infcomb} and
$\sig(y) =a(b)$. The infinite tree $\sig(x)$ is regular and the accepted language of a deterministic top-down
tree automaton with a B\"uchi-acceptance condition as defined for example in~\cite{tho90handbook}. The singleton $\soi(x(y(b)))$ is shown in \prref{fig:ffa} on the left with $n=1$ and $\sig(y) =a(b)$. Similar as in \prref{ex:zorn} we let
\begin{align*}
R=\set{t \in T(\os{f,a,b})}{\exists k \geq 1, a^k \text{ is no factor of any path in }t}.
\end{align*}
Then $\soi(x(y(b)))\sse R$. However, there is no maximal
\subst $\sig'$ such that $\sig\leq\sig'$  and $\soi(x(y(b)))\sse R$.
Indeed, given any such $\sig'$ there is some $n\in \N$ such that $\sig''(y) = \sig'(y) \cup \os{a^{n}(b)}$ is strictly larger than $\sig'$ and $\sig''$ still satisfies $\soi''(x(y(b)))\sse R$. If we define $\sig''(y)=a^*(b)$, then $\soi''(x(y(b)))$ is regular, but it is not a subset of $R$.
\end{exa}

%%%%%%%%%%%%%%%%%%%%%%
\subsection{The outside-in extension \texorpdfstring{$\sig_{\mathrm{oi}}$}{sigma-oi} for  infinite trees}\label{sec:oiinf}
Although the definition of $\sig_{\mathrm{oi}}(s)$ for infinite trees is not used in the main body of the paper, we include a definition to have a reference for possible future research.
Actually, the definition of
$\soi(s)$ is technically simpler than the one of $\sio(s)$ because no choice functions are used.
Let us define $\soi(L)$ for a set of trees
which can be finite or infinite.
Without restriction, we content ourselves
to define $\soi(L)$ for $S\sse T(\cX)$ where $\Sig\cap \cX=\es$ and $\sig(x)\sse T(\Sig\cup [\rk(x)])\sm H$ for all $x\in \cX$.
The first step reduces the problem to define $\soi(L)$ for a set $L$  to the case where $L$ is a singleton simply be letting, as expected,  $\soi(L)=\bigcup\set{\soi(s)}{s\in L}$. Thus, it is enough to define $\soi(s)$ for an infinite tree $s\in T(\cX)$ because for a finite tree $s$ we employ the definition in \prref{eq:soidef}. The idea is to use a nondeterministic program
which transforms $s\in T(\cX)$ into a ``hybrid'' tree of $T(\cX\cup \Sig)$. During the process top-down more and more symbols appear in $\Sig$ and the subtrees (which are subtrees of $s$) appear at greater levels.
Each run of the program defines at most one output. The set of all outputs over all runs defines the set $\soi(s)$.
The program does not need to terminate, but if it runs forever, then, in the limit, it defines (nondeterministically) a tree $t\in T(\Sig)$.
The nondeterministic program is denoted as  ``$\soi^{nd}$''.
The input for the program is any tree $s\in T(\cX)$.
\medskip

{\textbf{begin procedure}}
Initialize a tree variable  $t:=s\in T(\Sig\cup \cX)$.\\
\textbf{while} $t\notin T(\Sig)$ \textbf{do}
%Perform the following while-loop  as long as $t\notin T(\Sig)$.
\begin{enumerate}

\item Choose in the breadth-first order on $\Pos(t)$ the first position $v$ such that $x=s(v)\in \cX$.
\item Denote the subtree $t|_{v}$ rooted at $v$ as
$t|_{v}=x(t_1\lds t_r)$. This implies $\rk(x)=r$.
\item Choose nondeterministically some $t_x\in \sig(x)$.\\ If this is not possible because $\sig(x)=\es$, then EXIT without any output.
\item Replace in $t$ the subtree $t|_{v}$ by $t_x[i_j\lsa t_i]$. Recall that this means to replace in $t_x$ every $i_j\in \leaf_i(t_x)$ by the same tree $t_i$ where $t_i$ is the $i$-th child of the root in $t|_{v}$. Moreover, if $\leaf_i(t_x)\neq \es$, then $1 \leq i\leq \rk(x)$.
\item []\textbf{endwhile} \\
\textbf{end procedure}
\end{enumerate}
%\noindent{\textbf{end procedure}}
%\medskip

\noindent
Do not confuse the procedure with an inside-out extension.
We visit positions in $t$ labeled by a variable one after another and later choices of trees in $\sig(t(v))$ are fully independent of each other.
If the program terminates without using the EXIT branch, then we return the final tree $t$ as the output $t=\sigma(\rho,s)$  of the specific nondeterministic run $\rho$.
If the program runs forever, then let $t_n$ be the value of $t$ after performing the $n$-th loop of this run $\rho$. Since we always have $t_x\in T(\Sig \cup H)\sm H$, an easy reflection shows that there exists a unique infinite tree $\sigma(\rho,s)=\lim_{n\to \infty}t_n$. Moreover, $\sigma(\rho,s)\in T(\Sig)$. We then define
$\soi(s) = \{ \sigma(\rho,s) \mid \rho \text{ is a run of $\soi^{nd}(s)$}\}$ and \prref{eq:soidef} holds.

\subsection{Quotient metrics}\label{sec:pseu}
This short subsection explains \prref{foot:pm}: our definition of a \emph{quotient metric} agrees with the standard one in topology.
Given a metric space $(M,d)$ and any equivalence relation
$\sim$ on $M$, there is a canonical definition of a quotient (pseudo-)metric $d_\sim$ on the quotient space $M/{\sim}$.
For $x\in M$ let $[x']=\set{x'\in M}{x\sim x'}$ denote its equivalence class. We associate to $(M,d)$ a complete weighted graph with vertex set $M/{\sim}$ and  weight
$g([x],[y])= \inf\set{d(x',y')}{x\sim x' \wedge y\sim y'}$. (So the weight might be $0$ for $[x]\neq [y]$.) Then we define
$d_\sim([x],[y])$ by the infimum over all weights of paths in the undirected graph connecting $[x]$ and $[y]$. The path can be arbitrarily long and still have weight $0$.
Clearly, $d_\sim$ is a pseudo metric (but not a metric, in general) satisfying
\[
0 \leq d_\sim([x],[y]) \leq g([x],[y]) \leq d(x,y).
\]
If for each $[x]$ there exists $x_0\in[x]$
such that for all $y$ we have
\[0< g([x],[y])= \inf\set{d(x_0,y')}{y\sim y'}\] then $g([x],[y])=d_\sim([x],[y])$ is a metric.
In our situation it is a metric: we have $(M,d)=(T(\OO_\bot),d)$ and $M/{\sim}$
results by identifying all trees $s \in (T(\OO_\bot),d)$ where some position in $\Pos(s)$ is labeled by $\bot$.
To see this, just recall our definition of $d(s,s')$ for trees in $T(\OO_\bot)$:
\begin{align*}
 d(s,s') =
\begin{cases}
 \text{$1$ \hspace{3.8cm} if either $s$ or $s'$ uses the symbol $\bot$ but not both,}\\
2^{-\inf\set{|u|\in \N}{u\in \N^*: \, s(u)\neq s'(u)}} \text{ otherwise.}
\end{cases}
%\ngam n (s) =  \gamma(0)[i_j\lsa \gamma_i(s_i,n-1)].
\end{align*}
\end{document}